\def\CCSFORMAT{0}
\def\draftversion{1}

\ifnum\CCSFORMAT=1
\documentclass[sigconf]{acmart}
\else
\documentclass[11pt,oneside]{article}
\usepackage[margin=1.0in]{geometry}

\fi

\usepackage{amsmath,amsthm,amsfonts,amscd,amssymb}
\usepackage{enumerate} 
\usepackage{physics}
\usepackage{enumerate}
\usepackage{fancyhdr}
\usepackage{hyperref}
\usepackage{graphicx}
\usepackage{float}
\usepackage{multirow}
\usepackage{makecell}
\usepackage{tabu}
\usepackage[capitalize]{cleveref}
\crefname{equation}{Expression}{Expressions}
\crefname{exampleattack}{Example Attack}{Example Attacks}
\crefname{motivateattack}{Attack Idea}{Attack Ideas}
\crefname{attacklibrary}{Library-Mediated Attack}{Library-Mediated Attacks}

\usepackage{mathtools}

\usepackage[
backend=biber,
style=numeric, % or numeric?
sorting=anyt,
maxbibnames=5
]{biblatex}
\hypersetup{colorlinks,
    linkcolor=blue,
    citecolor=blue,      
    urlcolor=blue
}

\usepackage{xurl}

\newtheorem{theorem}{Theorem}[section]
\newtheorem{corollary}[theorem]{Corollary}
\newtheorem{lemma}[theorem]{Lemma}
\newtheorem{property}[theorem]{Property}
\newtheorem{proposition}[theorem]{Proposition}
\newtheorem*{theorem-non}{Theorem}

\theoremstyle{definition}
\newtheorem{remark}[theorem]{Remark}
\newtheorem{definition}[theorem]{Definition}

\newtheorem{note}[theorem]{Note}

\newtheorem{exampleattack}[theorem]{Example Attack}
\newtheorem{motivateattack}[theorem]{Attack Idea}
\newtheorem{attacklibrary}[theorem]{Library-Mediated Attack}
\newtheorem{method}[theorem]{Method}

\newcommand{\R}{\mathbb{R}}
\newcommand{\LL}{\mathcal{L}}
\newcommand{\UU}{\mathcal{U}}

\newcommand{\bs}{\mathit{BS}}
\newcommand{\bsfp}{\mathit{BS}^*}
\newcommand{\rp}{\mathit{RP}}
\newcommand{\ulp}{\mathit{ULP}_{(\mantlen, \explen)}}
\newcommand{\nrf}{\mathit{NRF}}
\newcommand{\rtz}[1]{\mathit{RTZ}(#1)} % round toward zero
\newcommand{\br}{\mathit{BRound}} % banker's rounding
\newcommand{\Ex}{\mathbb{E}} % expectation

\newcommand{\round}{\mathit{round}}
\newcommand{\FloatToInt}{\mathit{Float2Int}}
\newcommand{\Noise}{\mathit{Noise}}
\newcommand{\nmax}{n_{\max}}

\usepackage{listings}
\usepackage{xcolor}
\usepackage{tikz} 

\definecolor{codegreen}{rgb}{0,0.6,0}
\definecolor{codegray}{rgb}{0.5,0.5,0.5}
\definecolor{codepurple}{rgb}{0.58,0,0.82}
\definecolor{backcolour}{rgb}{0.95,0.95,0.92}

\lstdefinestyle{mystyle}{
    backgroundcolor=\color{backcolour},   
    commentstyle=\color{codegreen},
    keywordstyle=\color{magenta},
    numberstyle=\tiny\color{codegray},
    stringstyle=\color{codepurple},
    basicstyle=\ttfamily\footnotesize,
    breakatwhitespace=false,         
    breaklines=true,                 
    captionpos=b,                    
    keepspaces=true,                 
    numbers=left,                    
    numbersep=5pt,                  
    showspaces=false,                
    showstringspaces=false,
    showtabs=false,                  
    tabsize=2
}

\newcommand{\len}{\mathrm{len}}
\newcommand{\Prob}{\mathrm{Pr}}

\newcommand{\pos}{\mathit{pos}}
\newcommand{\negs}{\mathit{neg}}

\ifnum\draftversion=1
\newcommand{\questionc}[1]{\textcolor{red}{\textbf{Question:} #1}} 
\newcommand{\discuss}[1]{\textcolor{red}{\textbf{Discuss:} #1}} 
\newcommand{\notec}[1]{\textcolor{blue}{\textbf{Note:} #1}}
\newcommand{\silvia}[1]{{ {\color{olive}{(silvia)~#1}}}}
\newcommand{\connor}[1]{{ {\color{teal}{(connor)~#1}}}}
\newcommand{\salil}[1]{{ {\color{purple}{(salil)~#1}}}}

\newcommand{\todo}[1]{{ {\color{red}{(TODO:)~#1}}}}

\newcommand{\draft}[1]{{ {\color{olive}{(draft)~#1}}}}
\newcommand{\ps}[1]{ {\color{orange}Proof status:}~#1}
\else
\newcommand{\questionc}[1]{} 
\newcommand{\discuss}[1]{} 
\newcommand{\notec}[1]{}
\newcommand{\silvia}[1]{}
\newcommand{\connor}[1]{}
\newcommand{\salil}[1]{}
\newcommand{\todo}[1]{}

\newcommand{\draft}[1]{}
\newcommand{\ps}[1]{}
\fi

\newcommand{\Vect}{\text{Vec}}

\newcommand{\dsym}{d_\mathit{Sym}}
\newcommand{\dco}{d_\mathit{CO}}
\newcommand{\dham}{d_\mathit{Ham}}
\newcommand{\did}{d_\mathit{ID}}
\newcommand{\dmod}{d_\mathit{Mod}}

\newcommand{\nsym}{\simeq_\mathit{Sym}}
\newcommand{\nco}{\simeq_\mathit{CO}}
\newcommand{\nham}{\simeq_\mathit{Ham}}
\newcommand{\nid}{\simeq_\mathit{ID}}

\newcommand{\nd}{\simeq_d}
\newcommand{\ssym}{\Delta_\mathit{Sym}}
\newcommand{\sco}{\Delta_\mathit{CO}}
\newcommand{\sham}{\Delta_\mathit{Ham}}
\newcommand{\sid}{\Delta_\mathit{ID}}

\newcommand{\sd}{\Delta_\mathit{d}}
\newcommand{\mantlen}{k}
\newcommand{\explen}{\ell}
\newcommand{\bround}{\text{BRound}}

\newcommand{\cM}{\mathcal{M}}
\newcommand{\unbdd}{\simeq_{\mathit{unbdd}}}
\newcommand{\bdd}{\simeq_{\mathit{bdd}}}
\newcommand{\simham}{\simeq_{\mathit{Ham}}}
\newcommand{\simCO}{\simeq_{\mathit{CO}}}
\newcommand{\simSym}{\simeq_{\mathit{Sym}}}
\newcommand{\simID}{\simeq_{\mathit{ID}}}
\newcommand{\eps}{\varepsilon}
\newcommand{\Lap}{\mathrm{Lap}}
\newcommand{\eqdef}{\mathbin{\stackrel{\rm def}{=}}}
\newcommand{\cY}{\mathcal{Y}}

 {\begin{list}{}%
         {\setlength{\leftmargin}{#1}}%
         \item[]%
 }
 {\end{list}}

\tikzset{every picture/.style={line width=0.75pt}} %set default line width to 0.75pt 

\newcommand\blfootnote[1]{%
  \begingroup
  \renewcommand\thefootnote{}\footnote{#1}%
  \addtocounter{footnote}{-1}%
  \endgroup
}

\AtBeginDocument{%
  \providecommand\BibTeX{{%
    \normalfont B\kern-0.5em{\scshape i\kern-0.25em b}\kern-0.8em\TeX}}}

\ifnum\CCSFORMAT=1
\setcopyright{acmcopyright}
\copyrightyear{2022}
\acmYear{2022}
\acmDOI{10.1145/3548606.3560708}

\acmConference[CCS'22]{the 2022 ACM SIGSAC Conference on Computer and Communications Security}{November 7--11, 2022}{Los Angeles, USA}
\acmPrice{15.00}
\acmISBN{978-1-4503-9450-5/22/11}
\fi

\begin{document}

\title{Widespread Underestimation of Sensitivity in Differentially Private Libraries and How to Fix It}
%\textsc{\small Embargoed Draft: Please Do Not Distribute}}

\ifnum\CCSFORMAT=1
%\IEEEoverridecommandlockouts
\author{S\'ilvia Casacuberta}
\affiliation{%
    \institution{Harvard University}
    \city{Cambridge}
    \state{MA}
    \country{USA}
}
\email{scasacubertapuig@college.harvard.edu}

\author{Michael Shoemate}
\affiliation{%
    \institution{Harvard University}
    \city{Cambridge}
    \state{MA}
    \country{USA}
}
\email{shoematem@seas.harvard.edu}

\author{Salil Vadhan}
\affiliation{%
    \institution{Harvard University}
    \city{Cambridge}
    \state{MA}
    \country{USA}
}
\email{salil_vadhan@harvard.edu}

\author{Connor Wagaman}
\affiliation{%
    \institution{Boston University}
    \city{Boston}
    \state{MA}
    \country{USA}
}
\email{wagaman@bu.edu}

\else

\author{
  S\'ilvia Casacuberta\thanks{Harvard University and OpenDP. Email:~\texttt{scasacubertapuig@college.harvard.edu}. Supported by the Harvard Program for Research in Science and Engineering (PRISE) and a grant from the Sloan Foundation.} \\
  \and
  Michael Shoemate\thanks{Harvard University and OpenDP. Email:~\texttt{shoematem@seas.harvard.edu}. Supported by a grant from the Sloan Foundation, a grant from the US Census Bureau, and gifts from Apple and Facebook.} \\
  \and 
  Salil Vadhan\thanks{Harvard University and OpenDP. Email:~\texttt{salil\char`_vadhan@harvard.edu}. Supported by a grant from the Sloan Foundation, gifts from Apple and Facebook, and a Simons Investigator Award.} \\
  \and
  Connor Wagaman\thanks{Boston University and OpenDP. Email:~\texttt{wagaman@bu.edu}. Supported by the Harvard College Research Program (HCRP). Much of this work was completed during his time at Harvard University.} \\
\date{
    November 10, 2022
}
}
\fi

\ifnum\CCSFORMAT=0
\maketitle

\blfootnote{
Any opinions, findings, and conclusions or recommendations expressed in this material are those of the authors and do not necessarily reflect the views of our funders.}

\fi

\begin{abstract}
We identify a new class of vulnerabilities in implementations of differential privacy. Specifically, they arise when computing basic statistics such as sums, thanks to discrepancies between the implemented arithmetic using finite data types (namely, ints or floats) and idealized arithmetic over the reals or integers. These discrepancies cause the sensitivity of the implemented statistics (i.e., how much one individual's data can affect the result) to be much larger than the sensitivity we expect. Consequently, essentially all differential privacy libraries fail to introduce enough noise to meet the requirements of differential privacy, and we show that this may be exploited in realistic attacks that can extract individual-level information from private query systems. In addition to presenting these vulnerabilities, we also provide a number of solutions, which modify or constrain the way in which the sum is implemented in order to recover the idealized or near-idealized bounds on sensitivity.
\end{abstract}

\ifnum\CCSFORMAT=1



\maketitle
\fi

\ifnum\CCSFORMAT=0
\newpage
{
  \hypersetup{linkcolor=black}
  \tableofcontents
}
\newpage
\fi

\section{Introduction}

Differential privacy (DP)~\cite{dmns16} has become the prevailing framework for protecting individual-level privacy when releasing statistics or training machine learning models on sensitive datasets.  It has been the subject of a rich academic literature across many areas of research, and has seen major deployments by the US Census Bureau~\cite{MachanavajjhalaKiAbGeVi08,PSEO,Abowd18}, Google~\cite{ErlingssonPiKo14,Mobility,Symptoms,Vaccination}, Apple~\cite{apple18}, Facebook/Meta~\cite{MessingDeHiKiMaMuNaPeStWi20,Movement}, Microsoft~\cite{Telemetry,Broadband}, LinkedIn~\cite{LaborMarket}, and OhmConnect.\footnote{\url{https://edp.recurve.com/.}}
To facilitate the adoption of differential privacy, a number of researchers and organizations have released open-source software tools for differential privacy, starting with McSherry's PINQ~\cite{McSherry10}, and now including libraries and systems from companies like Google~\cite{google-dp}, Uber~\cite{JohnsonNeSo18}, IBM~\cite{hbal19}, and Facebook/Meta~\cite{ysstpmn21}, and the open-source  projects
OpenMined\footnote{\url{https://github.com/OpenMined/PyDP.}} and OpenDP~\cite{ghv20}. This paper came out of our work on the OpenDP project. 

However, implementing differential privacy correctly is subtle and challenging. Previous works have identified and attempted to address vulnerabilities in implementations of differential privacy coming from side channels such as timing~\cite{hpn11,AreWeThere} and the failure to faithfully emulate the noise infusion mechanisms needed for privacy when using floating-point arithmetic instead of idealized real arithmetic~\cite{Mironov12,i20,AreWeThere}.

In this work, we identify a new and arguably more basic class of vulnerabilities in implementations of differential privacy.  Specifically, they arise when computing basic statistics such as sums, thanks to discrepancies between the implemented arithmetic using finite data types (namely, ints or floats) and idealized arithmetic over the reals or integers.  These discrepancies cause the {\em sensitivity} of the implemented statistics --- how much one individual's data can affect the result --- to be much larger than the sensitivity we expect.  Consequently, essentially all differential privacy libraries fail to introduce enough noise to meet the requirements of differential privacy, and we show that this may be exploited in realistic attacks that can extract individual-level information from differentially private query systems. In addition to presenting these vulnerabilities, we also provide a number of solutions, which modify or constrain the way in which the sum is implemented in order to recover the idealized or near-idealized bounds on sensitivity.

\subsection{Differential Privacy}

Informally, a randomized algorithm $\cM$ is {\em differentially private} if for every two datasets $u,u'$ that differ on one individual's data, the probability distributions $\cM(u)$ and $\cM(u')$ are close to each other.  Intuitively, this means that an adversary observing the output cannot learn much about any individual, since what the adversary sees is essentially the same as if that individual's data were not used.

To make the definition of differential privacy precise, we need to specify both what it means for two datasets $u$ and $u'$ to ``differ on one individual's data,'' and how we measure the closeness of probability distributions $\cM(u)$ and $\cM(u')$.  For the former, there are two common choices in the differential privacy literature.  In one choice, we allow $u'$ to differ from $u$ by adding or removing any one record; this is called {\em unbounded differential privacy} and thus we denote this relation $u\unbdd u'$.  Alternatively, we can allow $u'$ to differ from $u$ by {\em changing} any one record; this is called {\em bounded differential privacy} and thus we will write $u\bdd u'$. Note that if $u\bdd u'$, then $u$ and $u'$ necessarily have the same number of records; thus, this is an appropriate definition when the size $n$ of the dataset is known and public.  When $u\simeq u'$ for whichever relation we are using ($\bdd$ or $\unbdd$), we call $u$ and $u'$ {\em adjacent} with respect to $\simeq$.

For measuring the closeness of the probability distributions $\cM(u)$ and $\cM(u')$, we use the standard definition of $(\eps,\delta)$-differential privacy~\cite{DworkKeMcMiNa06}, requiring that:
\begin{equation}
    \forall T \qquad \Pr[\cM(u)\in T] \leq e^{\eps}\cdot \Pr[\cM(u')\in T]+\delta,
    \label{eqn:epsdel}
\end{equation}
where we quantify over all sets $T$ of possible outputs. 
There are now a variety of other choices, like Concentrated DP~\cite{DworkRo16,BunSt16}, R\'enyi DP~\cite{Mironov17}, and $f$-DP~\cite{fdp}; our results are equally relevant to these forms of DP, but we stick with the basic $(\eps,\delta)$ notion for simplicity.
If $\cM$ satisfies (\ref{eqn:epsdel}) for all $u,u'$ such that $u\simeq u'$, then we say that $\cM$ is {\em $(\eps,\delta)$-DP with respect to $\simeq$}.  If $\delta=0$, we say $\cM$ is {\em $\eps$-DP with respect to $\simeq$}.  Intuitively, $\eps$ measures {\em privacy loss} of the mechanism $\cM$, whereas $\delta$ bounds the probability of failing to limit privacy loss to $\eps$ (so we typically take $\delta$ to be cryptographically small).

The fundamental building block of most differentially private algorithms is noise addition calibrated to the {\em sensitivity} of a function $f$ that we wish to estimate.

\begin{definition}[Sensitivity]
Let $f$ be a real-valued function on datasets, and $\simeq$ a relation on datasets.  The {\em (global) sensitivity} of $f$ with respect to $\simeq$ is defined to be:
$$\Delta_\simeq f = \sup_{u\simeq u'} |f(u)-f(u')|.$$
\end{definition}

\begin{theorem}[Laplace Mechanism~\cite{dmns16}] \label{thm:Laplace}
For every function $f$ and relation $\simeq$ on datasets, the mechanism
$$\cM(u) = f(u) + \Lap\left(\frac{\Delta_\simeq f}{\eps}\right)$$
is $\eps$-DP, where $\Lap(s)$ denotes a draw from a Laplace random variable with scale parameter $s$.
\end{theorem}

There are a number of other choices for the noise distribution, leading to the Discrete Laplace (a.k.a., Geometric) Mechanism~\cite{GhoshRoSu12}, the Gaussian Mechanism~\cite{mironov2019r}, and the Discrete Gaussian Mechanism~\cite{cks20}, where the latter two achieve $(\eps,\delta)$-DP with $\delta>0$.  The key point for us is that in all cases, the scale or standard deviation of the noise is supposed to grow linearly with the sensitivity $\Delta_\simeq f$, so it is crucial that the sensitivity is calculated correctly.  If we incorrectly underestimate the sensitivity as $\Delta=(\Delta_\simeq f)/c$ for a large constant $c$, then we will only achieve $c\eps$-DP; that is, our privacy loss will be much larger than expected.  Given that it is common to use privacy loss parameters like $\eps=1$ or $\eps=0.5$, a factor of 5 or 10 increase in the privacy-loss parameter can have dramatic effects on the privacy protection (since $e^5>148$, allowing a huge difference between the probability distributions $\cM(u)$ and $\cM(u')$).

The most widely used function $f$ in DP noise addition mechanisms is the {\em Bounded Sum} function.  

\begin{definition}[Bounded Sum] 
For real numbers $L\leq U$, and a dataset $v$ consisting of elements of the interval $[L,U]$, we define the
{\em Bounded Sum} function on $v$ to be:
$$\bs_{L, U}(v) = \sum_{i=1}^{\len(v)} v_i.$$ 
The restriction of $\bs_{L,U}$ to datasets $v$ of length $n$ is denoted
$\bs_{L,U,n}$.
When we do not constrain the data to lie in $[L,U]$, we omit the subscripts.
\end{definition}
Typically, the bounds $L$ and $U$ are enforced on the dataset $v$ via a record-by-record clamping operation. To avoid the extra notation of the clamp operation, throughout we will work with datasets that are assumed to already lie within the bounds.

It is well-known and straightforward to calculate the sensitivity of Bounded Sum.
\begin{proposition} \label{prop:bs-sensitivity}
$\bs_{L,U}$ has sensitivity $\max\{|L|,|U|\}$ with respect to $\unbdd$, and $\bs_{L,U,n}$ has sensitivity $U-L$ with respect to $\bdd$.
\end{proposition}

Combining Proposition~\ref{prop:bs-sensitivity} and Theorem~\ref{thm:Laplace} (or analogs for other noise distributions), we obtain a differentially private algorithm for approximating bounded sums, which we will refer to as {\em Noisy Bounded Sum}.  This algorithm is pervasive throughout both the differential privacy literature and software.  Many more complex statistical analyses can be decomposed into bounded sums, and any machine learning algorithm that can be described in Kearns' Statistical Query model~\cite{Kearns98} can be made DP using Noisy Bounded Sum. Indeed, the versatility of noisy sums was the basis of the SuLQ privacy framework~\cite{BlumDwMcNi05}, which was a precursor to the formal definition of differential privacy.  Noisy Bounded Sum is also at the heart of differentially private deep learning~\cite{acgbmmtz16}, as each iteration of (stochastic) gradient descent amounts to approximating a sum (or average) of gradients.  For this reason, every software package for differential privacy that we are aware of supports computing noisy bounded sums via noise addition.

\subsection{Previous Research} \label{sec:prior}

Despite their simplicity, Noisy Bounded Sum and related differentially private algorithms are surprisingly difficult to implement in a way that maintains the desired privacy guarantees. 

In the early days of implementing differential privacy, several challenges were pointed out by Haeberlen, Pierce, and Narayan~\cite{hpn11}.  Their attacks concern not the Bounded Sum function itself, but dataset transformations that are applied before the Bounded Sum function.  In a typical usage, Bounded Sum is not directly applied to a dataset $u=[u_1,\ldots,u_n]$, but rather to the dataset $$q(u)\eqdef [q(u_1),q(u_2),\ldots,q(u_n)]$$ where $q$ is a ``microquery'' mapping records 
to the interval $[L,U]$. (We can incorporate the clamping into $q$.)  If $u\simeq u'$, then $q(u)\simeq q(u')$, so we if we apply Noisy Bounded Sum (or any other differentially private algorithm) to the transformed dataset, we should still satisfy differential privacy with respect to the original dataset $u$.
Haeberlen et al.'s attacks rely on a discrepancy between this mathematical abstraction and implementations.  In code, $q$ may not be a pure function, and may be able to access global state (allowing information to flow from the execution of $q(u_i)$ to the execution of $q(u_j)$ for $j>i$) or leak information to the analyst through side channels (such as timing). The authors of the main differentially private systems at the time, PINQ~\cite{McSherry10} and Airavat~\cite{rsksw10}, were aware of and noted the possibility of such attacks, but the implemented software prototypes did not fully protect against them. As discussed in \cite{McSherry10,rsksw10,hpn11}, remedies for attacks like these include using a domain-specific language for the microquery $q$ to ensure that it is a pure function and enforcing constant-time execution for $q(u)$.

At the other end of the DP pipeline (after the calculation of Bounded Sum), the seminal work of Mironov~\cite{Mironov12} demonstrated vulnerabilities coming from the noise addition step. Specifically, the Laplace distribution in Theorem~\ref{thm:Laplace} is a continuous distribution on the real numbers, but computers cannot manipulate arbitrary real numbers.  Typical implementations approximate real numbers using finite-precision floating-point numbers, and Mironov shows that these approximations can lead to complete failure of the differential privacy property.  To remedy this, Mironov proposed the Snapping Mechanism, which adds a sufficiently coarse rounding and clamping after the Laplace Mechanism to recover differential privacy with a slightly larger value of $\eps$.  Subsequent works~\cite{GazeauMiPa13,BalcerVa18,cks20} avoided floating-point arithmetic, advocating for and studying the use of exact finite-precision arithmetic (e.g., using big-integer data types) and using discrete noise distributions, such as the discrete Laplace distribution~\cite{GhoshRoSu12} and the discrete Gaussian distribution~\cite{cks20}.  However, a recent paper by Jin et al.~\cite{AreWeThere} shows that current implementations of the discrete distribution samplers are vulnerable to timing attacks (which leak information about the generated noise value and hence of the function value it was meant to obscure).  They, as well as \cite{hb21}, also show that the floating-point implementations of the continuous Gaussian Mechanism are vulnerable to similar attacks as those shown by Mironov~\cite{Mironov12} for the Laplace Mechanism.

Ilvento~\cite{i20} studies the effect of floating-point approximations on another important DP building block, the Exponential Mechanism.  On a dataset $u$, the exponential mechanism samples an outcome $y$ from a finite set $\cY$ of choices with probability $p_y(u)$ proportional to $\exp(\eps f_y(u)/ (2\max_y \Delta_{\simeq} f_y))$, where $f_y(u)$ is an arbitrary measure of the ``quality'' of outcome $y$ for dataset $u$.  She shows that the discrepancy between floating-point and real arithmetic can lead to incorrectly converting the quality scores $f_y(u)$ into the probabilities $p_y(u)$ and hence violate differential privacy.  To remedy this, Ilvento proposes an exact implementation of the Exponential Mechanism using finite-precision base 2 arithmetic.  Note that, like noise addition mechanisms, the Exponential Mechanism also is calibrated to the sensitivities $\Delta_{\simeq} f_y$ of the quality functions.  Here too, if we underestimate the sensitivity by a factor of $c$, our privacy loss can be greater than intended by a factor of $c$, even if we perfectly implement the sampling or noise generation step.

Indeed, Mironov's paper~\cite{Mironov12} also suggests that sensitivity calculations may fail to translate from idealized, real-number arithmetic to implemented, floating-point arithmetic.  He gives an example of two datasets $u\bdd u'$ of 64-bit floating-point numbers such that $|\bs(u)-\bs(u')|=1$ but 
$|\bsfp(u)-\bsfp(u')|=129$, where $\bsfp$ is the standard, iterative implementation of summation of floating-point numbers, illustrated in Figure~\ref{fig:iterated-pseudocode}.

\begin{figure}[H]
\begin{lstlisting}[language = Python,frame=single, escapechar=|]
def iterated_sum(u):
    the_sum = 0
    for element in u:
        the_sum += element
    return the_sum
\end{lstlisting}
    \caption{Iterated Summation.}
    \label{fig:iterated-pseudocode}
\end{figure}

However, Mironov's example does not immediately lead to an underestimation of sensitivity, because the datasets $u$ and $u'$ include values ranging from $L=-2^{-23}$ to $U=2^{30}$, so the idealized sensitivity suggested by Proposition~\ref{prop:bs-sensitivity} is $U-L > 2^{30} \gg 129$.\footnote{As pointed out in Mironov's paper, this example does demonstrate an underestimation of sensitivity by a factor of 129 if we define $u\simeq u'$ to mean that $u$ and $u'$ agree on all but one coordinate and they differ by at most 1 on that coordinate. This notion of dataset adjacency is common in the DP literature when $u$ and $u'$ represent datasets in {\em histogram} format, where $u_i$ is the number of individuals of type $i$ (rather than individual $i$'s data).  However, in this case, the $u_i$'s would be nonnegative integers (so this example would not be possible) and it would be strange to use a floating-point data type.}
Mironov's paper suggests that these potential sensitivity issues may be addressed by replacing Iterated Summation by a {\em Kahan Summation}~\cite{kahan1965pracniques}, a different way of summing floating-point numbers that accumulates rounding errors more slowly. Unfortunately, this section of Mironov's paper seems to have gone mostly unnoticed, and we are not aware of any prior work that has addressed the question of how to correctly bound and control sensitivity in finite-precision implementations of differential privacy. 

\subsection{Our Contributions}

In our work, we identify new privacy vulnerabilities arising from underestimation of the sensitivity of the Bounded Sum function when implemented using finite data types, including 32-bit and 64-bit integers and floats.  Specifically, implementations of Bounded Sum often have sensitivities much larger than the idealized sensitivity given by Proposition~\ref{prop:bs-sensitivity}, and consequently the privacy loss of DP mechanisms using Bounded Sum is much larger than specified.  Thus, our work covers vulnerabilities arising from the ``middle step'' of the DP pipeline --- aggregation --- sitting between the foci of previous work, which considered vulnerabilities in the preprocessing step before aggregation and the noise generation/sampling step after aggregation (see Figure~\ref{fig:pipeline}).

\begin{figure}[t]
\centering
    \includegraphics[width=13cm]{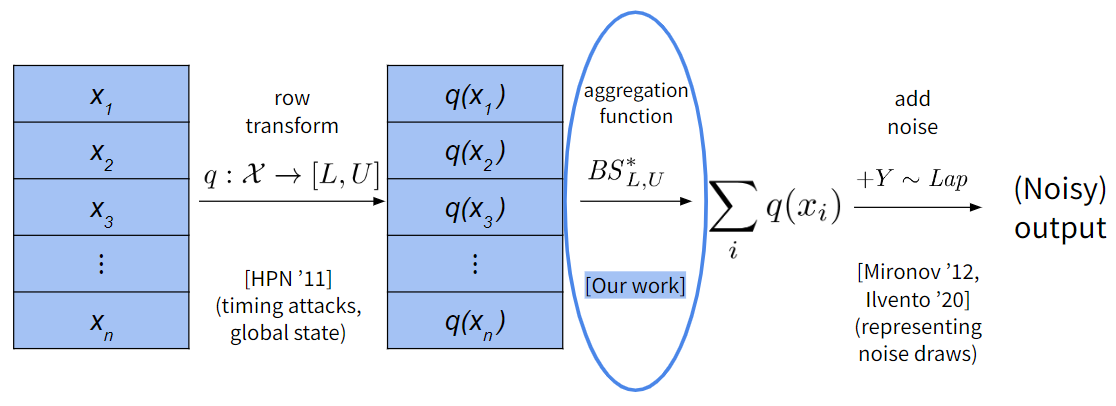}
    \caption{Illustration of the DP pipeline and the relationship between this paper and previous works uncovering vulnerabilities of DP implementations. \label{fig:pipeline}}
\end{figure}

In addition to describing the vulnerabilities that emerge due to sensitivity underestimation in essentially all libraries of DP functions and showing how to exploit them, we describe several solutions that recover the idealized or near-idealized bounds on sensitivity. Many of our solutions only require small modifications to current code.

We discovered these vulnerabilities and solutions as part of our work writing mathematical proofs to accompany components of the OpenDP library. We believe that our results demonstrate the value of a rigorous vetting process such as OpenDP's for differentially private libraries.

\subsection{Finite-precision Arithmetic in DP Libraries}

Before describing our results in more detail, we summarize how existing implementations of differential privacy address arithmetic issues (at the time of our work, prior to fixes implemented in response to our paper). 
All current DP implementations make use of the finite-precision data types (e.g., 32-bit or 64-bit ints or floats) to which our attacks apply.
Some of the libraries attempt to address the vulnerabilities uncovered by Mironov at the noise addition step \cite{Mironov12}. For example, Google's DP library \cite{wzldsg19} includes sampling algorithms for floating-point approximations to the Laplace and Gaussian distributions (based on Mironov's Snapping Mechanism) which they claim circumvent problems with na\"ive floating-point implementations.\footnote{\url{https://github.com/google/differential-privacy/blob/main/common_docs/Secure_Noise_Generation.pdf}.} 
IBM's \texttt{diffprivlib} \cite{hbal19} samples a floating-point approximation to the Laplace distribution using the method described in Holohan and Braghin \cite{hb21}. 
The OpenDP Library acknowledges\footnote{\url{https://docs.opendp.org/en/stable/user/measurement-constructors.html\#floating-point}} the floating point vulnerabilities discovered by Mironov, and users have access to floating-point mechanisms only if the \texttt{`contrib'} compilation flag is turned on to allow components that do not have verified proofs. 

Indeed, our work can be seen as following the call of the OpenDP Programming Framework paper~\cite{ghv20}, which says that
``any deviations from standard arithmetic (e.g., overflow of floating-point arithmetic) should
be explicitly modelled in the privacy analysis.''
(The paper~\cite{ghv20} goes further and advocates the use of fixed-point
and integer arithmetic as much as possible.  We now believe that 
abandoning floating-point arithmetic entirely may have a significant usability cost, so we consider both solutions that operate only on floating-point numbers and solutions that reduce floating-point summation to integer summation.)

In other DP software (e.g., Opacus~\cite{acgbmmtz16}, Chorus~\cite{jnhs20}, Airavat~\cite{rsksw10}, PINQ~\cite{McSherry10}), we did not find any mention of potential issues or solutions for floating-point computations.
 
All of the libraries we studied scale noise according to the idealized sensitivity of the Bounded Sum function (Proposition~\ref{prop:bs-sensitivity}):
\begin{itemize}
    \item Unbounded DP (\cref{thrm:reals-sens}, Part 1):
    Google's sum function,
    SmartNoise's sum function, Opacus,
    and Airavat.
    
    \item Bounded DP (\cref{thrm:reals-sens}, Part 2):
    IBM \texttt{diffprivlib}'s sum and mean,
    Google's mean, Chorus,
    and SmartNoise's sized sum function.
    \end{itemize}

The exact links for these functions can be found in Section~\ref{sec:ideal-sens-libraries}. All of these libraries underestimate the sensitivity of the implemented Bounded Sum function, for both integer and floating-point data types, and thus are vulnerable to our attacks.  We remark that carrying out our attacks in practice (to extract sensitive individual information from real-life datasets) does seem to require an adversary to carefully choose a microquery/row-transform $q$ (see Figure~\ref{fig:pipeline}), so these vulnerabilities are more of an immediate threat when the DP software is used as part of an interactive query system rather than for noninteractive data releases.  However, the fact that the DP guarantee fails raises the possibility of other attacks, which may not require interactive queries; this possibility can be avoided by implementing one of our solutions to recover a correct proof of differential privacy.

\iffalse
\subsection{Organization}

In Section~\ref{sec:notation}, we present the necessary notation for the paper.
% The rest of the preliminaries are deferred to Section~\ref{sec:preliminaries} in the appendix.
We present an overview of the attacks that yield the sensitivity lower bounds for both bounded and unbounded DP in Section~\ref{sec:lowerbounds}. We identify four different types of vulnerabilities: overflow, rounding, repeated rounding, and reordering.
% The formal theorem statements, proofs, and further examples are deferred to Section~\ref{sec:bounded-sum-computers} in the appendix.
In Section~\ref{sec:attacks} we show how these attacks can be carried out on the main existing DP libraries.
% and Section~\ref{sec:concreteattack} in the appendix contains the full details of these implementations.
Lastly, in Section~\ref{sec:sols2} we summarize the different solutions that we propose to fix these vulnerabilities: dataset adjacency relations, random permutations, checking or bounding parameters, truncated summation, split summation, sensitivity from accuracy, shifting bounds, and reducing floats to ints.
% The formal descriptions, theorem statements, and proofs of these solutions are deferred to Section~\ref{sec:sols} in the appendix.
In Section~\ref{sec:roadmap} we propose a roadmap for DP libraries with recommendations on how to best prioritize and implement our solutions.
\fi

\subsection{Basic Notation}\label{sec:notation}

Throughout, we will write $T$ for a finite numerical data type. We think of $T\subseteq \R \cup \{\pm \infty\}$, but adding two elements $a,b$ of $T$ can yield a number outside of $T$, so the {\em overflow mode} and/or {\em rounding mode} of $T$ determine the result of the computation $a+b$.  For $L,U\in T$ with $L\leq U$, we write $T_{[L,U]} = \{ x\in T : L\leq x\leq U\}$.
We will consider various implementations of the Bounded Sum function $\bsfp_{L,U}$ and $\bsfp_{L,U,n}$ on datasets consisting of elements of $T_{[L,U]}$.  Except when otherwise stated, $\bsfp$ will use the standard Iterated Summation method from Figure~\ref{fig:iterated-pseudocode}.  In such a case, the functions $\bsfp_{L,U}$ and $\bsfp_{L,U,n}$ and their sensitivities are completely determined by the data type $T$ and the choice of overflow and/or rounding modes.

The data types $T$ we will consider in the paper are the following:
\begin{itemize}
    \item $k$-bit unsigned integers, whose elements are $\{0,1,\ldots,2^k-1\}$.  Standard choices are $k=32$ and $k=64$.
    \item $k$-bit signed integers, whose elements are $\{-2^{k-1},-2^{k-1}+1,\ldots,-1,0,1,\ldots,2^{k-1}-1\}$. Standard choices are $k=32$ and $k=64$.
    \item $(k,\ell)$-bit (normal) floats, which are represented in binary scientific notation as $(-1)^s \cdot (1.M) \cdot 2^E$ with a $k$-bit mantissa $M$ and an exponent $E\in [-(2^{-\ell-1}-2),2^{\ell-1}-1]$.\footnote{Note that there are only $2^\ell-2$ choices for $E$.  The remaining choices are used to represent {\em subnormal} floats, as well as $\pm \infty$, and NaN. For simplicity, we only consider normal floats here in the introduction; the full set of $(k,\ell)$-bit floats is considered in Section~\ref{sec:bounded-sum-computers}.} In 32-bit floats (``singles''), we have $k=23$ and $\ell=8$; and in 64-bit floats (``doubles'') we have $k=52$ and $\ell=11$.  In machine learning applications, it is common to use even lower-precision floating-point numbers for efficiency, such as $k=7$ and $\ell=8$ in Google's bfloat16~\cite{wk19}. 
\end{itemize}

We find that for these data types, the implemented sensitivity of Bounded Sum can be much larger than the idealized sensitivity for several reasons, described in the section below.

\subsection{Overview of Sensitivity Lower Bounds}
\label{sec:lowerbounds}

\paragraph{Overflow (Section~\ref{sec:mod-attack}).
} The default way of dealing with overflow in $k$-bit integers $T$ (signed or unsigned) is {\em wraparound}, i.e., the result is computed modulo $2^k$.
It is immediately apparent how this phenomenon can lead to large sensitivity.  If the idealized sum on one dataset $u$ equals the largest element of $T$, call it $\max(T)$, and on an adjacent dataset $u'$ equals $\max(T)+1$, then modular arithmetic will yield results that differ by $2^k-1$.  If our parameter settings allow for us to construct two such datasets, then the implemented sensitivity of Bounded Sum will be $2^k-1$, a completely useless bound because every two numbers of type
$T$ differ by at most $2^k-1$.
This is formalized in Theorem~\ref{thrm:mod-attack} and Example Attack~\ref{ea:modular-sum}. 

In the case of bounded DP on datasets of size $n$, we can construct two such datasets $u$ and $u'$ if $n\cdot (U-L)\geq 2^k$. As one example of such a pair, consider the following datasets $u$ and $u'$ of unsigned ints $T$, where we have $L\leq 0$, $U > 0$, and set $n = \left\lceil \max(T)/U \right\rceil + 1$. Let $M = \max(T)- (n-2) \cdot U$. Then, set 
\[
    u = [U_1, \ldots, U_{n-2}, M, 0], \quad \quad u' = [U_1, \ldots, U_{n-2}, M, 1],
\]
where $U_i = U$ for all $i$.

Then, $\bs_{L,U}(u) = \max(T)$, and $\bs_{L,U}(u') = \max(T)+1$, since $M = \max(T)- (n-2) \cdot U$. But because we are using modular arithmetic, we need to evaluate both sums modulo $2^k$, in which case $\bsfp_{L,U}(u) = \max(T)$, but $\bsfp_{L,U}(u') = (\max(T) + 1) \mod 2^k = \min(T)$. Hence, $|\bsfp_{L,U}(u)-\bsfp_{L,U}(u')|= \max(T) - \min(T) = 2^k-1$. However, because $u \bdd u'$, the idealized sensitivity is $U-L$.

When working with the type $T$ of 64-bit unsigned ints, by setting $L = 0, U= 2^{47}$, and $n = \left\lceil \max(T) / U \right\rceil + 1 = 2^{17}+1$, we get a difference in sums of $2^{64}-1$, which is more than a factor of $2^{16}$ larger than these idealized sensitivities. This means, then, that a DP mechanism that claims to offer $\varepsilon$-DP here but calibrates its random distribution to the idealized sensitivity would instead offer no better than $2^{16} \varepsilon$-DP.

In practice, while $n$ may be modest (e.g., $n=2^{15}$) and much smaller than $2^k$, the bounds $U$ and $L$ are typically user-specified parameters. \color{black} More generally, for example, we can set $L=0$ and $U=\lceil 2^k/n\rceil$, and we get a sensitivity that is roughly a factor of $n$ larger than the idealized sensitivity $U-L$. 

In the case of unbounded DP, $n$ is unconstrained, so we always get an implemented sensitivity of $2^k-1$, provided that $U>L$.  Specifically, there are datasets $u$ and $u'$ of size at most $n=\lceil 2^k/U\rceil$ such that $u\unbdd u'$ and $|\bsfp_{L,U}(u)-\bsfp_{L,U}(u')|=2^k-1$.  Again, we get a blow-up in sensitivity of roughly a factor of $n$. For example, datasets $u  = [U_1,\ldots, U_{n-2}, M]$ and $u' = [U_1,\ldots, U_{n-2}, M, 1]$ are adjacent with respect to $\unbdd$. However, we again obtain that $|\bsfp_{L,U}(u)-\bsfp_{L,U}(u')|= \max(T) - \min(T) = 2^k-1$, which is much greater than the idealized sensitivity $\max\{|L|, |U| \}$.

\color{black}

Overflow also arises with floating-point arithmetic (leading to results that are $\pm \infty$), but the semantics of arithmetic with $\pm \infty$ are not clearly defined in the IEEE standard, and we consider it better to avoid this possibility entirely (as discussed in our solutions in \cref{sec:sols2}).

\paragraph{Rounding (\cref{sec:round-attack-floats}).
} 
When adding two floating-point numbers, the result may not be exactly representable as a float, but may lie strictly between two adjacent floats. Thus, the actual result is determined by the {\em rounding mode}. The standard, called {\em banker's rounding}, rounds the result to the nearest float, breaking ties by rounding to the float whose mantissa has least-significant bit 0. By inspection, when we round a real number $z$ to a $(k,\ell)$-bit float $\bround(z)$, the {\em relative} effect is minimal.  Specifically, $|\bround(z)-z| \leq |z|/2^{k+1}$.

This may have led to an incorrect impression that floating-point arithmetic is ``close enough'' to idealized real arithmetic to not cause a significant increase in sensitivity.
Unfortunately, we find that this is not the case.

In the case of bounded DP, even a single rounding can cause a substantial blow-up in sensitivity, as we illustrate with the following example (and, in more detail, in the proof of Theorem~\ref{prop:rounding-ham-floats}).
Let $Q$ be a float that is a power of 2, and let $n-1$ be a power of 2 between 2 and $2^{(k+1)/2}$.  Then, we take
\begin{eqnarray*}
L &=& \left(1+\frac{n-1}{2^{k+1}}\right)\cdot Q,\\
U &=& \left(1+\frac{n+1}{2^{k+1}}\right)\cdot Q.
\end{eqnarray*}
Thus, $L$ and $U$ are both very close to $Q$, but their difference is only $U-L=Q/2^k$.
Our datasets $u$ and $u'$ will both begin with $n-1$ copies of $L$.
The iterated sum of these first $n-1$ elements experiences no rounding, because all intermediate sums can be represented exactly as $(k, \ell)$-bit floats. This is because all these intermediate sums are integer multiples of $(n-1)Q/2^{k+1}$, and are each at most $(n-1)Q + \frac{(n-1)^2}{2^{k+1}}$, where $(n-1)$ is itself a positive power of 2.
The resulting sum of these $(n-1)$ terms is a floating-point number with exponent $\log_2 ((n-1)Q)$ (since $(n-1)L< 2(n-1)Q$). Thus, the space between adjacent floating-point numbers after this iterated sum is $(n-1)Q/2^k$.  Hence, when we add one more copy of $L$ to obtain $\bsfp(u)$, the result will lie exactly halfway between two adjacent floats (since $L$ is an odd multiple of $(n-1)Q/2^{k+1}$), and by banker's rounding, will get rounded down.  On the other hand, when we add a copy of $U$ to calculate $\bsfp(u')$, $\bs(u')$ will lie between the same adjacent floats as $\bs(u)$, but, by banker's rounding, the result will be rounded up. The effect of these two roundings gives us:
$$\bsfp(u')-\bsfp(u) = 2\frac{n-1}{2^{k+1}}\cdot Q = (n-1)\cdot (U-L).$$ 

In contrast, the idealized sensitivity with respect to $\bdd$ is $(U-L)$, so the sensitivity blows up by a factor of $(n-1)$, which is dramatic even for very small datasets. 
This is formalized in Theorem~\ref{prop:rounding-ham-floats} and Example Attack~\ref{ea:ham-round-once}.
We also show that this construction applies to Kahan summation (contrary to Mironov's hope that it would salvage the sensitivity) and pairwise summation (another common method --- the default in \texttt{numpy} --- where the values are added in a binary tree).
This is explained in Remarks \ref{rem:resist-kahan-thesis} and \ref{rem:resist-pair}, following the attack described in Section~\ref{sec:round-attack-floats}. 
%Intuitively, the idea is that the counterexample described in Section~\ref{sec:round-attack-floats} does not rely on accumulated error, but on the error inherent to rounding the sum to a $k$-bit float. \color{black}
Intuitively, this counterexample applies to Kahan summation and pairwise summation because it does not rely on accumulated error (against which Kahan and pairwise summation do help), but on the error inherent to rounding the sum to a $k$-bit float.

\paragraph{Repeated Rounding (Sections \ref{sec:accum-round-floats} and \ref{sec:accum-round-floats-ii}).
}
The above example does not give anything interesting for unbounded DP, since the idealized sensitivity is $\max\{|L|, |U|\}$, and 
$\bsfp(u')-\bsfp(u)\leq Q < U$. However, we can obtain a sensitivity blow-up by exploiting the accumulated effect of {\em repeated rounding}.  Consider a dataset whose first $(n/2)$ elements are $U$.  After that, each rounding can have the effect of increasing the sum by $\Theta(nU/2^k)$, for a total rounding error of $\Theta(n^2 U)/2^k$.  We show how to exploit this phenomenon to construct two datasets $u\unbdd u'$ which differ on their middle element such that 
$$\left|\bsfp_{L,U}(u)-\bsfp_{L,U}(u')\right| \geq
\left(1+\Omega\left(\frac{n^2}{2^{k}}\right)\right)\cdot U.$$
Thus, the sensitivity blows up by a factor of $\Theta(n^2/2^k)$.  This allows us to exhibit a sensitivity blow-up with datasets of size $n=\Theta(2^{k/2})$, which are plausible even for 64-bit floats (where $k=52$) and quite easy to obtain for 32-bit and lower-precision floats.

As formalized in Theorem~\ref{thrm:quadratic-round-attack} and Example Attack~\ref{ea:rr2}, this implemented sensitivity can be obtained with the following two datasets.
%This implemented sensitivity can be obtained with the following two datasets.
Let $U$ be a power of 2, $m$ an integer power of 2, $n = 2m$, and $$L  = -\left(\frac{U\cdot m}{2^\mantlen}\right)\cdot \left(\frac{1}{2} - \frac{1}{2^\mantlen}\right).$$ Additionally, let 
\[
    u = [U_1, \ldots, U_m, x_1, L_2, x_3, L_4, \ldots, x_{m-1}, L_m],
\]
\begin{equation}
\label{eqn:u_prime}
    u' = [U_1,\ldots, U_{m-1}, x_1, L_2, x_3, L_4,\ldots, x_{m-1}, L_m],
\end{equation}
where for all $i$, we have $U_i = U$, $L_i = L$, and
\[
    x_i = x = \left(\frac{U\cdot m}{2^\mantlen}\right)\cdot \left(\frac{1}{2} + \frac{1}{2^\mantlen}\right).
\]
Note that $u$ and $u'$ are equivalent with the exception that $u'$ contains $m-1$ copies of $U$ rather than $m$ copies of $U$, so $u \unbdd u'$.
% In the proof of Theorem~\ref{thrm:quadratic-round-attack} we show that
Similarly to the example for one rounding, we show that 
all intermediate sums computed in the calculation of $\bsfp_{L,U}[U_1,\ldots, U_m]$ can be represented exactly as $(\mantlen,\explen)$-bit floats. The ability to represent these intermediate sums exactly implies that $$\bsfp_{L,U}[U_1,\ldots,U_{m-1}] = \bs_{L,U}[U_1,\ldots,U_{m-1}] = (m-1)\cdot U$$ and $$\bsfp_{L,U}[U_1,\ldots,U_{m}] = \bs_{L,U}[U_1,\ldots,U_{m}] = m\cdot U.$$  

However, $\bs_{L,U}[U_1,\ldots,U_{m},x]=m\cdot U + x$ is not exactly representable as a $(\mantlen,\explen)$-bit float and will be rounded up to $m \cdot U + mU/2^k$. Continuing with the computation of $\bsfp_{L,U}(u)$, when we add $L$, the resulting sum is not large enough to escape rounding down by banker's rounding. The same reasoning applies to all subsequent additions of $x$ and $L$, where adding $x$ has the effect of adding $mU/2^k$ and adding $L$ has the effect of adding 0. This yields a total sum of
\[
    \bsfp_{L,U}(u) = m\cdot U + \frac{m^2 \cdot U}{2^{\mantlen+1}}.
\]
In the case of $u'$, we have one less $U$ term, so we consider the sum  $\bs_{L,U}[U_1,\ldots,U_{m-1},x] = (m-1) \cdot U + x$. This is again not a representable $(\mantlen,\explen)$-bit float, and so $\bsfp_{L,U}[U_1,\ldots,U_{m-1},x] = (m-1)\cdot U+ (m-1)U/2^k$. When we add $L$, the closest representable float to the sum is $(m-1) \cdot U$, and so by banker's rounding the sum gets rounded down to this value. Then, for $u'$, adding $x$ and $L$ in succession has the effect of adding $0$. In total, this yields
\[
    \bsfp_{L,U}(u') = (m-1) \cdot U.
\]  
Therefore, 
\[
    |\bsfp_{L,U}(u) - \bsfp_{L,U}(u')| = \frac{m^2\cdot U}{2^{\mantlen + 1}} + U = \frac{n^2\cdot U}{2^{\mantlen + 3}} + U,
\]
since $n = 2m$. This is a factor $n^2/2^{k+3} +1 $ larger than the idealized sensitivity $\max\{|L|, |U|\} = U$.

By setting $k = 52$ (as is the case for 64-bit floats), $L= - 2^{-23}\cdot (\frac{1}{2} + 2^{-52})$, $U=1$, and $n = 2^{30}$, we get a difference in sums of $2^5 + U = 33$, which is a factor 33 larger than the idealized sensitivity $\max\{|L|, |U|\} = 1$.

\paragraph{Reordering (Sections \ref{sec:non-assoc-ints} and \ref{sec:attack-floats-assoc}).
}  
Finally, we exhibit potential vulnerabilities based on ambiguity about whether datasets are ordered or unordered.  In the DP theory literature, datasets are typically considered unordered (i.e., multisets) when using unbounded DP. (Indeed, adjacency is often described by requiring that that the {\em symmetric difference} between the multisets $u$ and $u'$ has size at most 1, or by requiring that the histograms of $u$ and $u'$ have $\ell_1$ distance at most 1.)  On the other hand, when using bounded DP, it is common to denote datasets as ordered $n$-tuples.  (Indeed, adjacency is often described by requiring that the {\em Hamming distance} between the $n$-tuples is at most 1.)
Most implementations of DP are not explicit about these choices, but the wording in the documentation suggests the same conventions (e.g., see Section 1.3 in \cite{wzldsg19}).

In theory, the distinction does not matter much, because most of the functions we compute (such as Bounded Sum) are symmetric functions and do not depend on the ordering.

However, when implementing these functions using finite data types, the ordering can matter a great deal. Indeed, we show that there are datasets $u$ and $u'$ where $u'$ is a permutation of $u$ (so define exactly the same multisets) but 
$$\left|\bsfp_{L,U}(u)-\bsfp_{L,U}(u')\right| \geq
\left(1+\Omega\left(\frac{n^2}{2^{k}}\right)\right)\cdot U.$$
That is, we can obtain the same kind of blow-ups in sensitivity that we obtained due to rounding by instead just reordering the dataset.  (Our rounding attacks preserved order, i.e., we obtained $u'$ by either changing one element of the ordered tuple $u$, or by inserting/deleting one element of $u$ without changing the other elements.)

% As formalized in Theorem~\ref{prop:non-assoc-floats} and Example Attacks \ref{ea:32-nonassoc} and \ref{ea:64-nonassoc}

For the case $n = 3 \cdot 2^k$, this implemented sensitivity lower bound can be demonstrated with the following two datasets. Let $L = 2^j$ for some integer $j$, let $U = 2^{j+1}$, and let
\[
    u = [L_1, \ldots, L_{2^{(k+1)}}, U_1, \ldots, U_{2^k}],
\]
\[
    u' = [U_1, \ldots, U_{2^k}, L_1, \ldots, L_{2^{(k+1)}}].
\]
Note that $u'$ is a permutation of $u$.
In Theorem~\ref{prop:non-assoc-floats} we show that
the first $2^{k+1}$ terms of $u$ can be added exactly; i.e., $$\bsfp_{L,U}[L_1, \ldots, L_{2^{(k+1)}}] = 2^{k+1} \cdot L.$$ The next $2^k$ terms of $u$ can also be added exactly, and so
\[
    \bsfp_{L,U}(u) = 2^{k+1}\cdot L + 2^k \cdot U.
\]
Likewise, the first $2^k$ terms of $u'$ can be added exactly, and hence $\bsfp_{L,U}(u') = 2^k \cdot U$. However,
% as shown in Theorem~\ref{prop:non-assoc-floats},
the addition of the next $2^{k+1}$ $L$ terms to $u'$ yield intermediate sums which are \textit{not} representable exactly as $(k, \ell)$-floats. For every $L$ term that we add, the sum falls exactly in the middle of two adjacent floats, and since the corresponding mantissa ends with an even bit, by the definition of banker's rounding the sum will round down to $2^k \cdot U$ at each step. Therefore, we get a final sum of
\[
    \bsfp_{L,U}(u') = 2^k \cdot U.
\]
Therefore,
\[
    |\bsfp_{L,U}(u) - \bsfp_{L,U}(u')| = 2^{k+1}\cdot L = 2^k \cdot U.
\]
Under $\unbdd$, the idealized sensitivity is $\max\{|L|, |U|\} = U$, and so the implemented sensitivity is a factor $2^k$ larger than the idealized sensitivity. Under $\bdd$, the idealized sensitivity is $U-L$, and so the implemented sensitivity is a factor $2^k \cdot U/(U-L)$ larger than the idealized sensitivity.

As a concrete example, if we set $j = 0$ (so $L = 1$ and $U = 2$), then for 32-bit floats (i.e., $k=23$) we get datasets of length $3\cdot 2^{23}$ and a difference in sums of $2^{23} \cdot U = 2^{24}$, which is a factor $2^{23}$ larger than the idealized sensitivity $\max\{|L|, |U|\} = 2$ under $\unbdd$ and a factor $2^{24}$ larger than the idealized sensitivity $U-L = 1$ under $\bdd$.

One reason this issue may have been missed before is that floating-point arithmetic is {\em commutative}, e.g., $\bround(a+b) = \bround(b+a)$, so it may seem like order does not matter.  However, {\em associativity} is required for the sum to be invariant under arbitrary permutations, which fails for floating-point arithmetic due to rounding.

Modular integer arithmetic is associative, so this issue does not arise for $k$-bit (unsigned or signed) integers with wraparound.  However, if instead we use {\em saturation} arithmetic, where values get clamped to the range $[\min(T),\max(T)]$ (which addresses the aforementioned sensitivity problems with wraparound), then addition of {\em signed} integers is no longer associative.  Indeed, if $n\cdot \min\{U,-L\} \geq 2^{k+1}$, 
we exhibit two datasets $u,u'$ of length $n$ such that $u'$ is a permutation of $u$ and
$$\left|\bsfp_{L,U,n}(u)-\bsfp_{L,U,n}(u')\right| \geq 2^k-1.$$
That is, we get no improvement over the (trivial) sensitivity we had with modular arithmetic.

As formalized in Theorem~\ref{thrm:sens-ints-sat-unknown} and Example Attack \ref{ea:saturation}, this implemented sensitivity lower bound can be illustrated
with the following two datasets. Set $\alpha = \lceil \frac{\max(T)-\min(T)}{|L|} \rceil$, $\beta = \lceil \frac{\max(T)-\min(T)}{U} \rceil$, and let
\[
    u = [L_1, \ldots, L_{\alpha}, U_1, \ldots, U_{\beta}],
\]
\[
    u' = [U_1, \ldots, U_{\beta}, L_1, \ldots, L_{\alpha}],
\]
where $U>0$ and $L<0$. First, by definition of $\alpha$, we see that $\bsfp_{L, U}[L_1, \ldots, L_{\alpha}] = \min(T)$, given that saturation arithmetic clamps the negative values at $\min(T)$. When we add the next $\beta$ ``$U$'' terms, by definition of $\beta$, the intermediate sums go all the way up to $\max(T)$, and then saturation arithmetic clamps the sum at $\max(T)$, since all of the $U$ terms are positive. Therefore, 
\[
    \bsfp_{L, U}(u) = \max(T).
\]
Likewise, for $u'$, the first $\beta$ ``$U$'' terms result in an intermediate sum of $\max(T)$, and then the sum remains clamped there. The next $\alpha$ ``$L$'' terms then result in an intermediate sum of $\min(T)$, where it remains clamped. Therefore,
\[
    \bsfp_{L, U}(u') = \min(T).
\]
The result is a difference in sums of
\[
    \left|\bsfp_{L,U,n}(u)-\bsfp_{L,U,n}(u')\right| = \max(T) - \min(T) = 2^k-1,
\]
and since $u'$ is a permutation of $u$, we have $u \simeq u'$. Under $\unbdd$ the idealized sensitivity is $\max\{|L|, |U|\}$, and under $\bdd$, the idealized sensitivity is $U-L$, so this implemented sensitivity can be much greater than the idealized sensitivities.

As a concrete example, if we set $L = -2^{14}$, $U = 2^{15}$, and $k=32$, then the above construction gives us two datasets of size $n = \lceil \frac{\max(T)-\min(T)}{|L|} \rceil + \lceil \frac{\max(T)-\min(T)}{U} \rceil$ where the difference in sums is $2^k-1 = 2^{32}-1$.
This is more than a factor $2^{16}$ larger than the idealized sensitivity $\max\{|L|, |U|\} = 2^{15}$ under $\unbdd$ and the idealized sensitivity $U-L = 2^{15}+2^{14}$ under $\bdd$.

% \color{blue}

\def\arraystretch{1.0}

\begin{table*}
{\normalsize

\begin{center}

\begin{tabular}{|l|l|c|c|}\hline

\multirow{2}{*} & \hspace{0.55cm} \textbf{Adjacency relation} & \textbf{Implemented} & \textbf{Conditions} \\ & & \raisebox{0.3cm}{\textbf{sensitivity}} & \\[-0.3cm] \cline{1-4}

\multirow{2}{*}{Idealized
[Thm.\,\ref{thrm:reals-sens}]
} & Bounded DP & $U-L$ & \\ \cline{2-4}

& Unbounded DP & $U$ & \\ \cline{1-4}

\multicolumn{1}{|l|}{\raisebox{-0.1cm}{Modular addition
[Thm.\,\ref{thrm:mod-attack}]
}} & Bounded DP & $\max(T)-\min(T)$ & $nU > \max(T)$ \\

\cline{2-4}

\multicolumn{1}{|l|}{\raisebox{0.1cm}{(signed or unsigned \texttt{ints})}}& Unbounded DP & $\max(T)-\min(T)$ & \\

\hline

\multirow{3}{*}{Saturation (signed \texttt{ints})}  & \multirow{2}{*}{Bounded DP} & \multirow{2}{*}{$\max(T)-\min(T)$} & $nU > \max(T)$ \\ & & & \raisebox{0.15cm}{$U>0$ and $L<0$} \\[-0.15cm] \cline{2-4}

\multicolumn{1}{|l|}{\raisebox{0.1cm}{
[Thm.\,\ref{thrm:sens-ints-sat-unknown}]
}}
 & Unbounded DP  & $\max(T)-\min(T)$ & $U>0$ and $L<0$ \\ \cline{1-4}

\multirow{2}{*}{Floats with $k$-bit mantissa} & Bounded DP\
[Thm.\,\ref{prop:rounding-ham-floats}]
& $\geq (n-1) \cdot (U-L)$ & $U \geq 2^k \cdot (U-L)$ \\ \cline{2-4}

& Unbounded DP
[Thm.\,\ref{thrm:quadratic-round-attack}]
& $(1+ \Theta(n^2/2^k)) \cdot U$ & $L \leq -U \cdot n/2^{k+3}$ \\ \cline{1-4}

\end{tabular}

\caption{Lower bounds obtained in our attacks for numerical type $T$ and $-U \leq L \leq U$.}
\label{table:lowerbounds}

\end{center}

}

\end{table*}

\subsection{Attacks} \label{sec:attacks}

In Section~\ref{sec:concreteattack}, we show that our attacks can be carried out on many existing implementations of differential privacy. In each attack, we set the privacy-loss parameter to $\eps\in \{0.5,1\}$, set the parameters $L$, $U$, and possibly $n$ as required by our sensitivity lower bounds, construct two appropriate adjacent datasets $u\bdd u'$ or $u\unbdd u'$, and run the supposedly $\eps$-DP Bounded Sum mechanism $\cM$ on $u$ and $u'$. We show that by applying a threshold test to the outputs, we can almost perfectly distinguish $\cM(u)$ and $\cM(u')$.  For example, in one experiment, we succeed in correctly identifying whether the dataset is $u$ or $u'$ in all but 3 out of 20,000 runs of the mechanism, which we show would be astronomically unlikely for a mechanism that is truly $0.5$-DP. These attacks are described in more detail in Section~\ref{sec:bounded-sum-computers}.

Our overflow attack on integers works on IBM's diffprivlib (Section~\ref{sec:ibmattack}).
Our rounding attack on floats works on Google's DP library, IBM's diffprivlib, OpenDP / SmartNoise,\footnote{Our paper attacks the code at \url{https://github.com/opendp/opendp}, and the attacks also work on OpenDP/SmartNoise-core at \url{https://github.com/opendp/smartnoise-core}.} and Chorus.\footnote{Chorus only offers sensitivities for the bounded DP setting, so slight adjustments were made for the repeated rounding attacks --- specifically, a value of 0 was appended to dataset $u'$ in \cref{eqn:u_prime}.}
Our repeated rounding attacks on floats work on Google's DP library, IBM's diffprivlib, OpenDP / SmartNoise,\footnote{In particular, we are able to attack \texttt{opendp.trans.make\_sized\_bounded\_sum} using the attack described in Section~\ref{sec:round-attack-floats}, and to attack \texttt{opendp.trans.make\_bounded\_sum}. using the attack described in Section~\ref{sec:accum-round-floats-ii}.}
Chorus, and PINQ. Lastly, our re-ordering attack works on Google's DP library, IBM's diffprivlib, OpenDP / SmartNoise, Chorus, and PINQ (Section~\ref{sec:implement-nonassoc-attacks}).
With the exception of OpenDP, none of the libraries include a disclaimer indicating that the integer overflow or floating-point rounding behaviors we exploit could yield vulnerabilities.

\newcommand{\val}{\mathrm{val}}
\newcommand{\id}{\mathrm{uid}}

Our attacks utilize contrived datasets $u$ and $u'$ that are unlikely occur ``in the wild.''  However, our overflow and rounding attacks can easily be applied to attack realistic datasets given access to a DP system that fields external queries with user-defined microqueries (a.k.a., mappers or row transforms).
Specifically, consider a dataset $x=[x_1,\ldots,x_n]$ where each record $x_i$ contains a known/public user identifier $\id_i$ for the $i$'th individual, and a sensitive bit $\val_i$.  Suppose further that the dataset is sorted by the identifiers, i.e., $\id_1\leq \id_2\leq\cdots \leq \id_n$.
Then it is straightforward for an adversarial analyst to define a simple micro-query $q$
such that
$$q(x)=[q(x_1),\ldots,q(x_n)]  =
\begin{cases}
u  & \text{if $\val_i=1$}\\
u' & \text{if $\val_i=0$}
\end{cases}
$$
where $u, u'$ are the datasets in our attacks, and they differ on the $i$'th record.  Thus, by seeing the result of the bounded-sum mechanism on $q(x)$, the 
adversary can extract the $i$'th individual's bit $\val_i$ with very high probability. This can be generalized to not just attack a particular individual, but a constant fraction of the individuals in the dataset (e.g., any individual in the middle half of the dataset).
Note that our use of microqueries here is different than in Haeberlen et al.~\cite{hpn11}.  Our microqueries are pure functions, not making use of any side channels or global state, and are simple enough to be implemented in virtually any domain-specific language (DSL) for microqueries.  Indeed, see Figure~\ref{fig:sql-code} for a code snippet showing how our attack would look in SQL and Figure~\ref{fig:pinq-code} for PINQ.

\begin{figure}[H]
    \begin{lstlisting}[language=SQL,frame=single]
SELECT 
    SUM(CASE 
        WHEN uid < {m} THEN {U}
        WHEN uid = {m} THEN val * {U}
        WHEN uid % 2 = 1 THEN {pos_val}
        ELSE {neg_val}
    END)
FROM unsized_64_u
\end{lstlisting}
    \caption{Example SQL code for our repeated rounding attack.  Here the upper clamping bound is $\texttt{U}$, the lower clamping bound is $\texttt{L=-U}$, 
    $\texttt{uid}$ denotes the user id attribute of an individual record, $\texttt{m}$ is the middle user id (the one we wish to attack), $\texttt{val}$ is the individual's sensitive bit, and $\texttt{pos\_val}$ and $\texttt{neg\_val}$ are particular floats in the interval $[\texttt{L},\texttt{U}]$ from our attack.} 
    \label{fig:sql-code}
\end{figure}

\begin{figure}[H]
    \begin{lstlisting}[frame=single]
var dp_sum = new PINQueryable<string>(arr_v_queryable, new PINQAgentLogger(filepath))
    .Select(l => l.Split(','))
    .Select(terms => new Tuple<long, double>(Convert.ToInt64(terms[0]), Convert.ToDouble(terms[1])))
    .Select(tup => (tup.Item1 < attackUID) ? U :
        (tup.Item1 == attackUID) ? U * tup.Item2 :
        (tup.Item2 % 2 == 1) ? pos_val : neg_val)
    .NoisySum(100.0, v => v);
\end{lstlisting}
\caption{Example PINQ code for our repeated rounding attack.}
\label{fig:pinq-code}
\end{figure}

Since PINQ uses 64-bit floats, our repeated rounding attack requires quite a large dataset (e.g., $n=2^{28}$), and we were not able to complete execution due to memory timeouts, but the attack is possible in principle.  Perhaps due to concerns about timing and other side-channel attacks (see Section~\ref{sec:prior}), recent DP SQL systems such as Chorus do not support user-defined microqueries. However, it may be still be possible to carry out our basic rounding attack (not the repeated rounding one), since it can be implemented in such a way that the microquery is simply the clamp function (applied automatically in most implementations of Noisy Bounded Sum), if we know that the dataset is sorted by the value we wish to attack (e.g., consider a dataset of salaries, and where the employee with second-highest salary wishes to find out the CEO's salary).
Finally, we also hypothesize that the attack can be carried out on DP Machine Learning systems that use DP-SGD, using a user-defined loss function $\ell(\theta,x_i)$, designed so that $\grad_\theta \ell(\theta,x_i)$ equals our desired microquery $q(x_i)$.  If we set parameters so that there is only one iteration, using the entire dataset as a batch, then the output parameter vector $\theta$ will exactly give us a Noisy Bounded Sum. Since these ML systems allow using low-precision floats, these attacks should be feasible even on very small datasets.

To carry out our reordering attacks on a DP query system, we would need a system that can be made to perform a data-dependent reordering. The main takeaway from these attacks is that we need to be explicit and careful about how we treat ordering when we implement DP.

In any case, our attacks and experimental results reported in the appendix already demonstrate that essentially all of the implementations of DP fail to meet their promised DP guarantees.

\medskip

\textit{Responsible Disclosure.} Immediately after the submission of this paper, we shared the paper with the maintainers of all of the DP libraries that exhibit the vulnerabilities we described and informed them that we would wait 30 days to post the paper publicly, to give them time to implement any needed patches (like our solutions below). In response, all offered some indication that they are working to resolve these issues, and Google's Bug Hunter program acknowledged our contribution to Google's security with an Honorable Mention.\footnote{\url{https://bughunters.google.com/profile/d946f172-9bd8-4b84-9f17-d86046f5af11}.}

All of the authors of this paper are members of the OpenDP team and are involved with the development of the library. The findings of this paper occurred while the authors were writing mathematical proofs to accompany the algorithms that are part of the OpenDP library as part of the OpenDP vetting process, upon which we realized the vulnerabilities described in this paper. We are updating the OpenDP library following the roadmap described in Section~\ref{sec:roadmap}.\footnote{For example, see \url{https://github.com/opendp/opendp/pull/465} and \url{https://github.com/opendp/opendp/pull/467}.} We believe that our findings also illustrate the importance of vetting processes such as the one put in place for OpenDP.

\textit{Code.} Code for our attacks and experiments is available in the Github repository at \url{https://github.com/cwagaman/underestimate-sensitivity}.

\subsection{Solutions}\label{sec:sols2}

\paragraph{Dataset Adjacency Relations (Section~\ref{sec:distance-funcs}).}
Toward addressing the potential reordering vulnerabilities, we propose that when datasets $u=[u_1,\ldots,u_n]$ are stored as ordered tables, we should define and distinguish four adjacency 
relations
$\simeq$.  Specifically, the bounded-DP relation $\bdd$ should separate into 
the usual $\simham$ (``Hamming'', Definition~\ref{defn:ham-dist}\color{black})
where $u\simham u'$ means that there is at most one coordinate $i$ such that $u_i\neq u_i'$, and $\simCO$ (``change-one'', Definition~\ref{defn:co-dist}\color{black}),
where $u\simCO u'$ means that we can convert the multiset of elements in $u$ into the multiset of elements in $u'$ by changing one element.  Equivalently, $u\simCO u'$ iff there is a permutation $\pi$ such that $\pi(u) \simham u'$.  Similarly, the unbounded-DP relation $\unbdd$ should split into an ordered version $\simID$  (``insert-delete'', Definition~\ref{defn:id-dist}\color{black})
and an unordered version $\simSym$ (``symmetric distance'', Definition~\ref{defn:symm-dist}\color{black}). 
The relationships between the four adjacency relations are summarized in Table~\ref{table:metrics} and Lemma~\ref{lemma:metric-relate}.
Throughout this paper, all of our theorems clearly state the adjacency relation that is being used. \color{black}
By being explicit about ordering in our adjacency relation, and in particular analyzing DP and sensitivity with respect to a specific relation, we can avoid the reordering vulnerabilities.  In particular, a DP system using ordered adjacency relations should take care to disallow data-dependent reorderings, unless they can be shown to preserve ordered adjacency (or, more generally are stable with respect to Hamming distance or insert-delete distance). 

\paragraph{Random Permutations (RP) (Section~\ref{sec:rp-sat-int}).}
Given the above adjacency relations, it still remains to find implementations of Bounded Sum that achieve a desired sensitivity with respect to them. In general, it is easier to bound sensitivity with respect to the ordered relations $\simham$ and $\simID$.  For example, we can show that Iterated Sum of signed integers with saturation arithmetic has the idealized sensitivities of $U-L$ and $\max\{U,|L|\}$ with respect to $\simham$ and $\simID$, respectively, whereas our reordering attacks show that the sensitivities can be as large as $2^k-1$
with respect to $\simCO$ and $\simSym$.

Motivated by this observation in Method~\ref{method:sat-add-RP-ints},
we give a general method for converting sensitivity bounds with respect to the ordered relations into the same bounds with respect to the unordered relations: randomly permute the dataset before applying the function.  To formalize the effect of this transformation, we need to extend the definition of sensitivity to randomized functions.  
Following \cite{sv21}, we say that a randomized function $f$ has {\em sensitivity at most $\Delta$ with respect to $\simeq$} if for all pairs of datasets $u,u'$ such that $u\simeq u'$, there is a {\em coupling} of the random variables $f(u)$ and $f(u')$ such that $\Pr[|f(u)-f(u')|\leq \Delta]=1$. We prove that if $\rp$ is the random permutation transformation on datasets and $f$ is any function on datasets, we have:
\begin{eqnarray*}
    \ssym (f\circ \rp) &\leq& \sid  f, \text{ and}\\
    \sco  (f\circ \rp) &\leq& \sham f.
\end{eqnarray*}
Given this result (Theorem~\ref{thrm:rpsummation}),
it suffices for us to obtain sensitivity bounds with respect to ordered adjacency relations.

\paragraph{Checking or Bounding Parameters (Section~\ref{sec:checkmul}).}
As can be seen from Table~\ref{table:lowerbounds}, many of our attacks only lead to a large increase in sensitivity under certain parameter regimes, e.g., $n\cdot U \geq 2^k$.  In the case of bounded DP, all of these parameters ($n, U, L, k$) are known in advance (before we touch the sensitive dataset), and we can prevent the problematic scenarios by either constraining the parameters or incorporating the dependence on those parameters into our sensitivity bounds. Indeed, we show that many of the sensitivity lower bounds discussed
in Section~\ref{sec:lowerbounds} are tight by giving nearly matching upper bounds.

For example, in the case of integer data types with bounded DP,
we prove that the implemented sensitivity equals the idealized sensitivity 
provided that $U \cdot n \leq \max(T)$  
and $L \cdot n \geq \min(T)$.
% (Theorem~\ref{thrm:checked-arith-known}).
These conditions ensure that overflow cannot occur.  Since for bounded DP, $U$, $L$, and $n$ are public parameters that do not depend on the sensitive dataset, we can simply check that these conditions hold before performing Bounded Sum. 

For floats, we provide a similar-style parameter check to ensure that a summation does not hit $\pm \texttt{inf}$ (Section~\ref{sec:61floats}).

\newcommand{\Trunc}{\mathit{Trunc}}

\paragraph{Truncated Summation \cref{method:trunc}.} 
For the case of unbounded DP, we cannot
perform a parameter check involving $n$ as above, since $n$ is not publicly known and may be sensitive information.   We can, however, achieve a similar effect by composing a solution for bounded DP with a {\em truncation} 
operation on datasets, namely
$$\Trunc_n(u) = [u_1,u_2,\ldots,u_{\min\{n,\len(u)\}}].$$
\iffalse
the sum using a bound $\nmax$, namely implementing
\begin{equation} 
\label{eqn:truncation}
\bs^{**}_{L,U,\nmax}(u) = \bs^*_{L,U}(u_1,u_2,\ldots,u_{\min\{n,\len(u)\}}).
\end{equation}
\fi
This truncation operation behaves nicely with respect to the {\em ordered} adjacency relation $\simID$.  Specifically, we prove (Theorem~\ref{thrm:checked-arith-known})
that for every function $f$ on datasets,
\begin{equation} \label{eqn:truncation}
\Delta_{\simID} \left(f\circ \Trunc_n\right) \leq 
\max\left\{\Delta_{\simID,\leq n} f, \Delta_{\simCO,\leq n} f\right\},
\end{equation}
where $\Delta_{\simeq,\leq n}$ denotes sensitivity restricted to datasets $u\simeq u'$ of length at most $n$.  Intuitively, inserting or deleting an element from a dataset $u$ either results in inserting or deleting an element from $\Trunc_n(u)$ (if no truncation occurs) or changing one element of $u$ (if truncation occurs).  

Applying (\ref{eqn:truncation}) to the case of truncated integer summation, where $f$ is Iterated Summation, we recover the idealized sensitivity with respect to $\simID$ provided that $n\cdot U\leq \max(T)$ and $n\cdot L \geq \min(T)$ (to prevent overflow) and both $U$ and $L$ are of the same sign (so that $U-L \leq \max\{U,|L|\}$ and the idealized sensitivity with respect to $\simCO$ is no larger than the idealized sensitivity with respect to $\simID$).

\paragraph{Split Summation (\cref{thrm:rp-split-summation}).} The above example of truncated integer summation with respect to unbounded DP is one of several cases where we obtain better sensitivity bounds when $U$ and $L$ are of the same sign.  Another is integer summation with saturation arithmetic with respect to unordered adjacency relations: when $U$ and $L$ are of different signs, this is vulnerable to our reordering attack, but when they are of the same sign, we show that it has implemented sensitivity equal to the idealized sensitivity. We also give an example below with floating-point numbers where computing sums on terms with matching signs helps achieve an implemented sensitivity that is closer to the idealized sensitivity.  

To take advantage of the benefits that can come from summing terms with the same sign, we introduce  the {\em split summation} technique, where we separately sum the positive numbers and the negative numbers in the dataset (Method~\ref{method:split-sum-ints}).
That is, given a dataset $u$ and a base summation method $\bs^*$, we let $\pos(u)$ be the dataset consisting of the positive elements of $u$, $\negs(u)$ be the dataset consisting of the negative elements of $u$, and define
$$
\bs^{**}_{L,U}(u) = \bs^*_{0,U}(\pos(u)) + \bs^*_{L,0}(\negs(u)).
$$
Importantly, adding or removing an element from a dataset $u$ corresponds to adding or removing an element from only one of $\pos(u)$ and $\negs(u)$, so we do not incur a factor of 2 blow-up in sensitivity.  Using this split summation technique in combination with either truncation or saturation arithmetic allows us to recover the idealized sensitivity for summation of signed integers with unbounded DP.
\color{black}

\paragraph{Sensitivity from Accuracy (Section~\ref{sec:accuracy}).
} Our matching upper bound on the sensitivity of Iterated Summation of floats with banker's rounding is obtained via reduction to accuracy (Lemma~\ref{lemma:sa}).
Specifically, we can use the triangle inequality to show that we have
\begin{eqnarray*}
\lefteqn{|\bs^*(u)-\bs^*(u')| \leq} &&\\
&& |\bs(u)-\bs(u')| + |\bs^*(u)-\bs(u)| + |\bs^*(u')-\bs(u')|.
\end{eqnarray*}
The first term is bounded by the idealized sensitivity, and the latter two terms can be bounded by using known numerical analysis results about the accuracy of iterated summation. Specifically, a bound from Wilkinson~\cite{w60} shows that
$$|\bs^*(u) - \bs(u) | 
= O\left(\dfrac{n}{2^k} \sum_{i=1}^n |u_i|\right)
%+ O\left(\frac{\max_i |u_i|}{2^k}\right)^2
= O\left(\frac{n^2\cdot \max\{|L|, |U|\}}{2^k}\right).$$

Note that the resulting upper bound on sensitivity described above has an ``$n$'' term, so it can only be applied directly in the case of bounded DP.  To handle unbounded DP, we can combine it with the truncation technique described above.
\iffalse
seems to indicate the knowing the size of the dataset is necessary. Thus, while this upper bound works immediately in the bounded DP case since the length of the dataset is known in this setting, it does not work immediately in the unbounded DP case since the length of the dataset may be unknown in this setting. We then propose a change to the summation strategy so that this sensitivity can apply to the unbounded DP case; namely, to use \textit{truncated} Bounded Sum, which requires the user to provide an integer $n_{\max}$ that represents the (maximum) number of terms in the dataset that will be included in the sum. This allows the sensitivity to apply to the unbounded DP case, in addition to the bounded DP case. 
\fi

Specifically, combining Iterated Summation $\bs^*$ with
the Truncation transformation $\Trunc_n$, we get an unbounded-DP sensitivity bound of:
\begin{equation} \label{eqn:truncation-iterated}
    \Delta_{\simID} \left(\bs^*_{L,U}\circ \Trunc_n\right) \leq
\left(1+ O\left(\frac{n^2}{2^k}\right)\right)\cdot \max\{|L|, |U|\},
\end{equation}
provided that $L$ and $U$ have the same sign, and similarly for the unordered sensitivity $\Delta_{\simSym}$ if we also combine with a random permutation.
Thus, when $n\ll 2^{k/2}$, we recover almost the idealized sensitivity bound of $\max\{|L|, |U|\}$.  
Higham~\cite{h93} has proven that other summation methods for floats, such as Pairwise Summation (which is the default in \texttt{numpy}) and Kahan Summation have better bounds on accuracy,
allowing us to replace the $O(n^2)$ above with $O(n\log n)$ or $O(n)$, respectively (Theorem~\ref{thrm:sens-upper-bounds}).
These methods closely approximate the idealized sensitivity whenever $n\ll 2^k$, which covers many practical scenarios (see discussion in Section~\ref{sec:roadmap}). 
%Pairwise Summation and Kahan Summation determine an order of summation of elements in a dataset such that the accumulated error is minimized. Because Pairwise Summation and Kahan Summation are more accurate \cite{h93, w60, kahan1965pracniques}, the sensitivity of a Bounded Sum algorithm that uses Kahan or pairwise summation will be lower than the sensitivity of an iterative Bounded Sum.

\color{black}

\paragraph{Shifting Bounds (Section~\ref{sec:shift-bounds}).}
For the case of bounded DP, some of the sensitivity blow-ups (such as the one exhibited in the basic rounding attack) come from having $U$ (or $\max\{|L|, |U|\}$) rather than $U-L$ in the implemented sensitivity bound, as $U$ can be much larger than $U-L$.  One way to address this is to subtract $L$ from every element of the dataset, so that all elements lie in the interval $[L'=0,U'=U-L]$, apply our solutions that have $U'$ in the sensitivity bound, and then add $L\cdot n$ at the end as post-processing.  That is, we convert a method $\bs^*$ with sensitivity depending on $\max\{|L|, |U|\}$ (such as \cref{eqn:truncation-iterated}) into a method $\bs^{**}$ with sensitivity depending on $U-L$ as follows:
$$\bs^{**}_{L,U,n}(u) = \bs^*_{0,U-L,n}(u_1-L,u_2-L,\ldots,u_n-L)+L\cdot n,$$
where a noisy version of the $\bsfp_{L,U,n}$ term is computed first and $L\cdot n$ is added as a post-processing step, after noise addition.
In particular, combining this technique with the Sensitivity-via-Accuracy analysis of Iterated Summation, we obtain a bounded-DP sensitivity bound of 
$$\Delta_{\simham} \bs^{**}_{L,U,n}(u) \leq
\left(1+ O\left(\frac{n^2}{2^k}\right)\right)\cdot (U-L).$$\color{black}

\iffalse
In the case where $U-L$ is much smaller than $\max\{|L|,U\}$ (i.e., $U$ and $L$ are ``large'' and ``close''), the upper bound on sensitivity that we get from this accuracy-based bound will be much larger than $U-L$ since the accuracy depends on $\max\{|L|,U\}$. Shifting the bounds makes $\max\{|L|,U\}$ of a similar magnitude to $U-L$, allowing us to work with a sensitivity much closer to the idealized sensitivity. This is discussed in more detail in Section~\ref{sec:shift-bounds}.
\fi \color{black}

\paragraph{Reducing Floats to Ints (Section~\ref{sec:floats2ints}).} Another attractive solution for floating-point summation is to reduce it to integer summation, since the latter achieves the idealized sensitivity with simple solutions. The idea behind this method is to cast floats to fixed-point numbers, which we can think of as $k$-bit integers. To achieve this casting, we introduce a discretization parameter $D$ (which is chosen by the data analyst and corresponds to the precision with which numbers are represented -- e.g., setting $D=0.01$ means that values are represented to the hundredths place), round each of the dataset elements according to the discretized interval, and apply one of our integer solutions. We can think of this process of ``integerizing'' as a dataset transform. This idea of mapping floating-point values to integers is similar to the technique of performing \textit{quantization} to work with low-precision number formats in machine learning applications \cite{quantization}.

We provide a complete description and analysis of this solution in Section~\ref{sec:floats2ints}. To simplify the description here in the introduction, we present the solution in the specific case where $L = -U$. For general intervals $[L,U]$ and bounded DP, we shift the discretization range to take advantage of the full range of $k$-bit integers (see Section~\ref{sec:floats2ints}).

Before describing this method formally, we provide some intuition. The idea behind this strategy is to group floating-point values into buckets (where close values are put into the same bucket or close buckets) and enumerate the buckets. More specifically, given the bounds $L$ and $U$ for our floats and the discretization parameter $D$, where $L = -U$ and $K=\lfloor U/D\rfloor$, we round all elements of the interval $[L,U]$ to the nearest element of the sequence
$$-KD, -(K-1)D, \ldots, -D, 0, D, \ldots, (K-1)D, KD,$$ 
so $|L|\approx U \approx KD.$
We then use the signed integers to enumerate this sequence.
If $K\leq (2^{k-1}-1)$, we can think of the rounded elements as corresponding to  $k$-bit integers, which we can then sum using a summation method for integers. We can then post-process the resulting (noisy) sum into a (noisy) floating-point sum. We discuss how to optimize the choice of the discretization parameter $D$ after the description of the method.  

Formally, we can consider the function mapping a dataset
$u=[u_1,\ldots,u_n]$ of floats to the signed integer dataset
$$\FloatToInt_{L,U,D}(u) = [\round(u_1/D),\ldots,\round(u_n/D)],$$
where $\round(\cdot)$ denotes rounding to the nearest integer.
We can then apply a summation method $\bs^*_{-K,K}$ for $k$-bit integers, and then rescale and shift to obtain our floating-point result:
\begin{equation}
\label{eq:f2i-no-post}
\bs^{**}_{L,U,D}(u)
 = (\bs^*_{-K,K}\circ \FloatToInt_{L,U,D})(u)\cdot D.
\end{equation}
The sensitivity of (\ref{eq:f2i-no-post}) can be bounded: as long as we apply one of our solutions for integers, the Bounded Sum of the integers will have the idealized sensitivity. However, to bound the sensitivity of the overall function, we would still need to account for the potential rounding that can occur when multiplying the integer sum by $D$.

An even better approach for the summation than (\ref{eq:f2i-no-post}) is to perform noise addition and obtain DP {\em before} scaling back to floats.  That is, we consider the DP mechanism
$$\cM_{L,U,D}(u)
 = \left((\bs^*_{-K,K}\circ \FloatToInt_{L,U,D})(u)+\Noise(K/\varepsilon)\right)\cdot D,$$
where $\Noise(s)$ denotes a noise distribution with scale $s$ suitable for $k$-bit integers (e.g., the discrete Laplace Mechanism~\cite{GhoshRoSu12}).
Then the privacy of $\cM$ follows from the privacy of the noisy integer summation together with the post-processing property of differential privacy, and there is no need to analyze any floating-point rounding effects in either the sensitivity analysis or the noise addition step. Indeed, implementing
discrete noise-addition mechanisms is much simpler than implementing
floating-point noise-addition mechanisms (such as Mironov's Snapping Mechanism~\cite{Mironov12}).

One remaining question is how to pick the discretization parameter $D$ to maximize accuracy. 
A choice that maintains high precision and reduces the risk of (still-private) answers that overflow is $$D=(U-L)\cdot \nmax/(2^{m-2}-\nmax),$$ where $\nmax$ is the maximum expected dataset size and $m$ is the bitlength of the integers we are using (e.g., $m=64$). This choice of $D$ ensures that, if our dataset size is at most $\nmax$, then the sum will be no larger than $K\cdot \nmax < 2^m/4$.

For example, when working with the discrete Laplace Mechanism, by analyzing the tail bounds of the distribution, the probability of overflow for the Noisy Bounded Sum can be shown to be
$\exp(-\Omega(K\cdot \nmax/(K/\varepsilon))) = \exp(-\Omega(\eps\cdot \nmax)),$
which will be astronomically small for typical settings of parameters.

In \cref{sec:floats2ints}, we also analyze the impact of the rounding in $\FloatToInt$ on accuracy. We show that the rounding incurs an additive error of at most
$O((U-L)\cdot n\cdot \nmax/2^m)$ on worst-case datasets and
$O((U-L)\cdot \sqrt{n}\cdot \nmax/2^m)$ on typical datasets.  These errors are comparable to the accuracy bounds for non-private iterated summation of $(k,\ell)$-bit floats~\cite{h93}, with the advantage of not affecting the sensitivity or privacy analysis.

This solution is a variant of the common suggestion to replace floating-point arithmetic in DP libraries with integer or fixed-point arithmetic~\cite{BalcerVa18,ghv20}.  However, for a data analyst, our solution retains the usability benefits of floating-point numbers, such as the large dynamic range afforded by the varying exponents.  Indeed, the input dataset and output result remain as floating-point numbers, and the conversion to and from integer/fixed-point representation only happens internal to the (noisy) sum. In particular, the mapping between floating-point numbers and fixed-point numbers is determined dynamically based on the parameters $L$ and $U$, in contrast to adopting a single fixed-point representation throughout a DP library.

\paragraph{Modular Sensitivity (Section~\ref{sec:modularnoise}).
}  One way to address the large sensitivity of Bounded Sum over $k$-bit integers with wraparound is to change our definition of sensitivity, measuring the distance between $k$-bit integers as if they are equally spaced points on a circle. That is, we replace $|f(u)-f(u')|$ in the definition of sensitivity with $\min\{|f(u)\ominus f(u')|, |f(u') \ominus f(u)|\}$, where $\ominus$ is subtraction with wraparound over the integer type $T$ on which Bounded Sum is being computed; we call this the {\em modular sensitivity} of $f$ (Definition~\ref{defn:mod-sens}).  With this change, Bounded Sum recovers its idealized sensitivity (Theorem~\ref{thrm:modular-arith-unknown}). But this begs the question of whether functions with bounded modular sensitivity can still be estimated in a DP manner. Fortunately, we show that the answer is {\em yes}: if we add {\em integer-valued} noise (such as in the Discrete Laplace~\cite{GhoshRoSu12} or Discrete Gaussian~\cite{BunSt16} Mechanisms) and also do the noise addition with wraparound, then the result achieves the same privacy parameters as if we had done everything exactly over the integers, with no wraparound (Theorem~\ref{thm:modularnoise}).
Indeed, by the fact that modular reduction is a ring homomorphism, we can analyze the output distribution as if we had done modular reduction only at the end, which amounts to post-processing.  

Several DP libraries already implement this solution, because it happens by default when everything is a $k$-bit integer. The reason we were able to attack IBM's diffprivlib with the overflow attack (Section~\ref{sec:ibmattack})
is that, while it computes the Bounded Sum using integer arithmetic, the noise addition is done using floating-point arithmetic. 

It is not clear whether there is an analog of this solution for the use of sensitivity in the Exponential Mechanism (see Section~\ref{sec:prior}).

\paragraph{Changing Overflow Mode (Sections \ref{sec:splitsum} and \ref{sec:rp-sat-int}).
}  As mentioned above, another solution for the case of $k$-bit integers is to replace wraparound with saturation arithmetic, combining it with either the random permutation technique or split summation in order to handle unordered adjacency relations. Saturation arithmetic with split summation is analyzed in Section~\ref{sec:splitsum}, and with randomized permutation in Section~\ref{sec:rp-sat-int}.

\paragraph{Changing Rounding Mode (Section \ref{sec:splitsumrp}).
}  For our rounding attacks against floating-point numbers, another solution (beyond those based on sensitivity-via-accuracy) 
is to replace the default banker's rounding with another standard rounding mode, namely round toward zero (RTZ). The algorithm is described in Method~\ref{method:fsplit-sum}. In Theorems~\ref{thrm:rp-split-summation} and \ref{thrm:ord-toward-zero-known}, we show that this gives an implemented sensitivity of
$$\Delta_{\unbdd} \bsfp_{L,U,n}
\leq \left(\max\left\{1+O\left(\frac{n}{2^k}\right),2\right\}\right) \cdot \max\{|U|,|L|\},$$
provided that (a) all elements of the datasets have the same sign (i.e., $L\geq 0$ or $U\leq 0$), and (b) we work with an ordered dataset adjacency relation. To handle the  case of mixed signs, we can use
the shifting technique (in case of bounded DP) or split summation (in the case of unbounded DP).  
To handle unordered dataset relations, we can apply the random permutation technique (Theorem~\ref{thrm:ord-toward-zero-known}).
Altogether these solutions maintain a sensitivity that is within a small constant factor of the idealized sensitivity in all cases.

\def\arraystretch{1.2}

\begin{table*}

{\normalsize

\begin{center}

\begin{tabular}{|c|l|c|c|c|}\hline

\textbf{Data type} & \multicolumn{2}{|c|}{\raisebox{-0.2cm}{\textbf{Solution name}}} & \raisebox{-0.2cm}{$\dfrac{\textit{implemented sensitivity}}{\textit{idealized sensitivity}}$} & \raisebox{-0.2cm}{\textbf{Conditions}} \\[0.4cm] \hline

\texttt{int}s & \multicolumn{2}{|l|}{RP [Thm.\,\ref{thrm:rpsummation}]} & & \\%[-0.18cm]

\cline{2-3}

& \multicolumn{2}{|l|}{Dataset adjacency relations [Thm.\,\ref{thrm:ord-dist-saturation-sens}]} & 1 & \\%[-0.25cm]

\cline{2-3}

& \multicolumn{2}{|l|}{Checking parameters [Thm.\,\ref{thrm:checked-arith-known}]} & & \\%[-0.15cm]

\cline{2-3}

& \multicolumn{2}{|l|}{Modular noise addition [Thm.\,\ref{thrm:modular-arith-unknown}]} & & \\

\hline

Floats & \multicolumn{2}{|l|}{Reducing floats to ints [Sec. \ref{sec:floats2ints}]} & $1 + O(n\sqrt{n}/2^m)$ & \\[0.1cm] \hline

Floats & \multicolumn{2}{|l|}{\raisebox{-0.1cm}{RP + Split summation + RTZ}} & $\min\{(1+O(n/2^k)), 2\}$ & $\textrm{sign}(U) = \textrm{sign}(L)$ \\

\cline{4-5}

& \multicolumn{2}{|l|}{[Thms.\,\ref{thrm:rp-split-summation}, \ref{thrm:toward-zero-known}]} & $\min\{(2+O(n/2^k)), 5\}$ & $\textrm{sign}(U) \neq \textrm{sign}(L)$ \\

\cline{1-5}

Floats & \multicolumn{1}{|l|}{\raisebox{-0.2cm}{Sensitivity from accuracy}} & Iterative & $1+\Theta(n^2/2^k)$ & \\

\cline{3-5}

& \multicolumn{1}{|l|}{\raisebox{0.05cm}{+ Truncated summation}} &  Pairwise & $1+O(n \log(n) / 2^k)$ & $n<2^k$ \\

\cline{3-5}

& \multicolumn{1}{|l|}{\raisebox{0.2cm}{[Thm.\,\ref{thrm:sens-upper-bounds}]}} & Kahan & $1 + O(n/2^k)$ & $n < 2^k$  \\

\cline{1-5}

\end{tabular}

\caption{Upper bounds obtained in our solutions for numerical type $T$, $-U \leq L \leq U$, and datasets of length $n$. The parameter $k$ is the number of bits (for $k$-bit ints) and the mantissa length (for $(k,\ell)$-bit floats); the parameter $m$ is the bit-length of the integer data type to which the floats were reduced. We remark that in the iterative, pairwise, and Kahan sensitivities, the 1 factor becomes a 2 in the case where $n$ is unknown and $\textrm{sign}(U) \neq \textrm{sign}(L)$. Note that all the methods listed in a given solution box need to be used in combination (e.g., RP, Split Summation, and RTZ all need to be used together).{}}\label{table:upperbounds}

\end{center}

}

\end{table*}

\subsection{Roadmap for Implementing Solutions}\label{sec:roadmap}

In this section, we present a set of recommendations aimed at DP practitioners who wish to fix the vulnerabilities that we have presented in Section \ref{sec:bounded-sum-computers}.
Table~\ref{table:upperbounds} presents several solutions and their associated sensitivities. Several of the solutions can be implemented with only a few alterations to current code.

\subsubsection{Integer Summation} 
We recommend the following solution  as easiest to implement for integer summation:

\begin{itemize}
    \item Modular sensitivity (Section~\ref{sec:modularnoise}):
    In many programming languages, the default method for handling integer overflow is wraparound, which is equivalent to modular summation. Thus, many DP libraries luckily already implement the modular solution that we prove to be correct in Theorem~\ref{thm:modularnoise}.
    One point of caution is that {\em both} the summation and noise addition steps must occur in this modular fashion.
    Wraparound is a standard feature of arithmetic on integer data types, so this solution should not create new unexpected behaviors for data analysts. 
    
\end{itemize}

If library maintainers are uncomfortable with the possibility of wraparound or unable to offer modular noise addition, we encourage the following two solutions.

\begin{itemize}
    \item Checking parameters (Section~\ref{sec:61ints}):
    In the bounded DP setting, perform a check on the parameters ensuring that overflow cannot occur, e.g., check that $n\cdot U < 2^k$ in the setting of unsigned $k$-bit integers. 
    
    \item Split summation (Section~\ref{sec:splitsum}):
    In the unbounded DP setting, we recommend switching to saturation arithmetic (where overflow is handled by clamping to the range $[\min(T),\max(T)]$) and applying split summation, where we separately sum the positive and negative numbers (Method~\ref{method:split-sum-ints}).
    For $L$ and $U$ both non-negative or both non-positive, split summation is equivalent to standard summation. If split summation is not desirable, then we recommend using saturation arithmetic, and either applying a randomized permutation to the dataset (Method~\ref{method:sat-add-RP-ints})
    or using an ordered notion of neighboring datasets and being careful about stability with respect to ordering in all dataset transformations.
\end{itemize}

\subsubsection{Floating-point Summation}
%\paragraph{$(k,l)$-bit floats.} 
For libraries that have or are implementing a (correct) version of (noisy) integer summation, we recommend using our Float2Int solution if feasible.

If floating-point summation must be used (e.g. due to a machine-learning pipeline or
exernal compute engine that has hardwired numerical types),  we recommend using the implemented sensitivity bounds that come from accuracy bounds (see Table~\ref{table:upperbounds}) and Section~\ref{sec:accuracy})
for whatever floating-point summation method is already implemented.
Typically this will be Iterated Summation, but some libraries (e.g., those based on \texttt{numpy}) may already using a better method like Pairwise Summation.
The summation method could be changed from iterative to pairwise or Kahan for tighter sensitivities (Section~\ref{sec:accuracy}). 
For bounded DP, using these sensitivities only requires also checking the parameters to ensure that overflow cannot occur (Section~\ref{sec:61floats}).
To achieve unbounded DP, truncated summation together with a random permutation (Method~\ref{method:sat-add-RP-ints}) should be used as well.
These solutions should work very well for 64-bit floats, as they give an implemented sensitivity that is at most $1.5 \cdot U$ for datasets of size smaller than 67 million.  If it is important to have sensitivity close to $U-L$, then the ``shifting bounds'' technique (Section~\ref{sec:shift-bounds})
can be used as well.

For data consisting of lower-precision floats, another attractive approach is to keep the values as low-precision floats (e.g., to keep the memory footprint small in machine learning pipelines) but accumulate the sum in a 64-bit float, which will allow the above solutions to apply. For huge datasets (e.g., more than 67 million records), a 128-bit accumulator could be used or the summation method could be switched to a method where the accuracy bounds grow more slowly with $n$ (e.g., pairwise or Kahan summation). Otherwise, we recommend considering the ``changing rounding mode'' solution described in Section~\ref{sec:splitsumrp},
which requires changing the rounding mode from the standard banker's rounding to round-toward-zero. 

\section{Preliminaries}\label{sec:preliminaries}

\subsection{Measuring Distances}
\label{sec:distance-funcs}

\begin{definition}[Dataset]
For a data domain $\mathcal{D}$, a \emph{dataset on $\mathcal{D}$} is a vector $v$ of elements from $\mathcal{D}$, i.e., $v\in \Vect(\mathcal{D}) \coloneqq \bigcup_{n\geq 0} \mathcal{D}^n$. For $v\in \mathcal{D}^n$, we write $v = [v_1,\ldots, v_n]$ for the elements of $v$ and $\len(v) = n$.

\end{definition}

Note that datasets $v\in \Vect(\mathcal{D})$ are \emph{ordered}. When we are not concerned with the order of elements in $v$, we will often refer to a dataset's \textit{histogram}.

\begin{definition}[Histogram]
\label{defn:histogram}
For a dataset $v\in \Vect(\mathcal{D})$, the \emph{histogram} of $v$ is the function $h_v:\mathcal{D}\to \mathbb{N}$ defined as $$h_v(z) = \#\{i:v_i = z\}.$$
\end{definition}

\begin{definition}[Permuting Datasets]\label{def:perm}

Let $S_n$ denote the set of permutations on $n$ elements. For a permutation $\pi\in S_{\len(v)}$, we write $\pi(v) = [v_{\pi(1)},\ldots, v_{\pi(n)}]$. 
\end{definition}

\begin{lemma}
\label{lemma:hists-equal}
For datasets $u,v \in {\rm Vec}(\mathcal{D})$, $h_u = h_v$ if and only if $\len(u) = \len(v)$ and there exists a permutation $\pi\in S_{\len(u)}$ such that $\pi(u) = v$.
\end{lemma}

Depending on which distance function we want to employ, there are different notions of \textit{neighboring datasets}. In this paper we consider the following four distance functions between vectors/datasets: symmetric distance, Hamming distance, change-one distance, and insert-delete distance. 

\begin{definition}[Symmetric Distance]
\label{defn:symm-dist}
Let $u,v\in \Vect(\mathcal{D})$. The \emph{symmetric distance} between $u$ and $v$, denoted $\dsym(u,v)$, is

$$
    \dsym(u,v) = \sum_{z\in \mathcal{D}}|h_u(z) - h_v(z)|,
$$
where $h_u$ and $h_v$ are the histograms of $u$ and $v$.

Equivalently, $\dsym(u,v)$ is the size of the symmetric difference between $u$ and $v$ when viewed as multisets.

Observe that for every permutation $\pi\in S_{\len(u)}$, $\dsym(u,\pi(u)) = 0$, so $\dsym$ (and $\dco$ below) is only a pseudometric since unequal datasets can have distance 0.

\end{definition}

\begin{definition}[Change-One Distance]
\label{defn:co-dist}
Let $u,v\in \mathcal{D}^n$. The \emph{change-one distance} between $u$ and $v$, denoted $\dco(u,v)$, is
\[
    \dco(u,v) = \sum_{\substack{z\in \mathcal{D} \text{ s.t.}\\ h_u(z) > h_v(z)}} \left( h_u(z) - h_v(z)\right)\\ = \sum_{\substack{z\in \mathcal{D} \text{ s.t.}\\ h_v(z) > h_u(z)}} \left( h_v(z) - h_u(z)\right).
\]
For $u,v\in \Vect(\mathcal{D})$ with $\len(u)\neq \len(v)$, we define $\dco(u,v) = \infty$.

Equivalently, $\dco(u,v)$ is the minimum number of elements in $u$ that need to be changed to produce $v$, when $u$ and $v$ are viewed as multisets.
\end{definition}

\begin{definition}[Insert-Delete Distance]
\label{defn:id-dist}

For $u\in\mathcal{D}^m$, an \emph{insertion} to $u$ is an addition of an element $z$ to some location within $u$, resulting in a vector $u' = [u_1,\ldots, u_i, z, u_{i+1},\ldots, u_m] \in \mathcal{D}^n$. Likewise, a \emph{deletion} from $u$ is a removal of an element from some location within $u$, resulting in a vector $u' = [u_1,\ldots, u_{i-1}, u_{i+1}, \ldots, u_m]\in \mathcal{D}^m$.

The \emph{insert-delete distance} $\did(u,v)$ between $u\in\mathcal{D}^m, v\in\mathcal{D}^n$ is the minimum number of insertion operations and deletion operations needed to change $u$ into some vector $u'$ such that $u' = v$.
\end{definition}

\begin{definition}[Hamming Distance]
\label{defn:ham-dist}

Let $u,v\in\mathcal{D}^n$. The \emph{Hamming distance} between vectors $u,v$ is

$$
    \dham(u,v) = \#\{i: u_i\neq v_i\}.
$$
For $u,v\in \Vect(\mathcal{D})$ with $\len(u)\neq \len(v)$, we define $\dham(u,v) = \infty$.

\end{definition}

Equipped with functions for measuring the distance between datasets, we can now define the notion of neighboring / adjacent datasets.

\begin{definition}[Neighboring/Adjacent Datasets]
We say that vectors $u,v\in \Vect(\mathcal{D})$ are \emph{neighbors} or \emph{adjacent} with respect to dataset distance metric $d$ whenever $d(u,v) \leq 1$. In this case, we write $u \simeq_\mathit{d} v$, or $u\simeq_\mathit{name} v$, where $d$ corresponds to $d_\mathit{name}$ (e.g., $\dsym$, $\dham$, $\did$).
\end{definition}

A key difference between these distance metrics, which we will use in our attacks, is the fact that $\dsym$ and $\dco$ are \emph{unordered} metrics (meaning the distance between two vectors only depends on their histograms), whereas $\dham$ and $\did$ are \emph{ordered} metrics (meaning the distance between two vectors is affected by the order of elements in these vectors). Another key difference is that $\dham$ and $\dco$ can only be applied to vectors of the \emph{same} length $n$. On the other hand, $\dsym$ and $\did$ apply to vectors of potentially different length. Thus, we will use $\dco$ and $\dham$ in privacy contexts where the dataset size $n$ is known and public, and $\dsym$ and $\did$ otherwise.
Table~\ref{table:metrics} summarizes these differences.

\begin{table}[!h]
        \centering
\begin{tabular}{|p{0.15\textwidth}|p{0.13\textwidth}|p{0.13\textwidth}|}
\hline 
  & Unordered & Ordered \\
 \hline
  Unknown $n$ & $\dsym$ & $\did$  \\
 \hline 
  Known $n$ & $\dco$ & $\dham$ \\
  \hline
\end{tabular}
\caption{Summary of the dataset metrics considered in this paper.}\label{table:metrics}
\end{table}

The relationships between these four metrics in terms of sensitivity are described in the following lemma, which can be summarized with the following diagram:

\begin{center}
\tikzset{every picture/.style={line width=0.75pt}} %set default line width to 0.75pt        

\begin{tikzpicture}[x=0.55pt,y=0.55pt,yscale=-1,xscale=1]
%uncomment if require: \path (0,310); %set diagram left start at 0, and has height of 310

%Straight Lines [id:da2837702580406207] 
\draw    (120,112) -- (252,112.23) ;
\draw [shift={(254,112.23)}, rotate = 180.1] [color={rgb, 255:red, 0; green, 0; blue, 0 }  ][line width=0.75]    (10.93,-3.29) .. controls (6.95,-1.4) and (3.31,-0.3) .. (0,0) .. controls (3.31,0.3) and (6.95,1.4) .. (10.93,3.29)   ;
%Straight Lines [id:da2645667736150239] 
\draw    (125,229.23) -- (249,229.23) ;
\draw [shift={(251,229.23)}, rotate = 180] [color={rgb, 255:red, 0; green, 0; blue, 0 }  ][line width=0.75]    (10.93,-3.29) .. controls (6.95,-1.4) and (3.31,-0.3) .. (0,0) .. controls (3.31,0.3) and (6.95,1.4) .. (10.93,3.29)   ;
%Straight Lines [id:da2139835663034675] 
\draw    (90,124.23) -- (90,214.23) ;
\draw [shift={(90,216.23)}, rotate = 270] [color={rgb, 255:red, 0; green, 0; blue, 0 }  ][line width=0.75]    (10.93,-3.29) .. controls (6.95,-1.4) and (3.31,-0.3) .. (0,0) .. controls (3.31,0.3) and (6.95,1.4) .. (10.93,3.29)   ;
%Straight Lines [id:da3577231520258284] 
\draw    (282,124.23) -- (282,214.23) ;
\draw [shift={(282,216.23)}, rotate = 270] [color={rgb, 255:red, 0; green, 0; blue, 0 }  ][line width=0.75]    (10.93,-3.29) .. controls (6.95,-1.4) and (3.31,-0.3) .. (0,0) .. controls (3.31,0.3) and (6.95,1.4) .. (10.93,3.29)   ;

% Text Node
\draw (69,104) node [anchor=north west][inner sep=0.75pt]   [align=left] {$\dsym$};
% Text Node
\draw (270,104) node [anchor=north west][inner sep=0.75pt]   [align=left] {$\did$};
% Text Node
\draw (66,221) node [anchor=north west][inner sep=0.75pt]   [align=left] {$2\cdot \dco$};
% Text Node
\draw (262,220) node [anchor=north west][inner sep=0.75pt]   [align=left] {$2\cdot \dham$};
% Text Node
\draw (175,88) node [anchor=north west][inner sep=0.75pt]   [align=left] {$\leq$};
% Text Node
\draw (176,205) node [anchor=north west][inner sep=0.75pt]   [align=left] {$\leq$};
% Text Node
\draw (71,159) node [anchor=north west][inner sep=0.75pt]   [align=left] {$=$};
% Text Node
\draw (263,159) node [anchor=north west][inner sep=0.75pt]   [align=left] {$\leq$};

\end{tikzpicture}

\end{center}

\begin{lemma}[Relating Metrics]
\label{lemma:metric-relate}
\label{lemma:dhamtodid}
\label{lemma:dcotodsym}
These metrics are related as follows.
\begin{enumerate}
    \item For $u,v\in\Vect(\mathcal{D})$, $$\dsym(u,v) = \min_{\pi\in S_{\len(u)}} \did(\pi(u),v)\leq \did(u,v).$$
    
    \item For $u,v\in \mathcal{D}^n$, $$\dco(u,v) = \min_{\pi\in S_n} \dham(\pi(u),v)\leq \dham(u,v).$$
    
    \item For $u,v\in \mathcal{D}^n$, $$\dsym(u,v) = 2\cdot \dco(u,v).$$
    
    \item For $u,v\in \mathcal{D}^n$, $$\did(u,v) \leq 2\cdot \dham(u,v).$$
\end{enumerate}
\end{lemma}

\ifnum\CCSFORMAT=0
The proof can be found in Appendix~\ref{appendix:sec2}.
\fi

\ifnum\CCSFORMAT=1
\begin{proof}
We prove each part below.
\begin{enumerate}
    \item Let $u,v\in\Vect(\mathcal{D})$.
    
    We first show that there exists a permutation $\pi\in S_{\len(u)}$ such that $\did(\pi(u),v)\leq \dsym(u,v)$.
    
    We begin by constructing $v'$ such that $\did(v,v') \leq \dsym(u,v)$ and $h_{v'} = h_u$. We show that we can use $\dsym(u,v)$ insertions and deletions on $v$ to get $v'$ such that $h_{v'} = h_u$. We create $v'$ from $v$ as follows: (1) for all $z\in\mathcal{D}$ such that $h_v(z) > h_{u}(z)$, we delete $h_v(z) - h_u(z)$ copies of $z$ from $v$; (2) for all $z\in\mathcal{D}$ such that $h_v(z) < h_u(z)$, we insert $h_v(z) - h_u(z)$ copies of $z$ to $v$. By \cref{defn:id-dist}, then, 
    \begin{equation}
    \label{eqn:metric-relate-did}
    \did(v',v)\leq \dsym(u,v).
    \end{equation}
    
    The resulting vector $v'$ also has the property that, for all $z\in\mathcal{D}$, $h_{v'}(z) = h_u(z)$, so $h_u = h_{v'}$. By \cref{lemma:hists-equal}, $h_u = h_{v'}$ implies that there exists a permutation $\pi\in S_{\len(u)}$ such that $\pi(u) = v'$. Substituting this into \cref{eqn:metric-relate-did}, we get
    \begin{equation}
    \did(\pi(u),v)\leq \dsym(u,v).
    \end{equation}
    
    We next show that, for all permutations $\pi\in S_{\len(u)}$, we have $\dsym(u,v)\leq \did(\pi(u),v)$, which means that $\dsym(u,v) \leq \min_{\pi\in S_{\len(u)}}\did(\pi(u),v)$. This is true because, for all permutations $\pi\in S_{\len(u)}$, $\dsym(u,v) = \dsym(\pi(u),v)\leq \did(\pi(u),v)$, with the equality following from the fact that $h_{\pi(u)} = h_u$.

    \item Let $u,v\in\mathcal{D}^n$.
    
    We first show that there is a permutation $\pi\in S_n$ such that $\dco(u,v) = \dham(\pi(u),v)$.
    
    We begin by constructing $v'$ such that $\dham(v,v') \leq \dco(u,v)$ and $h_{v'} = h_u$. We create $v'$ from $v$ as follows: (1) for all $z\in\mathcal{D}$ such that $h_v(z) > h_{u}(z)$, replace $h_v(z) - h_u(z)$ copies of $z$ with some value $\diamond\not\in\mathcal{D}$; (2) for all $z\in\mathcal{D}$ such that $h_v(z) < h_u(z)$, replace $h_u(z)-h_v(z)$ $\diamond$ values with $z$. By \cref{defn:ham-dist}, then, 
    \begin{equation}
    \label{eqn:metric-relate-dham}
    \dham(v',v)\leq \dco(u,v).
    \end{equation}
    
    The resulting vector $v'$ also has the property that, for all $z\in\mathcal{D}$, $h_{v'}(z) = h_u(z)$, so $h_u = h_{v'}$. By \cref{lemma:hists-equal}, $h_u = h_{v'}$ implies that there exists a permutation $\pi\in S_n$ such that $\pi(u) = v'$. Substituting this into \cref{eqn:metric-relate-dham}, we get
    \begin{equation}
    \dham(\pi(u),v)\leq \dco(u,v).
    \end{equation}
    
    We next show that, for all permutations $\pi\in S_n$, we have $\dco(u,v)\leq \dham(\pi(u),v)$, which means that $\dco(u,v) \leq \min_{\pi\in S_n}\dham(\pi(u),v)$. This is true because, for every permutation $\pi\in S_n$, $$\dco(u,v) = \dco(\pi(u), v)\leq \dham(\pi(u), v),$$
    with the equality following from the fact that $h_{\pi(u)} = h_u$.
    
    \item Let $u,v\in\mathcal{D}^n$. Then,
    $$\begin{aligned}
    \dsym(u,v) &= \sum_{z\in \mathcal{D}} |h_u(z) - h_v(z)| \\
    &= \sum_{\substack{z\in \mathcal{D} \text{ s.t.}\\ h_u(z) > h_v(z)}} \left( h_u(z) - h_v(z)\right) \\
    &+ \sum_{\substack{z\in \mathcal{D} \text{ s.t.}\\ h_v(z) > h_u(z)}} \left( h_v(z) - h_u(z)\right)\\
    &= \dco(u,v) + \dco(u,v)\\
    &= 2\cdot \dco(u,v).
    \end{aligned}$$

    \item Let $u,v\in \mathcal{D}^n$, and let $\dham(u,v) = c$. Let $\mathcal{I}$ be the set of all indices $i^*$ such that $u_{i^*} \neq v_{i^*}$. We note that, to make it so that $u_{i^*} = v_{i^*}$, we can perform 1 insertion and 1 deletion to change $u_{i^*}$ to $v_{i^*}$. We note that the cardinality of $\mathcal{I}$ is $c$. Therefore, a total of (at most) $2c$ insertions and deletions need to be performed to change $u$ into $u'$ such that $u'=v$. By the definition of $\did$, then, $\did(u,v)\leq 2c = 2\cdot \dham(u,v)$.

\end{enumerate}

\end{proof}

\fi

\subsection{Sensitivity of Functions}
\label{sec:sensitivity}

Now that we have defined the notion of neighboring datasets, we can define \emph{sensitivity}, a term around which this paper revolves.

\begin{definition}[Sensitivity]
\label{defn:abs-sens}

We define the (global) \emph{sensitivity} of a function $f:\Vect(\mathcal{D}) \rightarrow \mathbb{R}$ with respect to a metric $d$ on $\Vect(\mathcal{D})$ as 

$$
    \Delta_\mathit{d} f = \sup_{
    \substack{
        u,v\in \Vect(\mathcal{D})  \\
        u\simeq_\mathit{d} v
        }
    }|f(u)-f(v)|.
$$
\end{definition}

The sensitivity of $f$ captures how much the output of $f$ can change if we change a single element of the input vector.

\subsubsection{The Path Property}
In this paper, we state our results in terms of neighboring datasets, as is customary in the DP literature. These results, though, readily generalize to datasets at arbitrary distance. This idea is captured by the \emph{path property}.
\begin{definition}[Path Property]
\label{defn:path-property}
    A metric\footnote{We use the term \textit{metric} throughout the paper, although some of our ``metrics'' (e.g., symmetric distance and change-one distance) allow distance 0 for unequal elements and should be considered \emph{pseudometrics} instead.} $d$ is said to satisfy the \textit{path property} if it fulfills the following two conditions:
    \begin{enumerate}
        \item For all datasets $u,v$, we have $d(u,v)\in\mathbb{Z}$.
\        \item If $d(u,v) = d$ then there exists a sequence $u = u^0, u^1, \ldots, u^d = v$ such that $d(u^{i-1}, u^i)\leq 1$.
    \end{enumerate}
\end{definition}

\begin{theorem}[Applying the Path Property]
\label{thrm:apply-path-prop}
For every function $f:{\rm Vec}(\mathcal{D})\to \mathbb{R}$, every dataset metric $d$ that satisfies the path property, and every $u,v\in{\rm Vec}(\mathcal{D})$,
$$|f(u) - f(v) | \leq \Delta_d f \cdot d(u,v).$$
\end{theorem}

\begin{proof}
    The result follows directly from $d$ applications of the triangle inequality to each of the pairs $d(u_{i-1}, u_i)$ given by the path property definition.
\end{proof}

\begin{lemma}
\label{lemma:metrics-path-prop}
The dataset metrics $\dsym, \dco, \dham$, and $\did$ all satisfy the path property.
\end{lemma}

\ifnum\CCSFORMAT=0
The proof can be found in Appendix~\ref{appendix:sec2}.
\fi

\ifnum\CCSFORMAT=1

\begin{proof}

By Definitions~\ref{defn:symm-dist}, \ref{defn:co-dist}, \ref{defn:id-dist}, \ref{defn:ham-dist} and the fact that the image of the histogram function is $\mathbb{N}$ it is clear that all of $\dsym, \dco, \dham$, and $\did$ satisfy Condition 1 of the Path Property.

For Condition 2, we show constructively how to build the sequence $u^0, u^1, \ldots, u^d$ for each of the metrics. Let $u, v$ be datasets such that $d(u, v) = d$ for some $d \in \mathbb{Z}$.

\begin{enumerate}
    \item For $\dsym$, let set $S$ be the symmetric difference of $u$ and $v$. Let $\mathcal{I}_u$ be the set of indices $k$ such that $u_k \in S$, and similarly for $\mathcal{J}_v$. Wlog, assume that $|\mathcal{I}_u| \leq |\mathcal{J}_v|$ (otherwise, swap $u$ and $v$ in what follows). To construct the $u^{(i)}$, we iteratively apply the following two steps. To go from $u^{(i)}$ to $u^{(i+2)}$, first pick one of the indices $k \in \mathcal{I}_u$ and delete it from $u^{(i)}$, which yields vector $u^{(i+1)}$. Next, insert into $u^{(i+1)}$ element $v_j$ such that $j \in \mathcal{J}_v$, never repeating either index $k$ or $j$ in the process. After $2|\mathcal{I}_u|$ steps, to continue constructing the $u^{(i)}$ iteratively we insert elements $v_j$ for all the remaining $j \in \mathcal{J}_v$ one at a time (without repeating any index $j$). Hence, after $|\mathcal{J}_v| - |\mathcal{I}_u|$ more steps, $u^{(i)} = v$.
    \item For $\dco$, to construct each $u^{(i)}$ we iteratively change an element of $u$ which is different from $v$ (when viewed as multisets). It follows directly from the definition of $\dco$ that this procedure requires $\dco(u, v) = d$ steps.
    \item For $\dham$, we apply the same procedure, except that $u$ and $v$ are now viewed as ordered vectors.
    \item For $\did$, we apply the same procedure as in the case of $\dco$, where we do one insert or one delete operation at a time. It then follows directly from the definition of $\did$ that this procedure requires $\did(u, v) = d$ steps.
\end{enumerate}

\end{proof}

\fi

Due to Lemma \ref{lemma:metrics-path-prop}, all of the theorems in this paper will be stated in terms of neighboring datasets, and the path property can be used to generalize these results to apply to any datasets at arbitrary distance $d$.

\begin{lemma}[Convert Sensitivities]
\label{thm:dsymtoid}
\label{thm:dcotoham}

We can convert between sensitivities in the following ways.
\begin{enumerate}
    \item For every function $f: {\rm Vec}(\mathcal{D}) \rightarrow \mathbb{R}$, $\Delta_\mathit{ID} f \leq \Delta_\mathit{Sym} f$.

    \item For every function $f: \mathcal{D}^n \rightarrow \mathbb{R}$, $\Delta_\mathit{Ham} f \leq \Delta_\mathit{CO} f$.
    
    \item For every function $f: \mathcal{D}^n \rightarrow \mathbb{R}$, $\Delta_\mathit{CO} f \leq 2 \Delta_\mathit{Sym}f$.

    \item For every function $f: \mathcal{D}^n \rightarrow \mathbb{R}$, $\Delta_\mathit{Ham} f \leq 2 \Delta_\mathit{ID} f$.

\end{enumerate}

\end{lemma}

\ifnum\CCSFORMAT=0
The proof can be found in Appendix~\ref{appendix:sec2}.
\fi

\ifnum\CCSFORMAT=1

\begin{proof} We use Lemma~\ref{lemma:metric-relate} to prove each part. Let each $f$ below be defined as described in its respective part above.
\begin{enumerate}
    \item By Lemma~\ref{lemma:metric-relate}, we see that, for all $u,v\in \Vect(\mathcal{D})$, we have $\dsym(u,v) \leq \did(u,v)$. Therefore, for all $u\nid u'$, we have $u\nsym u'$. By the definition of sensitivity in Definition~\ref{defn:abs-sens}, this means that $\sid f = \max_{u\nid u'} |f(u)-f(u')| \leq \max_{u\nsym u'} |f(u)-f(u')| = \ssym f.$
    
    \item By Lemma~\ref{lemma:metric-relate}, we see that, for all $u,v\in \mathcal{D}^n$, we have $\dco(u,v) \leq \dham(u,v)$. Therefore, for all $u\nham u'$, we have $u\nco u'$. By the definition of sensitivity in Definition~\ref{defn:abs-sens}, this means that $\sham f = \max_{u\nham u'} |f(u)-f(u')| \leq \max_{u\nco u'} |f(u)-f(u')| = \sco f.$
    
    \item We have:
    $$
    \begin{aligned}
    \sco f &= \max_{u\nco u'}|f(u) - f(u')| \\
    &\leq  \max_{u\nco u'} \ssym f \cdot \dsym(u,u') \quad \text{(\cref{thrm:apply-path-prop})} \\
    &= \max_{u\nco u'}\ssym f \cdot 2 \cdot \dco(u,u') \quad \text{(\cref{lemma:metric-relate})}\\
    &= 2 \ssym f,
    \end{aligned}
    $$
    so $\sco f \leq 2 \ssym f.$
    
    \item We have:
    $$
    \begin{aligned}
    \sham f &= \max_{u\nham u'}|f(u) - f(u')| \\
    &\leq \max_{u\nham u'} \sid f \cdot \did(u,u')  \quad \text{(\cref{thrm:apply-path-prop})}\\
    &\leq \max_{u\nham u'} \sid f \cdot 2 \cdot \dham(u,u') \quad \text{(\cref{lemma:metric-relate})}\\
    &= 2 \sid f,
    \end{aligned}
    $$
    so $\sham f \leq 2 \sid f.$

\end{enumerate}

\end{proof}

\fi

\subsection{Differential Privacy}

We now recall the definition of (pure) differential privacy, also known as $\varepsilon$-DP. Although some of the libraries that we investigate also offer the more general $(\varepsilon, \delta)$-DP, they all offer $\varepsilon$-DP, and we only consider the pure-DP versions of these implementations. Before proceeding with the definition of DP, we define a \emph{mechanism}.

\begin{definition}[Mechanism]
A \emph{mechanism} with input space $\Vect(\mathcal{D})$ and output space $\mathcal{Y}$ is a randomized algorithm $\mathcal{M}: \Vect(\mathcal{D}) \rightarrow \mathcal{Y}$ that on input $u\in\Vect(\mathcal{D})$ outputs a sample from the distribution $\mathcal{M}(u)$ over $\mathcal{Y}$.

\end{definition}

\begin{definition}[Pure Differential Privacy]
    For any $\varepsilon \geq 0$, a mechanism $\mathcal{M}: \Vect(\mathcal{D}) \rightarrow \mathcal{Y}$ is \emph{$\varepsilon$-differentially private with respect to metric $d$} on $\Vect(\mathcal{D})$ if for all subsets $S \subset \mathcal{Y}$ and for all adjacent datasets $u \simeq_d u'$,
    \[
        \Prob[\mathcal{M}(u) \in S] \leq e^{\varepsilon} \cdot \Prob[\mathcal{M}(u') \in S].
    \]
\end{definition}

We note that differential privacy is ``robust to post-processing''. Intuitively, this means that a data analyst -- with any amount of additional knowledge -- cannot take an $\varepsilon$-DP answer and make it less private.

\begin{proposition}[DP is Robust to Post-Processing \cite{dr14}]
\label{prop:post-proc}
Let $\mathcal{M}:\Vect(\mathcal{D})\to \mathcal{Y}$ be a randomized algorithm that is $\varepsilon$-DP. Let $f:\mathcal{Y}\to \mathcal{Z}$ be an arbitrary, randomized mapping. Then $f\circ \mathcal{M}:\Vect(\mathcal{D})\to \mathcal{Z}$ is $\varepsilon$-DP.
\end{proposition}
\begin{proof}
See the proof of Proposition 2.1 in \cite{dr14}.
\end{proof}

\subsubsection{The Laplace Mechanism}

A common way to fulfill DP is to have mechanisms that add noise to a true query response, with the noise scaled to the sensitivity of the underlying function. These mechanisms are called \emph{additive mechanisms}. An elementary example is the Laplace mechanism, which was first shown to fulfill $\varepsilon$-DP in \cite{dmns16}.

\begin{definition}
[The Laplace Mechanism \cite{dmns16}]
    Given a function $f: \Vect(\mathcal{D}) \rightarrow \mathbb{R}$, the \emph{Laplace mechanism for $f$ with scale $\lambda$} is defined as
    \[
        M(u) = f(u) + Y, \textrm{ where } Y \sim \rm{Lap}(\lambda).
    \]
    where $\mathit{Lap}(\lambda)$ denotes the zero-centered continuous probability distribution defined by the density function
    \[
        g_{\lambda}(x) := \dfrac{1}{2 \lambda} \textrm{exp} \Big( - \dfrac{|x|}{\lambda} \Big).
    \]
\end{definition}

\begin{theorem}
\label{thrm:scale-laplace}
The Laplace mechanism for $f$ with scale $\lambda$ is $\varepsilon$-DP with respect to dataset metric $d$ if and only if $\lambda \geq \Delta_d f / \varepsilon$.
\end{theorem}
\begin{proof}
The forward direction is proven in \cite{dmns16}. For the other direction, suppose $\lambda < \Delta_d f / \varepsilon$. Then there are $u\simeq_d u'$ such that $|f(u) - f(u')| > \lambda \varepsilon$. Evaluating the PDFs for $M(u)$ and $M(u')$ at the point $f(u)$ then yields the fraction
$$\frac{\exp(-\frac{|f(u)-f(u)|}{\lambda})} {\exp(-\frac{|f(u')-f(u)|}{\lambda})}
>
\frac{1}{\exp(-\frac{\lambda \varepsilon}{\lambda})}
=
\exp(\varepsilon).
$$
Because the ratio of probability densities is not upper-bounded by $e^{\varepsilon}$, the Laplace mechanism with scale $\lambda < \Delta_d f / \varepsilon$ does not offer $\varepsilon$-DP.
\end{proof}

Therefore, even with a correctly implemented noise addition mechanism, incorrectly calibrated noise will imply that DP will not be fulfilled. In this paper, we demonstrate how the sensitivity is frequently underestimated, which means that the DP condition is not fulfilled even when the mechanism would fulfill DP if its noise were scaled correctly.

\subsection{Integers}
\label{sec:integers}

In this section, we review the facts about integer representation on computers that we require in our proofs.

\subsubsection{Integer Values}

Integers can be \emph{signed} or \emph{unsigned}: the signed integers can hold both positive and negative values, whereas the unsigned integers only hold non-negative values. \cref{table:types} summarizes the data types that we will consider.

\begin{definition}[Unsigned integers]
A \emph{$k$-bit unsigned integer} can take on any value in $$[0,2^{k}-1]\cap\mathbb{Z}.$$
\end{definition}

\begin{definition}[Signed integers]
\label{defn:signed-ints}
A \emph{$k$-bit signed integer} can take on any value in $$[-2^{k-1}, 2^{k-1}-1]\cap\mathbb{Z}.$$

\end{definition}

\subsubsection{Integer Arithmetic}
\label{sec:ints-are-mod-n}

In this paper, we are only concerned with the \emph{addition} of integers. We use two varieties of integer addition: modular addition and saturating addition.

Many libraries and programming languages default to using modular addition when adding values of integer types. Before defining modular addition, we first note that the $k$-bit unsigned integers and $k$-bit signed integers each consist of unique representations of the congruence classes in $\mathbb{Z}/(2^{k})\mathbb{Z}$.

\begin{definition}[Modular Addition on the $k$-bit Integers]
\label{defn:mod-add}

Let $T$ be a (signed or unsigned) type for $k$-bit integers. For all values $x,y$ of type $T$, we define $[x+y]_{2^k}$ as the unique value $z$ such that $z\equiv x+y \pmod{2^k}$ and such that $z$ is of type $T$ (meaning $z\in [0, 2^{k}-1]\cap\mathbb{Z}$ for unsigned $T$, and $z\in [-2^{k-1}, 2^{k-1}-1]\cap\mathbb{Z}$ for signed $T$).

\end{definition}

Saturation arithmetic is a strategy commonly used to prevent integer overflow from occurring. In our attacks and in our proposed solutions, we consider the effects of using saturating addition. We describe its specific behavior below. 

\begin{definition}[Saturation Addition]
\label{defn:sat-add}
Let $T$ be a (signed or unsigned) type for $k$-bit integers. Additionally, let $\max(T)$ represent the maximum representable integer of type $T$, and let $\min(T)$ represent the minimum representable integer of type $T$. Also, let $+$ represent addition in $\mathbb{Z}$, and let $\boxplus$ represent saturation addition. For all values $x,y$ of type $T$, we define $$x\boxplus y= \begin{cases}
      \min(T) & \textrm{if } x+y < \min(T) \\
      x+y & \textrm{if } \min(T) \leq x+y\leq \max(T) \\
      \max(T) & \textrm{if } x+y >   \max(T)
  \end{cases}$$
\end{definition}

\subsection{Floating-Point Representation}\label{sec:floatingpoint}

In this section, we review the facts about floating-point representation that we require in our proofs.

\subsubsection{Floating-Point Values}

According to the IEEE floating-point standard, floating-point numbers are represented with an exponent $E$ and a mantissa $m$ \cite{ieee08}, in a base-2 analog of scientific notation. More precisely,
\begin{definition}[Normal Floating-Point Number]
\label{defn:fp}
A \emph{normal $(\mantlen,\explen)$-bit floating-point number} $z$ is represented as 
$$
    z = (-1)^s \cdot (1.M) \cdot 2^E,
    % n = (-1)^s \cdot 1.m \cdot 2^{E-(2^{\explen-1}-1)}
$$
where
\begin{itemize}
    \item $s\in \{0,1\}$ is used to represent the \emph{sign} of $z$.
    \item $M\in \{0,1\}^\mantlen$ is a $\mantlen$-bit string that represents the part of the \emph{mantissa} to the right of the radix point. That is, $$1.M = 1+\sum_{i=1}^\mantlen M_i 2^{-i}.$$
    
    \item $E\in\mathbb{Z}$ represents the \emph{exponent} of 2. When $\explen$ bits are allocated for representing $E$, then $E\in [-(2^{\explen-1}-2), 2^{\explen-1}-1]\cap \mathbb{Z}$.
    
\end{itemize}
\end{definition}

Note that the number of possibilities for $E$ is $2^\explen - 2$ rather than $2^\explen$. The remaining two choices for an $\explen$-bit $E$ are used to represent the following additional floating-point numbers.

The key characteristic of normal floating-point representation is that representable values are spaced \emph{non-uniformly} throughout the real line \cite{Mironov12}: for a $(\mantlen,\explen)$-bit float, and for $m\in\mathbb{Z}$, successive intervals from 0 to $\infty$ of the form $[2^m,2^{m+1})$ double in length (with the length always corresponding to a power of 2), and each interval contains exactly $2^\mantlen$ representable real values. Analogously, successive intervals from 0 to $-\infty$ are of the form $[-2^m,-2^{m+1})$ and double in length. This implies that, between each successive interval $[ 2^m, 2^{m+1})$ and $[ 2^{m+1}, 2^{m+2})$, and between $[ -2^m, -2^{m+1})$ and $[ -2^{m+1}, -2^{m+2})$, the spacing between adjacent representable floats also doubles. This phenomenon is depicted in Figure~\ref{fig:floats}.

\begin{figure}
\centering
\begin{tikzpicture}[x=0.75pt,y=0.75pt,yscale=-0.5,xscale=0.5]

%Straight Lines [id:da736455358488413] 
\draw    (50,111.23) -- (612,113.23) ;
%Straight Lines [id:da6735512534393959] 
\draw    (50,71.73) -- (50,150.73) ;
%Straight Lines [id:da17920340760580444] 
\draw    (71,89.23) -- (71,131.23) ;
%Straight Lines [id:da9226585011233237] 
\draw    (91,89.23) -- (91,131.23) ;
%Straight Lines [id:da07812729737328827] 
\draw    (111,89.23) -- (111,131.23) ;
%Straight Lines [id:da902576944530296] 
\draw    (131,73.73) -- (131,152.73) ;
%Straight Lines [id:da40049218418530774] 
\draw    (171,91.23) -- (171,133.23) ;
%Straight Lines [id:da4569507239835855] 
\draw    (211,90.23) -- (211,132.23) ;
%Straight Lines [id:da0055974581668627454] 
\draw    (250,89.23) -- (250,131.23) ;
%Straight Lines [id:da2776678792499947] 
\draw    (291,73.73) -- (291,152.73) ;
%Straight Lines [id:da35257709773158186] 
\draw    (371,91.23) -- (371,133.23) ;
%Straight Lines [id:da33162267128076595] 
\draw    (451,91.23) -- (451,133.23) ;
%Straight Lines [id:da9880567365533108] 
\draw    (531,91.23) -- (531,133.23) ;
%Straight Lines [id:da882212737446757] 
\draw    (612,73.73) -- (612,152.73) ;
%Shape: Brace [id:dp7844656589752306] 
\draw   (70,150.23) .. controls (70,153.25) and (71.51,154.76) .. (74.53,154.76) -- (74.53,154.76) .. controls (78.84,154.76) and (81,156.27) .. (81,159.29) .. controls (81,156.27) and (83.16,154.76) .. (87.47,154.76)(85.53,154.76) -- (87.47,154.76) .. controls (90.49,154.76) and (92,153.25) .. (92,150.23) ;
%Shape: Brace [id:dp10286563129041348] 
\draw   (170,150.23) .. controls (170,154.9) and (172.33,157.23) .. (177,157.23) -- (180,157.23) .. controls (186.67,157.23) and (190,159.56) .. (190,164.23) .. controls (190,159.56) and (193.33,157.23) .. (200,157.23)(197,157.23) -- (203,157.23) .. controls (207.67,157.23) and (210,154.9) .. (210,150.23) ;
%Shape: Brace [id:dp49936224169694254] 
\draw   (370,151.23) .. controls (370.06,155.9) and (372.42,158.2) .. (377.09,158.14) -- (400.59,157.85) .. controls (407.26,157.77) and (410.62,160.06) .. (410.67,164.73) .. controls (410.62,160.06) and (413.92,157.69) .. (420.59,157.6)(417.59,157.64) -- (444.09,157.31) .. controls (448.76,157.25) and (451.06,154.89) .. (451,150.23) ;

% Text Node
\draw (43,158) node [anchor=north west][inner sep=0.75pt]   [align=left] {1};
% Text Node
\draw (123,159) node [anchor=north west][inner sep=0.75pt]   [align=left] {2};
% Text Node
\draw (283,159) node [anchor=north west][inner sep=0.75pt]   [align=left] {4};
% Text Node
\draw (603,160) node [anchor=north west][inner sep=0.75pt]   [align=left] {8};

\end{tikzpicture}

\caption{Depiction of the increased spacing between adjacent representable floats. For demonstration purposes, the mantissa is limited to a length of 2 bits. The curly braces illustrate how the distance between adjacent representable floats doubles across successive intervals of the form $[2^k,2^{k+1})$.} \label{fig:floats}

\end{figure}
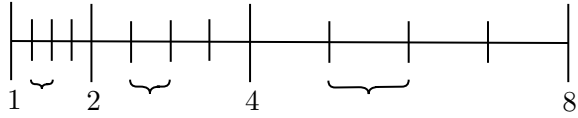

The definition of $\ulp$ provided below enables discussion of the precision available when working with a normal $(\mantlen, \explen)$-bit floating-point number $z$.

\begin{definition}[ULP]
\label{defn:ulp}
We define $\ulp(z)$ as the \emph{unit in the last place} for a normal $(\mantlen,\explen)$-bit floating-point number $z$; i.e., the place-value of the least significant digit of a floating-point number $z$. That is, if $z = (-1)^s \cdot (1.M) \cdot 2^{E}$, then $\ulp(z) = 2^{E-\mantlen}$.

We generalize this notion to apply to all $z\in\mathbb{R}\setminus \{0\}$. For $z$ such that $|z|\in [2^m, 2^{m+1})$ for $m\in\mathbb{Z}$, we define $\ulp(z) = 2^{m-\mantlen} = 2^{ \lfloor \log_2|z| \rfloor - k }$.

\end{definition}

\cref{lemma:ulp-float-test} shows how $\ulp$ can be used to determine whether a value $z\in\mathbb{R}$ can be exactly represented as a normal $(\mantlen, \explen)$-bit float.

\begin{lemma}
\label{lemma:ulp-float-test}
A number $z\in \mathbb{R}$ can be represented exactly as a normal $(\mantlen, \explen)$-bit float if and only if
\begin{enumerate}
    \item $z$ is an integer multiple of $\ulp(z) = 2^{\lfloor \log_2 |z| \rfloor - k}$, and
    
    \item $\lfloor \log_2 |z| \rfloor  \in [-(2^{\explen-1} - 2), 2^{\explen - 1}-1]$. 
\end{enumerate}

\end{lemma}
\begin{proof}
We prove each direction.
\begin{itemize}
    \item $(\Rightarrow)$ Let $z$ be a normal $(\mantlen, \explen)$-bit float. Then, $z$ must be of the form $(1.j) \cdot 2^m = (1+j\cdot 2^{-k}) \cdot 2^m$ for some $j\in\{0,1\}^k$ and $m\in[-(2^{\mantlen-1}-2), 2^{\mantlen-1}-1]$. Equivalently, we must have $z = \mathit{sign}(z)\cdot (2^\mantlen+j)\cdot 2^{m-k}$. Therefore, because $(2^\mantlen+j)$ is an integer, $z$ is an integer multiple of $2^{m-k}$, and $\lfloor \log_2|z| \rfloor = m\in [-(2^{\mantlen-1}-2), 2^{\mantlen-1}-1]$.

    \item $(\Leftarrow)$ Let $m = \lfloor \log_2|z| \rfloor$, and suppose $z$ is an integer multiple of $\ulp(z) = 2^{m-k}$. We can then write $|z| = 2^m + j\cdot 2^{m-k}$ for some $j\in[0,2^k) \cap \mathbb{Z}$. Equivalently, $z = \mathit{sign}(z)\cdot (1+j\cdot 2^{-k})\cdot 2^{m}$.

    The set of the bitstrings of the form $\{0,1\}^k$ is the full set of possible values of $j$, so $(1+j\cdot 2^{-k})$ corresponds directly with a $\mantlen$-bit mantissa $1.j$. We can write $z = \mathit{sign}(z)\cdot (1.j)\cdot 2^m$, for $m\in[-(2^{\explen - 1}-2), 2^{\explen - 1} - 1]\cap \mathbb{Z}$ and $j\in\{0,1\}^k$, so, by definition, $z$ is representable as a normal $(\mantlen,\explen)$-bit float.

\end{itemize}

This completes the proof.
\end{proof}

\begin{definition}[NRF]
\label{defn:nrf}
Let $\nrf(z)$ be to the \emph{nearest (bigger) representable float} for a normal $(\mantlen,\explen)$-bit floating-point number $z$. That is, $\nrf(z)$ is the floating-point value $z'$ such that $|z'| > |z|$ and such that $\forall z'' \neq z$ where $|z''|>|z|$, we have $|z' - z|\leq |z'' - z|$. (For the $z=0$ case, we require that $\ulp(z)$ is positive.)
\end{definition}

\begin{definition}[Subnormal Floats]
\label{note:subnormal-floats}
A \emph{subnormal $(\mantlen,\explen)$-bit float $z$} is represented as
$$
    z = (-1)^s \cdot (0.M) \cdot 2^E,
    % n = (-1)^s \cdot 1.m \cdot 2^{E-(2^{\explen-1}-1)}
$$
where
\begin{itemize}
    \item $s\in \{0,1\}$ is used to represent the \emph{sign} of $z$.
    \item $M\in \{0,1\}^\mantlen$ is a $\mantlen$-bit string that represents the part of the \emph{mantissa} to the right of the radix point. That is, $$0.M = \sum_{i=1}^\mantlen M_i 2^{-i}.$$
    
    \item $E = -(2^{\explen-1}-2)$. Note that $E$ is a constant.
    
\end{itemize}

\end{definition}

\begin{definition}[$\pm \texttt{inf}$]
\label{defn:inf}

%The $(\mantlen,\explen)$-bit floats $\pm \texttt{inf}$ are defined as follows. 
Let $n_\mathit{max}$ and $n_\mathit{min}$ be the largest and smallest $(\mantlen,\explen)$-bit normal floating-point numbers. Then, for arithmetic and rounding purposes, $\texttt{inf} = n_\mathit{max} + \ulp(n_\mathit{max})$ and $\texttt{-inf} = n_\mathit{min} - \ulp(n_\mathit{min})$. This means, for example, that any real number $x \geq n_\mathit{max} + \ulp(n_\mathit{max})$ is rounded to $\texttt{inf}$. We assume the behavior that, for all $z\neq \pm \texttt{inf}$, $\texttt{inf} + z=\texttt{inf}$ and $\texttt{-inf} + z=\texttt{-inf}$.
\end{definition}

\begin{definition}[$(\mantlen, \explen)$-bit Float]
The set of $(\mantlen, \explen)$-bit floats is the union of the set of all normal $(\mantlen, \explen)$-bit floats, all subnormal $(\mantlen, \explen)$-bit floats, and $\pm \texttt{inf}$. Any element of this set is termed a \emph{$(\mantlen, \explen)$-bit float}.
\end{definition}

In the libraries of DP functions that we investigate, floating-point numbers come in two main varieties: single-precision floating-point numbers, which use 32 bits of memory; and double-precision floating-point numbers, which use 64 bits of memory.

\begin{definition}[32-Bit and 64-Bit Floats]
\label{defn:32-bit-fp}
\label{defn:64-bit-fp}
We define the two types of floating-point numbers we encounter.
\begin{enumerate}
\item \emph{Single-precision floating-point} numbers are $(23,8)$-bit floating-point numbers. We denote the set of single-precision floats with the symbol $\mathbb{S}$.

\item \emph{Double-precision floating-point} numbers are $(52,11)$-bit floating-point numbers. We denote the set of double-precision floats with the symbol $\mathbb{D}$.
\end{enumerate}
\end{definition}

As described in Section~\ref{sec:integers}, integers can be \emph{signed} or \emph{unsigned}: the signed integers can hold both positive and negative values, whereas the unsigned integers only hold non-negative values. \cref{table:types} summarizes the data types that we will consider.

\def\arraystretch{1}

\begin{table}[!h]

\centering

\begin{tabular}{|p{0.2\textwidth}|p{0.2\textwidth}|p{0.2\textwidth}|}

\hline

\multirow{2}{*}{\makecell[l]{Integers (Sec.\,\ref{sec:integers})}} & 32-bit signed & 64-bit signed \\

\cline{2-3} & 32-bit  unsigned & 64-bit unsigned \\

\hline

Floats (Sec.\,\ref{sec:floatingpoint}) & 32-bit & 64-bit \\

\hline

\end{tabular}

\caption{Summary of the data types considered in this paper.}\label{table:types}

\end{table}

\subsubsection{Rounding Modes}
\label{sec:rounding-modes}

Because floating-point representations on computers have a limited number of bits, sometimes rounding is necessary. There are many different rounding modes, but the IEEE 754 standard \cite{ieee08} for floating-point arithmetic prescribes using \emph{banker's rounding} as the default rounding mode. Thus in this paper (unless stated otherwise) we will always assume the use of banker's rounding. 

\begin{definition}[Banker's Rounding]
\label{defn:bankers-rounding}
    For $x \in \mathbb{R}$, we define $\br(x)$ to be the floating-point value $z$ such that, for all floating-point values $z'\neq z$, we have $|z-x|\leq |z'-x|$. In the event of a tie (meaning $x$ is equally close to some $z$ and some $z'\neq z$), $x$ is rounded to the value $\in\{z,z'\}$ whose mantissa ends in an even bit.
\end{definition}

There is one more rounding mode that we consider in this paper: \emph{round toward zero}.

\begin{definition}[Round Toward Zero]\label{defn:roundto0}
    For $x \in \mathbb{R}$, we define $\rtz{x}$ to be the floating-point value $z$ such that $|z|<|x|$ and for all floating-point values with $|v'| < |x|$, we have $|v-x|\leq |v'-x|$.
\end{definition}

Lastly, we remark that \emph{fixed-point arithmetic} constitutes an alternative way of representing real numbers on computers, which is simpler than floating-point arithmetic. A fixed-point representation of a fractional number is essentially treated as an integer, which is implicitly multiplied by a fixed scaling factor. Thus, a fixed-point number has a fixed number of digits after the decimal point, whereas a floating-point number allows for a varying number of digits after the decimal point.

\subsubsection{Floating-Point Arithmetic}

In this paper, we are only concerned with the \emph{addition} of floating-point numbers. We use two varieties of floating-point addition: addition using banker's rounding, and addition using round toward zero.

Before providing these definitions, though, we provide a note about the behavior of all floating-point addition.

\begin{note}[Floating-Point Arithmetic Saturates]
\label{note:floats-saturate}
All floating-point addition is \emph{saturating} (see Definition~\ref{defn:sat-add}), with the exception that, as described in Section~\ref{sec:rounding-modes}, the results of computations are floating-point values (which means that rounding occurs if necessary). That is, if the sum $x+y$ would be larger than $+\texttt{inf}$ (respectively, smaller than $-\texttt{inf}$), then the result returned is $+\texttt{inf}$ (respectively, $-\texttt{inf}$).
\end{note}

\begin{definition}[Addition With banker's Rounding]
\label{defn:bround-add}
We define $x\oplus y = \text{BRound}(x+y)$, where $\text{BRound}(z)$ returns $z\in\mathbb{R}$ banker's rounded to a floating-point value. (See Definition~\ref{defn:bankers-rounding} for the definition of banker's rounding.)

\end{definition}

\begin{definition}[Addition With Round Toward Zero]
\label{defn:rtz-add}
We use $\rtz{x + y}$ to denote the result of adding two floating-point values $x$ and $y$ using round toward zero, where $\rtz{z}$ returns $z\in\mathbb{R}$ ``rounded toward zero'' to a floating-point value. (See Definition~\ref{defn:roundto0} for the definition of round toward zero.)

\end{definition}

\begin{lemma}
\label{lemma:floats-exact-add}
For $(\mantlen, \explen)$-bit floats $x,y$, $x+y = {\rm BRound}(x+y) = \rtz{x+y}$ if and only if $x+y$ is a $(\mantlen, \explen)$-bit float.
\end{lemma}
\begin{proof}
This follows immediately from \cref{defn:bround-add,defn:rtz-add}.
\end{proof}

\section{An Overview of DP Libraries}
\label{sec:dp-libraries}

\textbf{Google DP Library.} Google's DP Library \cite{wzldsg19} was first released in September 2019 and remains under active development. It already supports many fundamental functions, such as count, sum, mean, variance, the Laplace mechanism, and the Gaussian mechanism, among others.

The repository (\url{https://github.com/google/differential-privacy}) contains an end-to-end differential privacy framework built on top of Apache Beam, and three DP building block libraries, in C++, Go, and Java, which implement basic noise addition primitives and differentially private aggregations. It also contains a differential privacy accounting library, which is used for tracking privacy budget, and a command line interface for running differentially private SQL queries with ZetaSQL. The GoogleDP library contains documentation concerning sampling algorithms for the Laplace and Gaussian distributions\footnote{\url{https://github.com/google/differential-privacy/blob/main/common_docs/Secure_Noise_Generation.pdf}.} which they claim circumvent problems with na\"ive floating-point implementations. These sampling mechanisms essentially follow the snapping mechanism proposed by Mironov \cite{Mironov12}.

Specifically, this document says, ``A key property that we rely on is that the IEEE floating-point standard guarantees that the results of basic arithmetic operations are the same as if the computation was performed exactly and then rounded to the closest floating-point number.'' Although this is true for a single operation, this is \textit{not} true when several basic arithmetic operations are \emph{iteratively} applied, as it is the case of bounded sum. Issues related to accumulated rounding and non-associativity of floating-point addition can arise when more than a single addition is performed (as happens when computing the bounded sum of a dataset with more than two elements). We demonstrate the severity of these issues in Section~\ref{sec:bounded-sum-computers}.

\smallskip

\textbf{OpenMined PyDP.} A related library is OpenMined's PyDP (\url{https://github.com/OpenMined/PyDP}), which consists of a Python wrapper for Google's DP library. Similarly, they support the fundamental functions of bounded mean, bounded sum, max, min, and median, among others. For now, they only use Laplace noise.

\smallskip

\textbf{SmartNoise/OpenDP.} SmartNoise is a toolkit for differential privacy on tabular data built through a collaboration between Harvard and Microsoft as part of the OpenDP project. It was originally driven by the SmartNoise-Core library of differentially private algorithms, which has since been deprecated and replaced by the OpenDP library (\url{https://opendp.org/}). The OpenDP library has a stringent framework for expressing and validating privacy properties \cite{ghv20}. The library is built in Rust but includes several Python bindings.

OpenDP releases are built by assembling a number of constituent transformations and measurements using operations such as ``chaining'' and ``composition''. As the other libraries, it supports fundamental functions such as count, clamp, mean, bounded sum, and resize, and it also supports other types of noise such as Gaussian noise. Also like the other libraries, it maintains a privacy budget tracker.

\smallskip

\textbf{IBM's \texttt{diffprivlib}.} IBM's differentially private library \cite{hbal19} is implemented in Python and is more geared towards applications of DP in machine learning (\url{https://github.com/IBM/differential-privacy-library}). The library includes a host of mechanisms alongside a number of applications to machine learning and other data analytics tasks.

In addition to implementing fundamental functions such as bounded sum, the library implements DP mechanisms such as the Exponential mechanism, the Geometric mechanism, and the Bingham mechanism. Due to floating-point imprecision, functions in the library sample from the Laplace distribution using the method described in Holohan and Braghin \cite{hb21}. The library also implements tools for computing DP histograms, training naive DP Bayes classifiers, and performing linear regressions. The library is organized into three main modules: mechanisms, models, and tools. Like the other libraries described in this section, it uses a so-called budget accountant to track the privacy budget and calculate total privacy loss using advanced composition techniques.

\smallskip

\textbf{Opacus.} Facebook's DP library \cite{acgbmmtz16} is called Opacus and is heavily geared towards machine learning applications of differential privacy, and so in particular they implement the well-known differentially-private Stochastic Gradient Descent (DP-SGD) function \cite{acgbmmtz16}. Their library (\url{https://github.com/pytorch/opacus}) explicitly enables training PyTorch models with differential privacy.

Opacus allows building image classifiers with DP, training a differentially private LSTM model for name classification, and building a text classifier with differential privacy on BERT, among others. It also maintains a privacy tracking with an accountant. A main difference between Opacus and the rest of the libraries is that it is based on R\'enyi differential privacy \cite{Mironov17}.

\smallskip

\textbf{Chorus.} Chorus \cite{jnhs20} consists of a query analysis and rewriting framework to enforce differential privacy for general-purpose SQL queries. The library (\url{https://github.com/uber-archive/sql-differential-privacy}) was deployed at Uber for its internal analytics tasks as part of the company's efforts to comply with the General Data Protection Regulation (GDPR). Chorus supports integration with any standard SQL database and is focused on the large-scale deployment of DP methods, designed to process datasets consisting of billions of rows. At Uber, Chorus processed more than $10{,}000$ queries per day, and they evaluated $18{,}774$ real-world queries with a database of 300 million rows. Beyond implementing fundamental functions such as sum and the Laplace mechanism, they also implement more complex functions such as the matrix mechanism.

While these are the modern main libraries that we will use to demonstrate our attacks, other relevant DP libraries and systems include PINQ \cite{McSherry10}, Airavat \cite{rsksw10}, Fuzz \cite{hpn11}, and Ektelo \cite{zmkhmm18}.

\section{Bounded Sum}\label{sec:bs}

\subsection{Bounded Sum is a Building Block in DP Libraries}

In this paper, we focus on the bounded sum function, an elementary function which adds the elements of a (numerical) dataset after clipping all values in the dataset to a bounded range $[L, U]$ and then returns the resulting sum. This function is important for two reasons: first, it is simple, and thus allows us to clearly convey the general problems that arise from the use of finite precision arithmetic. Second, it is the building block of many important and complex DP functions, such as functions used for computing averages and performing stochastic gradient descent (SGD). For this reason, bounded sum is a common function in DP libraries, and it appears in Google's DP library, SmartNoise / OpenDP, IBM's \texttt{diffprivlib}, Chorus, and Facebook's Opacus, to name a few. For these reasons, bounded sum is often a function that is chosen as a motivating example when presenting DP libraries, as was done for Chorus \cite{jnhs20} and Airavat \cite{rsksw10}.

Whenever the bounded sum function is implemented, the code implementation must specify the sensitivity of the function. We will now review the sensitivity of this function over $\mathbb{R}$ (i.e., its \emph{idealized} sensitivity) and point to evidence that this is the sensitivity used by DP libraries. In Section~\ref{sec:bounded-sum-computers}, we will show that this idealized sensitivity is not always exhibited by the implementation of the bounded sum algorithm, and we will show how mismatches between the idealized sensitivity can cause blatant privacy violations.

\subsection{The Idealized Bounded Sum Function}

We begin by formally defining the bounded sum function.

\begin{definition}[Clamped Domain]
For a numeric domain or data type $T$, and $L\leq U\in T$ we define $$T_{[L,U]} = \{x\in\mathbb{T}\mid L\leq x\leq U\}.$$

\end{definition}

\begin{definition}[Bounded Sum]
\label{defn:bounded-sum}

For a numeric domain or data type $T$, and $L\leq U\in T$ we define the \emph{bounded sum} function $\bs_{L, U}:\Vect(T_{[L,U]})\to T$ as
$$
    \bs_{L, U}(v) := \sum_{i=1}^{\len(v)} v_i.
$$
We write $\bs_{L, U, n}$ to denote the restriction of $\bs_{L,U}$ to $T_{[L,U]}^n$. 

\end{definition}

\begin{remark}
To differentiate between the ideal and real-world realizations of the bounded sum function, we will later use $\bs$ when referring to the bounded sum function over $\mathbb{R}$ (i.e., with infinite precision) and $\bsfp$ when referring to the computer implementations of it.
\end{remark}

\subsection{Sensitivity of the Idealized Bounded Sum}\label{sec:theoreticalsensitivity}

We now recall the sensitivity for the bounded sum function $\bs_{L,U}$ over $\mathbb{Z}$ and $\mathbb{R}$. Because the sensitivity depends on whether or not $n$ is known, we will prove the sensitivity for each case.

\begin{theorem}[Idealized Sensitivities of $\bs_{L, U}$ and $\bs_{L, U, n}$]
\label{thrm:reals-sens}
\label{thrm:reals-sens-known-n}
The sensitivities of the bounded sum function are the following.
\begin{enumerate}
    \item (Unknown $n$.) $\Delta_\mathit{Sym}\bs_{L,U} = \Delta_\mathit{ID}\bs_{L,U} = \max\{|L|, U\}.$
    \item (Known $n$.) $\Delta_\mathit{CO}\bs_{L,U,n} = \Delta_\mathit{Ham}\bs_{L,U,n} = U-L.$
\end{enumerate}

\end{theorem}

\ifnum\CCSFORMAT=0
The proof can be found in Appendix~\ref{appendix:sec4}.
\fi

\ifnum\CCSFORMAT=1

\begin{proof}
Each part is proven below.
\begin{enumerate}
\item 

Let $u,u'\in\Vect(\mathbb{R}_{[L,U]})$ be two datasets such that $u\simeq_\mathit{Sym} u'$. From the formal definition of a histogram in Definition~\ref{defn:histogram}, we observe that, for all $v\in\Vect(\mathbb{R})$, $\bs_{L,U}(v) = \sum_{i=1}^{\len(v)} v_i = \sum_{z\in\mathbb{R}}h_v(z)\cdot z$, where the last sum is well defined because $h_v(z)\neq 0$ for only finitely many values of $z$. Because $u\simeq_\mathit{Sym} u'$, we know that there is at most one value $z^*$ such that $|h_u(z^*) - h_{u'}(z^*)| = 1$, and that, for all $z\neq z^*$, we have $|h_u(z) - h_{u'}(z)| = 0$. 

We can then write the following expressions.
$$
\begin{aligned} 
\MoveEqLeft\left|\bs_{L,U}(u) - \bs_{L,U}(u') \right| \\ &= \left|\sum_{z\in\mathbb{R}}h_u(z)\cdot z - \sum_{z\in\mathbb{R}}h_{u'}(z)\cdot z \right| \\
&= \left|\sum_{z\in\mathbb{R}\setminus z^*}h_u(z)\cdot z - \sum_{z\in\mathbb{R}\setminus z^*}h_{u'}(z)\cdot z \right|\\
&\quad + \left|h_u(z^*)\cdot z^* - h_{u'}(z^*)\cdot z^*\right| \\
&= 0 + \left|z^* \cdot (h_u(z^*) - h_{u'}(z^*))\right|\\
&\leq |z^*|\\
&\leq \max\{|L|,U\},
\end{aligned}
$$
with the final inequality following from the fact that all values are clamped to the interval $[L,U]$, so the largest difference in sums arises when $z^* = \max\{|L|,U\}$.

By the definition of sensitivity, then, $\Delta_\mathit{Sym}\bs_{L,U} \leq \max\{|L|,U\}$. By Theorem~\ref{thm:dcotoham}, we also have $\Delta_{ID} \bs_{L,U} \leq \max\{|L|,U\}$.

For the lower bound, consider the datasets $u=[0], u'=[0,\max\{|L|,U\}]$. We note that $u\simeq_\mathit{ID} u'$. We also note that $|\bs_{L,U}(u) - \bs_{L,U}(u')| = \max\{|L|,U\}$. This means, then, that $\Delta_\mathit{ID}\bs_{L,U} \geq \max\{|L|,U\}$. By the contrapositive of Theorem~\ref{thm:dcotoham}, then, $\Delta_\mathit{Sym}\bs_{L,U} \geq \max\{|L|,U\}$.

Combining these upper and lower bounds on the idealized sensitivity tells us, then, that $$\Delta_\mathit{Sym}\bs_{L,U} = \Delta_\mathit{ID}\bs_{L,U} = \max\{|L|,U\}.$$

\item

Let $u,u'\in\mathbb{R}^n$ be two datasets such that $u\simeq_\mathit{CO} u'$. By Lemma~\ref{lemma:metric-relate}, there is a permutation $\pi\in S_n$ such that $\pi(u)\simeq_\mathit{Ham} u'$. This means there is at most one index $i^*$ such that $\pi(u)_{i^*}\neq u'_{i^*}$. We can then write the following expressions.
$$\begin{aligned}
    \MoveEqLeft|\bsfp_{L,U,n}(u) - \bsfp_{L,U,n}(u')|\\
    &= |\bsfp_{L,U,n}(\pi(u)) - \bsfp_{L,U,n}(u')|\\
    &=\left| \sum_{i=1}^n (\pi(u)_i) - \sum_{i=1}^n (u'_i) \right|\\
    &=\left| (\pi(u)_{i^*}) - (u'_{i^*}) \right|\\
    &\leq U-L.
\end{aligned}$$
with the final inequality following from the fact that all values are clamped to the interval $[L,U]$, so the largest difference in sums arises when, without loss of generality, $\pi(u)_{i^*} = U$ and $u'_{i^*} = L$.

By the definition of sensitivity, then, $\Delta_\mathit{CO}\bs_{L,U,n} \leq U-L$. By Theorem~\ref{thm:dcotoham}, we have $\Delta_{Ham} \bs_{L,U,n} \leq U-L$.

For the lower bound, consider the datasets $u=[L]$ and $u'= [U]$. Then, $u \simeq_{\mathit{Ham}} u'$ and $|\bs_{L,U,n}(u) - \bs_{L,U,n}(u')| = U-L$. By the contrapositive of Theorem~\ref{thm:dcotoham}, it follows that $\Delta_\mathit{CO}\bs_{L,U,n} \geq U-L$.

Combining these upper and lower bounds on the idealized sensitivity, we then conclude that $$\Delta_\mathit{Ham}\bs_{L,U,n} = \Delta_\mathit{CO}\bs_{L,U,n} = U-L.$$

\end{enumerate}

\end{proof}
\fi

\subsubsection{Uses of the Idealized Sensitivity in Libraries of DP Functions}
\label{sec:ideal-sens-libraries}

Noise is scaled according to the idealized sensitivity of the bounded sum function in following libraries:
\begin{itemize}
    \item Unknown $n$ (\cref{thrm:reals-sens}, Part 1):
    Google's sum function,\footnote{Exact link: \url{https://github.com/google/differential-privacy/blob/f3e565a14b7d48869b650483de897eebc89ad494/cc/algorithms/bounded-sum.h\#L84-L96.}}
    SmartNoise's sum function,\footnote{Exact link: \url{https://github.com/opendp/opendp/blob/92b82ba75fc7e1c376f475c27a1184175796dc22/rust/opendp/src/trans/sum/mod.rs\#L26.}} Opacus,\footnote{Exact link: \url{https://github.com/pytorch/opacus/blob/6a3e9bd99dca314596bc0313bb4241eac7c9a5d0/opacus/optimizers/optimizer.py\#L416}}
    and Airavat \cite[\S 4.1]{rsksw10}.
    
    \item Known $n$ (\cref{thrm:reals-sens}, Part 2):
    IBM \texttt{diffprivlib}'s sum and mean,\footnote{Exact link for sum: \url{https://github.com/IBM/differential-privacy-library/blob/90b319a90414ebf12062887c07e1609f888e1a34/diffprivlib/tools/utils.py\#L687}. Exact link for mean: \url{https://github.com/IBM/differential-privacy-library/blob/90b319a90414ebf12062887c07e1609f888e1a34/diffprivlib/tools/utils.py\#L278-L279}.}
    Google's mean,\footnote{Exact link: \url{https://github.com/google/differential-privacy/blob/4c867aae8ea6a6831d2ac0ab749cd5ae24e047b4/cc/algorithms/bounded-mean.h\#L101}} Chorus \cite[\S 3]{jnhs20},
    and SmartNoise's sized sum function.\footnote{Exact link: \url{https://github.com/opendp/opendp/blob/92b82ba75fc7e1c376f475c27a1184175796dc22/rust/opendp/src/trans/sum/mod.rs\#L52-L53}.}
\end{itemize}

\begin{note}
\label{note:opendp-contrib}
SmartNoise-Core is deprecated, and, though the functions have been ported to OpenDP, they are now kept behind the \texttt{contrib} and \texttt{floating-point} flags with the disclaimer that no guarantees are made about whether these functions offer differential privacy; the other libraries listed above and mentioned in \cref{sec:dp-libraries} do not appear to have such a disclaimer.
\end{note}

As far as we are aware, none of these libraries concretely specify the metric used for defining neighboring datasets, although we infer that all of them implicitly use an unordered distance metric.\footnote{E.g., see the documentation \url{https://github.com/google/differential-privacy/blob/main/common_docs/Differential_Privacy_Computations_In_Data_Pipelines.pdf} for the GoogleDP library.} If the sensitivity used by a library matches the sensitivity with respect to a known $n$ distance, we infer that the library defines neighboring datasets in terms of known $n$ distances, and likewise for unknown $n$ distances.

However, as we develop below in \cref{sec:bounded-sum-computers}, computers work with finite precision arithmetic, and thus implementations of the bounded sum function no longer fulfill the idealized sensitivity of Theorem~\ref{thrm:reals-sens}. We will now analyze the mismatches between the idealized and the implemented sensitivities of the bounded sum function, and show that they can lead to blatant violations of privacy. Thus, whenever the libraries above work with floating-point types, they should use a different sensitivity bound than is used in their code.

\section{Implementing Bounded Sum on Computers}
\label{sec:bounded-sum-computers}

\textbf{Iterative summation.} While the proofs of idealized sensitivity presented in \cref{sec:floatingpoint} are over the reals $\mathbb{R}$, computers cannot work with the reals
and typically use finite-precision integers or the floating point numbers, as defined in \cref{sec:integers} and \cref{sec:floatingpoint}, respectively. Note that, as occurs in many of the libraries discussed in Section \ref{sec:dp-libraries}, the sum is calculated \emph{iteratively}, meaning that the sum is accumulated in a single variable, and that each term in the vector is added to this variable in order of occurrence in the vector. As before, we assume that all elements of the dataset $u$ are numerical elements of the same type. Libraries thus implement a bounded sum function $\bsfp_{L,U}$ in the following iterative way (written in Python-style pseudocode):

\begin{definition}\label{def:bsiterative}
We let $\bsfp_{L, U}:T_{[L,U]}\to T$ denote the iterative bounded sum function on type $T$ as follows:
\begin{lstlisting}[language = Python,frame=single, escapechar=|]
def bounded_sum(u):
    the_sum = 0
    for element in u:
        the_sum += element
    return the_sum
\end{lstlisting}

\end{definition}

We remark that in Section~\ref{sec:accuracy} we consider other orderings than the iterative one when performing the summation of the elements $u_i$ in a database, such as \textit{pairwise summation} (also known as \textit{cascade summation}) \textit{Kahan summation}. We introduce both summation methods in this subsection.

\begin{definition}[Pairwise Summation \cite{h93}]\label{def:bspairwise}
    Given $n$ values $u_i$, for $i = 1, \ldots, n$, the \textit{pairwise summation} method consists of adding the $u_i$ as follows:
    
\begin{lstlisting}[language = Python, escapechar=|,frame=single]
def pairwise_sum(u):
    n = len(u)
    if n == 0:
        return 0
    else if n == 1:
        the_sum = u[1]
    else:
        m = floor(n/2)
        the_sum = pairwise_sum(u[1,...,m]) + 
        pairwise_sum(u[m+1,...,n])
return the_sum
\end{lstlisting}
\end{definition}

That is, we add the $u_i$ terms as follows:
    \[
        v_i = u_{2i-1} + u_{2i}, \quad i = 1 : [n/2]
    \]
    ($y_{[(n+1)/2]} = u_n$ if $n$ is odd). This pairwise summation process is repeated recursively on the $v_i$, $i = 1: [(n+1)/2]$. The final sum is obtained in $\lceil \log_2 n \rceil$ steps.

For example, to add 4 values $u_i$, we would compute $((u_1 + u_2) + (u_3 + u_4))$. Pairwise summation is the default summation mode in NumPy.

Another summation method is that of \textit{Kahan's summation}, which is also known as \textit{compensated summation}. 

\begin{definition}[Kahan Summation \cite{k06}]\label{def:bskahan}
    Given $n$ values $u_i$, for $i = 1, \ldots, n$, \textit{Kahan's summation} method consists of adding the $u_i$ values as follows:
\begin{lstlisting}[language = Python, escapechar=|,frame=single]
def kahan_sum(u):
    the_sum = 0.0
    c = 0.0
    for element in u:
        y = element - c
        t = the_sum + y
        c = (t - the_sum) - y
        the_sum = t
return the_sum
\end{lstlisting}
\end{definition}

We will return to the question of the summation function in Section~\ref{sec:accuracy}. Meanwhile, for the rest of Section~\ref{sec:bounded-sum-computers} we will focus on detailing the attacks. We will present two different styles of attacks which are based on two different principles: those which take advantage of non-associativity, and those which take advantage of rounding errors. 

Additionally, while we do not run attacks on all DP libraries, similar issues to the ones we demonstrate throughout this section can occur in these libraries. In \cref{sec:theoreticalsensitivity} we pointed out where exactly in the code of the different DP libraries one can find the implementation of the sensitivity of the bounded sum function, which in all cases corresponds to the idealized sensitivity (Theorem~\ref{thrm:reals-sens-known-n}). As we have shown throughout this section, the idealized sensitivity does not hold in practice, and hence the same vulnerabilities that we exploit in this section also hold for the rest of the libraries listed in Section~\ref{sec:theoreticalsensitivity}.

\subsection{Overflow Attack on Integers with Modular Addition}
\label{sec:mod-attack}
We prove that the sensitivities $\sid \bsfp_{L,U}$ and $\sham \bsfp_{L,U,n}$ on the integers with modular addition (defined in \cref{defn:mod-add}) are in fact much larger than the idealized sensitivities $\sid \bs_{L,U}$ and $\sham \bs_{L,U,n}$ presented in \cref{sec:theoreticalsensitivity}.

\begin{theorem}[$\sid\bsfp_{L,U}$ and $\sham\bsfp_{L,U,n}$ on the Integers with Modular Addition]
\label{thrm:mod-attack}

Let $\bsfp_{L,U}:\Vect(T)\to T$ be iterative bounded sum with modular addition on the type $T$ of $k$-bit integers. For all $U\geq 1$ and $L=0$, the sensitivity is $\sid \bsfp_{L,U} = 2^k - 1 \geq \sid \bsfp_{L,U} = \max\{ |L|, U \}$. In fact, for all $n$ satisfying the condition
$$n = \left\lceil \frac{\max(T)}{U} \right\rceil + 1,$$
there are datasets $u,v\in T^n_{[L,U]}$ such that $\dham(u,v) = 1$ and $|\bsfp_{L,U,n}(u) - \bsfp_{L,U,n}(v)| = 2^k-1$, so $\sham \bsfp_{L,U,n} = 2^k - 1.$

\end{theorem}

\begin{exampleattack}[64-bit Unsigned Integers]
\label{ea:modular-sum}
Let $L=0$ and $U=2^{47}$. \cref{thrm:mod-attack} gives us datasets $u,v\in T_{[L,U]}^n$ of size
$$n = \left\lceil \frac{\max(T)}{U} \right\rceil + 1 = 2^{17}+1$$
where $$|\bsfp_{L,U,n}(u) - \bsfp_{L,U,n}(v)| = 2^{64} - 1.$$

The idealized sensitivities presented in Section \ref{sec:theoreticalsensitivity}, however, suggest a maximum difference in sums of $\sham\bs_{L,U} = U-L = 2^{47}$. However, $u$ and $v$ actually experience a difference in sums of $2^{64}-1$, more than a factor of $2^{16}$ larger than these idealized sensitivities. This means, then, that a DP mechanism that claims to offer $\varepsilon$-DP here but calibrates its random distribution to the idealized sensitivity would instead offer no better than $2^{16} \varepsilon$-DP.
\end{exampleattack}

\begin{proof}[Proof of \cref{thrm:mod-attack}]

Let $M = \max(T) - (n-2)\cdot U$. We see that $M\in[L,U]$. Consider the datasets
$$ u = [U_1,\ldots, U_{n-2}, M],
v = [U_1, \ldots, U_{n-2}, M, 1].$$

By the definition of modular addition in \cref{defn:mod-add}, we see that $$\bsfp_{L,U}(u) = U_1 + \cdots + U_{n-2} + M = \max(T),$$ while $$\bsfp_{L,U}(v) = \left(\max(T) + 1 \right) \bmod{2^k}$$ such that $\bsfp_{L,U}(v)\in [\min(T),\max(T)]$. This means, then, that $\bsfp_{L,U}(v) = \min(T)$. Therefore, $$|\bsfp_{L,U}(u) - \bsfp_{L,U}(v)| = \max(T) -  \min(T) = 2^k - 1,$$
so, because $\did(u,v) = 1$, $\sid \bsfp_{L,U} = 2^k - 1 \geq \max\{|L|, U\} = \sid \bs_{L,U}$.

Now, consider $u' = [U_1,\ldots, U_{n-2}, M, 0]$. $$\bsfp_{L,U}(u') = U_1 + \cdots + U_{n-2} + M = \max(T),$$ so, because $\dham(u',v) = 1$, $\sham \bsfp_{L,U,n} = 2^k - 1 \geq (U-L) = \sham \bs_{L,U,n}$.

\end{proof}

\subsection{Reordering Attack on Signed Integers with Saturation Addition}\label{sec:non-assoc-ints}

We prove that the sensitivity $\Delta_\mathit{Sym}\bsfp_{L,U}$ on the signed integers with saturation addition (defined in \cref{defn:sat-add}) is in fact much larger than the idealized sensitivity $\Delta_\mathit{Sym}\bs_{L,U}$ presented in \cref{sec:theoreticalsensitivity}.

\begin{theorem}[$\ssym\bsfp_{L,U}$ With Saturation Addition on the Integers]
\label{thrm:sens-ints-sat-unknown}
Let $\bsfp_{L,U}: \Vect(T) \to T$ be iterative bounded sum with saturation addition on the type $T$ of $k$-bit signed integers.
For $U>0$ and $L<0$, the sensitivity is $\Delta_\mathit{Sym}\bsfp_{L,U} = 2^k-1 > \ssym \bs_{L,U}$.

In fact, for all $n$ satisfying the condition
$$n \geq \left\lceil \frac{\max(T)-\min(T)}{|L|} \right\rceil + \left\lceil \frac{\max(T)-\min(T)}{U}\right\rceil,$$
there are datasets $u,v\in T^n_{[L,U]}$ with $\dco(u,v) = 0$ such that $|\bsfp_{L,U}(u) - \bsfp_{L,U}(v)| = 2^k - 1$.
\end{theorem}

\begin{exampleattack}\label{ea:saturation}[32-bit Signed Integers]

Let $L=-2^{14}=-16{,}384$, $U=2^{15}=32{,}768$, and $k=32$. \cref{thrm:sens-ints-sat-unknown} gives us datasets $u,v\in T_{[L,U]}^n$ of size
$$n = \left\lceil \frac{\max(T)-\min(T)}{|L|} \right\rceil + \left\lceil \frac{\max(T)-\min(T)}{U}\right\rceil$$
where $$|\bsfp_{L,U,n}(u) - \bsfp_{L,U,n}(v)| = 2^{32}-1, $$ despite the fact that $d_\mathit{Sym}(u,v)=0$ would suggest that the difference in sums should be 0.

Moreover, the idealized sensitivities presented in Section \ref{sec:theoreticalsensitivity} suggest a maximum difference in sums of $\ssym\bs_{L,U} = \max\{ |L|,|U| \} = U = 2^{15}$ and $\sco\bs_{L,U,n} = (U-L) = 2^{15}+2^{14}$. However, $u$ and $v$ actually experience a difference in sums of $2^{32}-1$, more than a factor of $2^{16}$ larger than these idealized sensitivities. This means, then, that a DP mechanism that claims to offer $\varepsilon$-DP here but calibrates its random distribution to the idealized sensitivity would instead offer no better than $2^{16} \varepsilon$-DP.
\end{exampleattack}

The proof of \cref{thrm:sens-ints-sat-unknown} relies on the fact that saturation addition is not associative, i.e., it is not always the case that $(a\boxplus b) \boxplus c = a\boxplus (b \boxplus c)$. For example,  consider $a=\max(T),b=\max(T),c=-1$. Then $(a\boxplus b) \boxplus c = \max(T) \boxplus -1 = \max(T)-1$, while $a\boxplus (b \boxplus c) = \max(T) \boxplus (\max(T)-1) = \max(T)$. The operation $\boxplus$ is commutative, so if it were associative we could freely reorder the summation without changing the result, which would imply that, for all $u,v$ such that $\dsym(u,v) = 0$, $\bsfp_{L,U}(u) = \bsfp_{L,U}(v)$. However, $\boxplus$ is not associative, and we explore the impacts of that below.

We first state a property of saturation arithmetic.

\begin{property}
\label{prop:sat-add-result}
For any starting summation value $z_0$ of type $T$ and any sequence of values $z_1,\ldots,z_n$ of type $T$, if $z_i\geq 0$ for all $i>0$, $$z_0\boxplus\cdots\boxplus z_n = \min\{z_0+\cdots+ z_n,\max(T)\}.$$
Likewise, if $z_i\leq 0$ for all $i>0$, then $$z_0\boxplus\cdots\boxplus z_n = \max\{z_0+\cdots+ z_n,\min(T)\}.$$
\end{property}
\begin{proof}[Proof of \cref{prop:sat-add-result}]
We first prove the case where, for all $i>0$, $z_i\geq 0$. We use a proof by induction. By the definition of saturation addition in \cref{defn:sat-add}, we see that the base case holds: $z_0 \boxplus z_1 = \min\{z_0 + z_1, \max(T)\}$. Assume that $z_0\boxplus \cdots \boxplus z_{k-1} = \min\{z_0 + \cdots + z_{k-1}, \max(T)\}$. Then, $z_0 \boxplus \cdots \boxplus z_k = \min\{z_0 + \cdots + z_{k-1}, \max(T)\} \boxplus z_k = \min\{z_0 + \cdots + z_{k}, \max(T)\}$. Therefore, the claim holds.

A symmetric argument can be used to prove the claim associated with the case where $z_i\leq 0$ for all $i>0$.
\end{proof}

Note that \cref{prop:sat-add-result} means that, for $L,U\geq 0$ or $L,U\leq 0$, the idealized sensitivity presented in \cref{thrm:reals-sens} is an upper bound on the sensitivity $\ssym\bsfp_{L,U}$ and $\sco\bsfp_{L,U,n}$, and thus the condition $U>0$ and $L<0$ in \cref{thrm:sens-ints-sat-unknown} is essential. We now proceed with the proof of \cref{thrm:sens-ints-sat-unknown}.

\begin{proof}[Proof of \cref{thrm:sens-ints-sat-unknown}]
Let $b=\lceil \frac{\max(T)-\min(T)}{|L|} \rceil$ and let $c=\lceil \frac{\max(T)-\min(T)}{U}\rceil $. Consider the two datasets
$$u=[L_1,\ldots,L_b, U_1,\ldots, U_c], v=[U_1,\ldots,U_c, L_1,\ldots, L_b ],$$
where, for all $i$, $L_i = L$ and $U_i = U$.
We note that $\dsym(u,v) = 0$, so $u\simeq_\mathit{Sym} v$. We now evaluate and compare the results of $\bsfp_{L,U}(u)$ and $\bsfp_{L,U}(v)$.

We first show that $\bsfp_{L,U}(u) = \max(T)$.

We begin by showing that $\bsfp_{L,U}[L_1,\ldots, L_b] =\min(T)$. By Property \ref{prop:sat-add-result}, we know that $L_1\boxplus\cdots \boxplus L_b = \max\{L_1+\cdots+ L_b,\min(T)\}$. We see that $L_1+\cdots+ L_b = b\cdot L \leq \frac{\max(T)-\min(T)}{|L|}\cdot L = -\max(T)+\min(T).$ Therefore, $$\bsfp_{L,U}[L_1,\ldots,L_b] = \max\{-\max(T)+\min(T), \min(T)\} = \min(T).$$

To complete the evaluation of $\bsfp_{L,U}(u)$, we now consider the addition of the $U_i$ terms to the intermediate sum $\bsfp_{L,U}[L_1,\ldots,L_b] = \min(T)$. By Property \ref{prop:sat-add-result}, we know that $\min(T)\boxplus U_1\boxplus\cdots \boxplus U_c = \min\{ \min(T) + U_1+\cdots+ U_c,\max(T) \}$. We see that $U_1+\cdots+ U_c = c\cdot U \geq \frac{\max(T)-\min(T)}{U}\cdot U = \max(T)-\min(T).$ Therefore, by Property \ref{prop:sat-add-result}, $\bsfp_{L,U}(u) = \min\{ \max(T)-\min(T)+\min(T), \max(T) \} = \max(T)$.

We next show that $\bsfp_{L,U}(v) = \min(T)$. By an argument similar to the one used above, $\bsfp_{L,U}[U_1,\ldots,U_c] = \max(T)$; then, when we consider the addition of the $L_i$ terms, we find that $\bsfp_{L,U}(v) = \min(T)$.

We see that $|\bsfp_{L,U}(u)-\bsfp_{L,U}(v)| = \max(T)-\min(T)$. Because $u\simeq_\mathit{Sym} v$, we have $\Delta_\mathit{Sym}\bsfp_{L,U} \geq \max(T)-\min(T) = 2^k-1$. Additionally, by \cref{defn:signed-ints}, $U\leq 2^{k-1}-1$ and $|L|\leq 2^{k-1}$, so we necessarily have $\ssym\bsfp_{L,U} = 2^k-1 > \max\{ |L|, U\} = \ssym\bs_{L,U}$.

The results described above will hold for any pair of datasets $$u'=[L_1,\ldots,L_{b'}, U_1,\ldots, U_{c'}], v'=[U_1,\ldots,U_{c'}, L_1,\ldots, L_{b'} ]$$ with $b'\geq b$ and $c'\geq c$. We note that $\dco(u',v') = 0$. Therefore, for all $$n = b'+c' \geq \left\lceil \frac{\max(T)-\min(T)}{|L|} \right\rceil + \left\lceil \frac{\max(T)-\min(T)}{U}\right\rceil,$$
$\sco \bsfp_{L,U,n} = 2^{k}-1$. By \cref{defn:signed-ints}, $U\leq 2^{k-1}-1$ and $|L|\leq 2^{k-1}$, so we necessarily have $\sco\bsfp_{L,U,n} = 2^k-1 \geq U-L = \sco\bs_{L,U,n}$.
\end{proof}

\subsection{Reordering Attack on Floats}
\label{sec:attack-floats-assoc}

We now provide a generalized counterexample showing that the idealized sensitivities $\Delta_\mathit{Sym}\bs_{L,U}$ and $\Delta_\mathit{CO}\bs_{L,U,n}$ in Section \ref{sec:theoreticalsensitivity} do not apply to the floating-point numbers. Similarly to how the example in \cref{sec:non-assoc-ints} leverages the non-associativity of saturating addition on integers, this counterexample leverages the non-associativity of addition on floats.

Because the sensitivities $\ssym \bsfp_{L,U}$ and $\sco \bsfp_{L,U,n}$ on the floats are neither easy to calculate nor useful, we only prove lower bounds on the sensitivities to show that  $\ssym \bsfp_{L,U}\neq \ssym \bs_{L,U}$ and $\sco \bsfp_{L,U,n}\neq \sco \bs_{L,U,n}$. We prove useful upper bounds on sensitivities for alternative implementations of floating-point summation in \cref{sec:splitsumrp}.

\begin{theorem}
\label{prop:non-assoc-floats}
Let $\bsfp_{L,U} : \Vect(T) \to T$ be iterative bounded sum on the type $T$ of $(\mantlen,\explen)$-bit floats. Let $j,a,d\in\mathbb{Z}$ satisfy $-(2^{(\explen-1)}-2)\leq j \leq 2^{(\explen-1)}-3-\mantlen$, $a\geq 0$, $d>0$, and $a+d \leq k+1$. For $L = 2^j$, $U=2^{j+d}$, and $n = 2^{(\mantlen + 1) - a} + 2^{(\mantlen + 1) - d}$, there are datasets $u,v\in T^n_{L,U}$ with $\dsym(u,v) = \dco(u,v) = 0$ such that
$$|\bsfp_{L,U}(u) - \bsfp_{L,U}(v)|\geq 2^{(\mantlen + 1) - (a+d)}\cdot U .$$
In particular,
$$\frac{\ssym \bsfp_{L,U}}{\ssym \bs_{L,U}}\geq 2^{(\mantlen + 1) - (a+d)}$$
and
$$\frac{\sco \bsfp_{L,U,n}}{\sco \bs_{L,U,n}}\geq 2^{(\mantlen + 1) - (a+d)}\cdot \left(\frac{U}{U-L}\right).$$

\end{theorem}
To maximize the blow-up in sensitivity, we should minimize $a+d$, namely by setting $a=0$ and $d=1$, which yields a blow-up of $2^k$, albeit on datasets of size $n=3\cdot 2^k$. To obtain smaller datasets, we should maximize both $a$ and $d$, subject to the constraint $a+d\leq k+1$. For example, setting $a = d = k/2$, we obtain a blow-up by a factor of $2$ (or a factor of $2\cdot (\frac{U}{U-L})$ in the $\sco$ case) on datasets of size $n = 4\cdot 2^{k/2}$.

\begin{exampleattack}[32-bit floats]
\label{ea:32-nonassoc}
Here, $\mantlen = 23$ and $\explen = 8$. We set $a=0$, $d=1$, and $ j=0$. We note that these settings satisfy the conditions of \cref{prop:non-assoc-floats}, which lead to $L=1$, $U=2$, and $n = 2^{24} + 2^{23}$. This dataset construction results in a difference in sums of $2^{24-1}\cdot U = 2^{24}$. This is a factor of $2^{24}$ larger than the sensitivity $\sco \bs_{L,U,n} = (U-L) = 1$, and it is a factor of $2^{23}$ larger than the sensitivity $\ssym \bs_{L,U} = \max\{|L|,U\} = 2$.

This means, then, that a DP mechanism that claims to offer $\varepsilon$-DP but calibrates its random distribution to the idealized sensitivity would instead offer no better than $2^{24} \varepsilon$-DP when dealing with the change-one distance notion of neighboring datasets, and no better than $2^{23} \varepsilon$-DP when dealing with the symmetric distance notion of neighboring datasets.
\end{exampleattack}

\begin{exampleattack}[64-bit floats]
\label{ea:64-nonassoc}
Here, $\mantlen = 52$ and $\explen = 11$. We set $a=26$, $d=26$, and $ j=-26$. We note that these settings satisfy the conditions of \cref{prop:non-assoc-floats}, which leads to $L=2^{-26}$, $U=1$, and $n = 2^{28}$. This dataset construction results in a difference in sums of $2\cdot U = 2$. This is more than a factor of $2$ larger than the sensitivity $\sco \bs_{L,U,n} = (U-L) = 1-2^{-26} < 1$, and it is a factor of $2$ larger than the sensitivity $\ssym \bs_{L,U} = \max\{|L|,U\} = 2$.

This means, then, that a DP mechanism that claims to offer $\varepsilon$-DP but calibrates its random distribution to the idealized sensitivity would instead offer no better than $2 \varepsilon$-DP when dealing with the change-one distance notion of neighboring datasets, and no better than $2 \varepsilon$-DP when dealing with the symmetric distance notion of neighboring datasets.
\end{exampleattack}

\begin{proof}[Proof of \cref{prop:non-assoc-floats}]

Let $b = 2^{(\mantlen+1)-a}$ and $c = 2^{(\mantlen+1)-d}$, and consider the datasets $$u = [L_1,\ldots,L_b,U_1,\ldots,U_c], v = [U_1,\ldots,U_c,L_1,\ldots,L_b]$$ where, for all $i$, $L_i = L$ and $U_i = U$. Note that $\dsym(u,v) = \dco(u,v) = 0$.

We first show that the first $b-1$ terms of $u$ can be added exactly; that is, we show that $\bsfp_{L,U}(u) = \bs_{L,U}(u) = b\cdot L + c\cdot U$.

We begin by showing that $$\bsfp_{L,U}[L_1,\ldots, L_b] = \bs_{L,U}[L_1,\ldots, L_b] = b\cdot L.$$ We take the approach of showing that all intermediate sums computed in the calculation of $\bsfp_{L,U}[L_1,\ldots, L_b]$ can be represented exactly as $(\mantlen,\explen)$-bit floats, which (with the exception of $b\cdot L$) is done by first showing that all of these intermediate sums are multiples of their $\ulp$ and then applying \cref{lemma:ulp-float-test}.

We observe that $\{L,\ldots, L\cdot (b-1)\}$ is the set of all intermediate sums that result from calculating $\bs_{L,U}[L_1,\ldots, L_{b-1}]$ and note that $b\cdot L = 2^{j+k+1-a}$. For all $x\in \{L,\ldots, L\cdot (b-1)\}$,
$$\begin{aligned}
\ulp(x)&\leq \ulp(L\cdot (b-1))\\
& < \ulp(2^{j+\mantlen + 1-a})\\
&= 2^{j+1-a}\\
&\leq 2^{j+1}.
\end{aligned}
$$
All such $x$ are multiples of $L=2^j$, so they are necessarily multiples of $\ulp(x)$. By \cref{lemma:ulp-float-test}, all of these values can be represented exactly as $(\mantlen,\explen)$-bit floats. We also note that $b\cdot L = 2^{j+k+1-a}$ can be represented exactly as a $(\mantlen,\explen)$-bit float (by our constraint on $j$), so by \cref{lemma:floats-exact-add}, $\bsfp_{L,U}[L_1,\ldots, L_b] = \bs_{L,U}[L_1,\ldots, L_b] = b\cdot L$.

We now complete the evaluation of $\bsfp_{L,U,n}(u)$. We use \cref{lemma:ulp-float-test} to show that all of the intermediate sums of $\bs_{L,U}(u)=b\cdot L + c\cdot U$ are exactly representable as $(\mantlen,\explen)$-bit floats, which then shows that $\bsfp_{L,U}(u) = \bs_{L,U}(u)=b\cdot L + c\cdot U$.

We consider the addition of the $U_i$ terms to the intermediate sum $b\cdot L$. We note that $b\cdot L = 2^{j+k+1-a}$ is an integer multiple of $U = 2^{j+d}$ since $a+d\leq k+1$. Thus each intermediate sum $x\in \{ b\cdot L, b\cdot L + U, \ldots,  b\cdot L + c\cdot U \}$ is a multiple of $U$. Moreover,
$$\begin{aligned}
\ulp(x)&\leq \ulp(b\cdot L + c\cdot U)\\
&= \ulp(2^{(\mantlen + 1) + j - a} + 2^{(\mantlen + 1) + j })\\
&\leq 2^{j+1},
\end{aligned}$$
unless $a=0$ and $x= b\cdot L + c\cdot U$. Thus $U = 2^{j+d}$ is a multiple of $\ulp(x)$ and hence $x$ is representable as a normal $(\mantlen, \explen)$-bit float. For the case $a = 0$ and $x = b\cdot L + c\cdot U$, we note that $b\cdot L + c\cdot U = 2^{j+\mantlen + 2}$, which is a representable $(\mantlen,\explen)$-bit float by the condition on $j$.

We next evaluate $\bsfp_{L,U,n}(v)$. We use \cref{lemma:ulp-float-test} to show that 
\[
    \bsfp_{L,U}[U_1,\ldots, U_c] = \bs_{L,U}[U_1,\ldots, U_c]=c\cdot U
\]
We observe that $\{U,\ldots, c\cdot U\}$ is the set of intermediate sums that result from calculating the value $\bs_{L,U}[U_1,\ldots, U_c]$ and note that $c\cdot U = 2^{k+1+j}$. For all $x\in \{U,\ldots, c\cdot U\}$, $\ulp(x) = 2^{j+1-m}$ for some integer $m \geq 0$. All $x\in \{U,\ldots, c\cdot U\}$ are multiples of $U=2^{j+d}$ with $d>0$, so they are necessarily multiples of $\ulp(x) \leq \ulp(c\cdot U) = 2^{j+1}$. By \cref{lemma:ulp-float-test}, all of these values can be represented exactly as $(\mantlen,\explen)$-bit floats, so by \cref{lemma:floats-exact-add}, $\bsfp_{L,U}[U_1,\ldots, U_c] = \bs_{L,U}[U_1,\ldots, U_c] = c\cdot U$.

We now complete the evaluation of $\bsfp_{L,U}(v)$. We use \cref{lemma:ulp-float-test} to show that the intermediate sums in the calculation of $\bs_{L,U}(v)$ are not exactly representable as $(\mantlen,\explen)$-bit floats, and we show how to calculate the result $\bsfp_{L,U}(v)$. 

We note that $c\cdot U + L_1 = 2^{\mantlen+1+j} + 2^j \in [2^{\mantlen+1+j}, 2^{\mantlen+2+j})$, so $\ulp(c\cdot U + L_1) = 2^{j+1}$. However, $c\cdot U + L_1 = 2^{\mantlen+1+j} + 2^j$ is not a multiple of $2^{j+1}$ so by \cref{lemma:ulp-float-test}, $c\cdot U + L_1$ is not exactly representable as a $(\mantlen,\explen)$-bit float. $c\cdot U + L_1$ lies exactly halfway between the consecutive $(\mantlen,\explen)$-bit floats $2^{\mantlen+1+j}$ and $2^{\mantlen+1+j} + 2^{j+1}$. Because the mantissa associated with $2^{\mantlen+1+j}$ ends in an even bit, by banker's rounding (\cref{defn:bankers-rounding}), $c\cdot U \oplus L_1 = 2^{\mantlen+1+j}$. For all $i$, then, the summation $c\cdot U \oplus L_i = 2^{(\mantlen+1)+ j} = c\cdot U $. Therefore, $\bsfp_{L,U}(v) = c\cdot U.$

Thus, $$|\bsfp_{L,U}(u)-\bsfp_{L,U}(v)| = b\cdot L = 2^{(\mantlen+1)-a}\cdot 2^j = 2^{(\mantlen + 1) - (a+d)}\cdot U ,$$
so $\ssym \bsfp_{L,U} \geq 2^{(\mantlen + 1) - (a+d)}\cdot U$ and $\sco \bsfp_{L,U,n} \geq 2^{(\mantlen + 1) - (a+d)}\cdot U$. 

We compare this result to the idealized sensitivities presented in \cref{thrm:reals-sens}: $\ssym \bs_{L,U} = \max\{|L|,U\} = U$ and $\sco \bs_{L,U,n} = U-L$. We see that $$\frac{\ssym \bsfp_{L,U}}{\ssym \bs_{L,U}}\geq 2^{(\mantlen + 1) - (a+d)}$$
and
$$\frac{\sco \bsfp_{L,U,n}}{\sco \bs_{L,U,n}}\geq 2^{(\mantlen + 1) - (a+d)}\cdot \left(\frac{U}{U-L}\right).$$

\end{proof}

\subsection{Rounding-Based Attack on Floats}
\label{sec:round-attack-floats} 

In the proof of \cref{prop:non-assoc-floats}, $\dsym(u,v) = 0$, but $\dham(u,v)$ is large. Now we show that the sensitivity can be unexpectedly large even when neighboring datasets are defined in terms of an ordered distance metric such as $\dham$ or $\did$.

\begin{theorem}[$\sham \bsfp_{L,U,n}\neq \sham \bs_{L,U,n}$]
\label{prop:rounding-ham-floats}

    Let $\bsfp_{L,U,n}: T^n \to T$ be iterative bounded sum on the type $T$ of $(\mantlen, \explen)$-bit floats. Let $j,k\in\mathbb{Z}$ satisfy $1<j< (\mantlen+1)/2$ and $-2^{(\explen-1)}+2\leq m \leq 2^{\explen-1}-1-j$.
    For $n=2^j+1$, $L=(1+2^{-(\mantlen+1)+j})\cdot 2^{m}$, and $U = L + 2^{m-\mantlen}$, there are datasets $u,v\in T^n_{[L,U]}$ with $\dham(u,v) = 1$ such that
    $$|\bsfp_{L,U,n}(u) - \bsfp_{L,U,n}(v)| \geq 2^{j+m-\mantlen} = 2^j \cdot (U-L).$$
    In particular, $$\frac{\sham \bsfp_{L,U,n}}{\sham \bs_{L,U,n}} \geq (n-1).$$
    
\end{theorem}

\begin{exampleattack}[64-bit floats]
\label{ea:ham-round-once}

Here, $\mantlen = 52$ and $\explen = 11$. We set $j=4$ and $m=-1$. We note that these settings satisfy \cref{prop:rounding-ham-floats}, which lead to $L = (1+2^{-49})/2$, $U = (1+2^{-49} + 2^{-53})/2$, and $n = 2^4 + 1 = 17$. The result is a difference in sums of $2^{4-1-52} = 2^{-49}$. This is a factor of $2^4$ larger than $\Delta_\mathit{Ham}\bs_{L,U,n} = (U-L) = 2^{-53}$. 

This means, then, that a DP mechanism that claims to offer $\varepsilon$-DP but calibrates its random distribution to the idealized sensitivity would instead offer no better than $ 16 \varepsilon$-DP in the known $n$ case. More generally, no better than $2^{j}\varepsilon$-DP is offered for $1<j<26.5$.

\end{exampleattack}

\begin{proof}[Proof of \cref{prop:rounding-ham-floats}]

Let $b=2^j$ and $m=0$, and consider the datasets $$u = [L_1,\ldots,L_b,U], v = [L_1,\ldots,L_b,L_{b+1}],$$
where, for all $i$, $L_i = L$. 
Note that $\dham(u,v) = 1$.

We begin by showing that $$\bsfp_{L,U}[L_1,\ldots, L_b] = \bs_{L,U}[L_1,\ldots, L_b] = b\cdot L.$$ We take the approach of showing that all intermediate sums computed in the calculation of $\bsfp_{L,U}[L_1,\ldots, L_b]$ can be represented exactly as $(\mantlen,\explen)$-bit floats, which is done by first showing that all of these intermediate sums are multiples of their $\ulp$ and then applying \cref{lemma:ulp-float-test}.

We observe that $\{L,\ldots, L\cdot (b-1)\}$ is the set of intermediate sums that result from calculating $\bs_{L,U}[L_1,\ldots, L_{b-1}]$. Let $\LL$ denote this set $\{L,\ldots, (b-1)\cdot L\}$. We note that
$$ (b-1)\cdot L = b\cdot (L-1) - L + b < b = 2^j,$$
where the inequality holds because $b\cdot (L-1) = 2^{2j-(\mantlen+1)}\leq 1 < L$. Therefore, $\ulp((b-1)\cdot L) = 2^{j-k-1}$. For all $x\in \LL$, then, $$\ulp(x)\leq \ulp(L\cdot (b-1)) = 2^{-(\mantlen+1)+j}.$$

Next, we show that for all $x \in \LL$, $L = 1+2^{-(\mantlen+1)+j}$ is a multiple of $\ulp(x)$. We note that $\ulp(x) = 2^{-(\mantlen+1)+j}$ divides $2^{-(\mantlen+1)+j}$ and, because we require $j<\frac{\mantlen+1}{2}$, we note that $\ulp(x) = 2^{j-k-1-a}$ also divides $1=2^0$. Therefore, $L$ is a multiple of $\ulp(x)$ for all $x\in \LL$.

All $x\in\LL$ are multiples of $L$, so they are necessarily multiples of $\ulp(x)$. By \cref{lemma:ulp-float-test}, all $x\in \LL$ can be represented exactly as $(\mantlen,\explen)$-bit floats. We also note that $b\cdot L = 2^j+2^{-(k+1)+2j} = (1+2^{-(k+1)+j})\cdot 2^j$ can be represented exactly as a $(\mantlen,\explen)$-bit float, so by \cref{lemma:floats-exact-add}, $\bsfp_{L,U}[L_1,\ldots, L_b] = \bs_{L,U}[L_1,\ldots, L_b] = b\cdot L$.

We now evaluate $\bsfp_{L,U,n}(u)$. We use \cref{lemma:ulp-float-test} to show that $\bs_{L,U,n}(u)=b\cdot L + U$ is not exactly representable as a $(\mantlen,\explen)$-bit float, and we show how to calculate the result $\bsfp_{L,U,n}(u)$.

As shown above, the sum of the first $b$ terms is $$\bsfp_{L,U}[L_1,\ldots, L_b] = b\cdot L.$$ Thus, $$\bsfp_{L,U}(u) = b\cdot L \oplus U.$$ We note that
$$b\cdot L + U = 2^j+2^{-(k+1)+2j} + 1+2^{-(\mantlen+1)+j} + 2^{-\mantlen} \in [2^j, 2^{j+1}),$$
so $\ulp(b\cdot L + U) = 2^{j-k}$. However, $b\cdot L + U$ is not an integer multiple of $2^{j-k}$, so by \cref{lemma:ulp-float-test}, $b\cdot L + U$ is not exactly representable as a $(\mantlen,\explen)$-bit float. $b\cdot L + U$ lies between the consecutive $(\mantlen,\explen)$-bit floats $2^{j-\mantlen}(2^\mantlen + 2^{j-1} + 2^{\mantlen-j})$ and
$2^{j-\mantlen}(2^\mantlen + 2^{j-1} + 2^{\mantlen-j} + 1)$. The sum $b\cdot L + U$ is closer to the larger value, $2^{j-\mantlen}(2^\mantlen + 2^{j-1} + 2^{\mantlen-j} + 1)$, so by banker's rounding as defined in \cref{defn:bankers-rounding}, the result is $\bsfp_{L,U,n}(u) = 2^{j-\mantlen}(2^\mantlen + 2^{j-1} + 2^{\mantlen-j} + 1)$.

We next evaluate $\bsfp_{L,U,n}(v)$. We again use \cref{lemma:ulp-float-test} to show that $\bs_{L,U,n}(v)=b\cdot L + L$ is not exactly representable as a $(\mantlen,\explen)$-bit float, and we show how to calculate the result $\bsfp_{L,U,n}(v)$.

As shown above, the sum of the first $b$ terms is $$\bsfp_{L,U}[L_1,\ldots, L_b] = b\cdot L.$$ Thus, $\bsfp_{L,U}(v) = b\cdot L \oplus L$. We note that $$b\cdot L + L = 2^j+2^{-(k+1)+2j} + 1+2^{-(\mantlen+1)+j} \in [2^j, 2^{j+1}),$$ so $\ulp(b\cdot L + L) = 2^{j-k}$. However, $b\cdot L + L$ is not a multiple of $2^{j-k}$ (it is a multiple of $2^{j-k-1}$) so by \cref{lemma:ulp-float-test}, $b\cdot L + L$ is not exactly representable as a $(\mantlen,\explen)$-bit float. $b\cdot L + L$ lies exactly between the consecutive $(\mantlen,\explen)$-bit floats $2^{j-\mantlen}(2^\mantlen + 2^{j-1} + 2^{\mantlen-j})$ and
$2^{j-\mantlen}(2^\mantlen + 2^{j-1} + 2^{\mantlen-j} + 1)$. Because the mantissa associated with the smaller value $2^{j-\mantlen}(2^\mantlen + 2^{j-1} + 2^{\mantlen-j})$ ends in an even bit, by banker's rounding, $b\cdot L \oplus L = 2^{j-\mantlen}(2^\mantlen + 2^{j-1} + 2^{\mantlen-j})$, so $\bsfp_{L,U,n}(u) = 2^{j-\mantlen}(2^\mantlen + 2^{j-1} + 2^{\mantlen-j})$.

We see that $|\bsfp(u)-\bsfp(v)| = 2^{j-\mantlen}$.

From the bounds $-2^{(\explen-1)}+2\leq m \leq 2^{\explen-1}-1-j$, we see that $m$ only affects the exponent of floating point values in the proof and does not affect whether values are representable as $(\mantlen,\explen)$-bit normal floats or the direction of rounding. All values in the proof above, then, can be multiplied by $2^m$, so, for all $m$ such that $-2^{\explen-1}+2\leq m \leq 2^{\explen-1}-1-j$, $\sham \bsfp_{L,U,n} \geq 2^{j+m-\mantlen}$.

We observe that $U-L = 2^{m-\mantlen}$, $\sham \bs_{L,U,n} = U-L = 2^{m-\mantlen}$, and $n-1 = 2^j$. Therefore, $$\frac{\sham \bsfp_{L,U,n}}{\sham \bs_{L,U,n}} \geq (n-1).$$

\end{proof}

\begin{definition}[Compensated Summation]
\label{defn:comp-sum}
A \emph{compensated summation} algorithm, denoted $\mathit{CS}:\Vect(T)\to T$, attempts to track the error that is accumulated in the computation of an iterative summation, and uses this error-tracking to adjust its final answer accordingly. One limitation of a compensated summation algorithm that we assume here is that its output type must be the same as its input type. An idealized version of a compensated summation algorithm, then, computes the exact sum, and then rounds or casts this value to some output of type $T$.
\end{definition}

\begin{remark}[Resistance to Compensated Summation]
\label{rem:resist-kahan-thesis}
The same increased sensitivity seen in the proof of Proposition \ref{prop:rounding-ham-floats} also applies when Kahan summation \cite{h93} (or any type of compensated summation as defined in \cref{defn:comp-sum}) is used. This is because the inaccuracy in the final answer occurs due to an inability to represent the idealized sum, rather than due to an accumulation of rounding error throughout the summation. (Rounding only occurs in the addition of the final term.) We see, then that the idealized sensitivity of $\sco \bs_{L,U,n} = (U-L)$ and $\sham \bs_{L,U,n} = (U-L)$ is fundamentally unsalvageable under the assumption that the data type of elements in the input vector is the same as the data type of the output of $\bsfp_{L,U,n}$.
\end{remark}

\begin{remark}[Resistance to Pairwise Summation]
\label{rem:resist-pair}
The same increased sensitivity seen in the proof of Proposition \ref{prop:rounding-ham-floats} also applies when pairwise summation is used. This is because the dataset $u$ is of length $2^k+1$. As a result, the final term is only added once the sum of the rest of the dataset has been calculated, so the same analysis presented in the proof of Proposition \ref{prop:rounding-ham-floats} applies when calculating the sum by way of pairwise summation.
\end{remark}

\subsection{Repeated Rounding Attack on Floats I}
\label{sec:accum-round-floats}

The counterexample described in \cref{sec:round-attack-floats}, which shows that we can have $\sham\bsfp_{L,U,n}\gg \sham \bs_{L,U,n}$, only takes advantage of rounding that happens at the very last summation step, when we add the last elements $u_n$ and $v_n$ in the datasets $u$ and $v$. Since there last elements are bounded in absolute value by $U$, the rounding changes each result by at most $\pm U$, so $\bsfp_{L,U}(u) - \bsfp_{L,U}(v)|\leq 2U$. The reason for which we get a large blow-up in sensitivity is that $\sham \bs_{L,U,n} = (U-L)$, which can be much smaller than $U$. To get a large blow-up relative to $\sid \bs_{L,U} = U$, we need to take advantage of rounding throughout the summation. Here, we show that when $u$ and $v$ differ in their first element, rounding errors can accumulate throughout the summation, yielding $\sid \bsfp_{L,U} \gg \sid\bs_{L,U}$.

\begin{theorem}
\label{prop:rounding-id-floats}
    Let $\bsfp_{L,U} : \Vect(T_{[L,U]}) \to T$ be iterative bounded sum on the type $T$ of $(\mantlen, \explen)$-bit floats. Let $j,m\in\mathbb{Z}$ satisfy $0<j\leq \mantlen$ and $-(2^{\explen-1}-2) \leq m \leq 2^{\explen-1}-2-j-\mantlen$.
    
    Then for $L = 0$, $U=2^{\mantlen + m}$, and $n = j\cdot 2^{\mantlen} + 2$, there are datasets $u,v\in\Vect(T)$ where $len(u) = n$ and $\len(v) = n-1$ such that $\did(u,v) = 1$ and
    $$|\bsfp_{L,U}(u) - \bsfp_{L,U}(v)| \geq 2^{\mantlen+j+m}.$$
    In particular,
    $$\frac{\sid \bsfp_{L,U}}{\sid \bs_{L,U}} \geq 2^j.$$
\end{theorem}

\begin{exampleattack}[32-bit floats]
\label{ea:32-ham-round-accum}
Here, $\mantlen = 23$ and $\explen = 8$. We set $j=15$ and $m=0$. We note that these settings satisfy the conditions of \cref{prop:rounding-id-floats}, which lead to $L=0$, $U=2^{23}$, and $\len(u) = 15\cdot 2^{23} + 2$ and $\len(v) = 15\cdot 2^{23} + 1$. This dataset construction results in a difference in sums of $2^{15}\cdot U = 2^{38}$. This is a factor of $2^{15}$ larger than the sensitivity $\sid \bs_{L,U} = \max\{|L|, U \} = 2^{23}$.

This means, then, that a DP mechanism that claims to offer $\varepsilon$-DP but calibrates its random distribution to the idealized sensitivity would instead offer no better than $2^{15} \varepsilon$-DP when dealing with the insert-delete distance notion of neighboring datasets.
\end{exampleattack}

\ifnum\CCSFORMAT=0
The proof of Theorem~\ref{prop:rounding-id-floats} can be found in Appendix \ref{appendix:sec5}.
\fi

\ifnum\CCSFORMAT=1

\begin{proof}[Proof of \cref{prop:rounding-id-floats}]

Let $f(x)$ be defined as $x\mapsto(2^{x} + 2^{x-\mantlen})\cdot 2^{m}$, and consider the datasets
$$    \begin{aligned}
u ={} & [2^{\mantlen+m},2^{\mantlen+m}, f(0)_{1}, \ldots, f(0)_{2^{\mantlen}}, \\
      & f(1)_{2^{\mantlen}+1},\ldots,f(1)_{2\cdot 2^{\mantlen}},\ldots, \\
      & f(j-1)_{(j-1)\cdot 2^{\mantlen}+1},\ldots, f(j-1)_{j\cdot 2^{\mantlen}}],
      \end{aligned}
$$
and $v$ such that $v_i = u_{i+1}$ for all $i$ (so the first term of $u$ is not in $v$). Note that $\did(u,v) = 1$. Additionally, we consider $m=0$ until the conclusion of the proof.

We first evaluate $\bsfp_{L,U}(u)$. We begin by seeing that the sum of the first two terms can be computed exactly: $2^\mantlen \oplus 2^\mantlen = 2^{\mantlen + 1}$.

We next use \cref{lemma:ulp-float-test} to show that the sum of the first three terms cannot be computed exactly: $2^{\mantlen+1} + f(0) = 2^{\mantlen+1} + 1 + 2^{-\mantlen}\in [2^{\mantlen+1}, 2^{\mantlen+2})$, so $\ulp(2^{\mantlen+1} + f(0)) = 2$. However,  $2^{\mantlen+1} + f(0)$ is not a multiple of 2, so by \cref{lemma:ulp-float-test}, $2^{\mantlen+1} + f(0)$ cannot be represented exactly as a $(\mantlen,\explen)$-bit float. This means that we must use banker's rounding (described in \cref{defn:bankers-rounding}) to determine the value of $2^{\mantlen+1} \oplus f(0)$. The value $2^{\mantlen+1} + f(0)$ is between the adjacent $(\mantlen,\explen)$-bit floats $2^{\mantlen+1} + 0$ and $2^{\mantlen+1} + 2$. It is closer to $2^{\mantlen+1} + 2$, so we have $2^{\mantlen+1} \oplus f(0) = 2^{\mantlen+1} + 2$.

Having determined the result of adding $f(0)_1$ to the intermediate sum $2^{k+1}$, we now determine the result of adding the remaining $f(0)$ values. By reasoning similar to the logic used above, we find that for every addition of $f(0)_i$, adding $f(0) = 1 + 2^{-\mantlen}$ has the effect of adding $2$. Therefore, $2^{\mantlen+1}\oplus f(0)_1\oplus \cdots \oplus f(0)_{2^{\mantlen}} = 2^{\mantlen+2}$.

We use \cref{lemma:ulp-float-test} to show that none of the remaining additions are exact, and we show the rules of banker's rounding affect the computation of $\bsfp_{L,U}(u)$. By \cref{lemma:ulp-float-test}, because $\ulp(2^{\mantlen+2}) = 4$ and $2^{\mantlen+2} + f(1)$ is not a multiple of 4, we must use banker's rounding (described in \cref{defn:bankers-rounding}) to determine the value of $2^{\mantlen+2} \oplus f(1)$. Reasoning similar to the logic above shows that adding each of the $f(1) = 2+2^{-\mantlen+1}$ terms has the effect of adding $4$, so $2^{\mantlen+2}\oplus f(1)_{\mantlen+1}\oplus \cdots \oplus f(1)_{2\cdot 2^{\mantlen}} = 2^{\mantlen+3}$. Such logic follows for all of the floating point additions, resulting in an overall sum of $\bsfp_{L,U}(u) = 2^{\mantlen+j+1}$.

We now evaluate $\bsfp_{L,U}(v)$. The general structure of the proof is similar, but the resulting sums are different due to different effects of banker's rounding.

We next use \cref{lemma:ulp-float-test} to show that the sum of the first three terms cannot be computed exactly: $2^{\mantlen} + f(0) = 2^{\mantlen} + 1 + 2^{-\mantlen}\in [2^{\mantlen}, 2^{\mantlen+1})$, so the $\ulp(2^{\mantlen} + f(0)) = 2^0$. However,  $2^{\mantlen+1} + f(0)$ is not a multiple of 1,
% (it is a multiple of $2^{-\mantlen})$
so by \cref{lemma:ulp-float-test}, $2^{\mantlen} + f(0)$ cannot be represented exactly as a $(\mantlen,\explen)$-bit float. This means that we must use banker's rounding (described in \cref{defn:bankers-rounding}) to determine the value of $2^{\mantlen} \oplus f(0)$. The value $2^{\mantlen} + f(0)$ is between the adjacent $(\mantlen,\explen)$-bit floats $2^{\mantlen} + 0$ and $2^{\mantlen} + 1$. It is closer to $2^{\mantlen} + 1$, so we have $2^{\mantlen} \oplus f(0) = 2^{\mantlen} + 1$.

Having determined the result of adding $f(0)_1$ to the intermediate sum $2^{k}$, we now determine the result of adding the remaining $f(0)$ values. By reasoning similar to the logic used above, we find that for every addition of $f(0)_i$, adding $f(0) = 1 + 2^{-\mantlen}$ has the effect of adding $1$ (note in the calculation of $\bsfp_{L,U}(u)$ above that adding $f(0) = 1 + 2^{-\mantlen}$ has the effect of adding $2$ -- this is a critical difference). Therefore, $2^{\mantlen}\oplus f(0)_1\oplus \cdots \oplus f(0)_{2^{\mantlen}} = 2^{\mantlen+1}$.

We now use \cref{lemma:ulp-float-test} to show that none of the remaining additions are exact, and we show the rules of banker's rounding affect the computation of $\bsfp_{L,U}(v)$. By \cref{lemma:ulp-float-test}, because $\ulp(2^{\mantlen+1}) = 2$ and $2^{\mantlen+1} + f(1)$ is not a multiple of 2, we must use banker's rounding (described in \cref{defn:bankers-rounding}) to determine the value of $2^{\mantlen+1} \oplus f(1)$. Reasoning similar to the logic above shows that adding each of the $f(1) = 2+2^{-\mantlen+1}$ terms has the effect of adding $2$, so $2^{\mantlen+1}\oplus f(1)_{\mantlen+1}\oplus \cdots \oplus f(1)_{2\cdot 2^{\mantlen}} = 2^{\mantlen+2}$. Such logic follows for all of the floating point additions, resulting in an overall sum of $\bsfp_{L,U}(u) = 2^{\mantlen+j}$.

We see that $|\bsfp(u)-\bsfp(v)| = 2^{\mantlen+j}$.

From the bounds $-(2^{\explen-1}-2) \leq m \leq 2^{\explen-1}-2-j-\mantlen$, we see that $m$ only affects the exponent of floating point values in the proof and does not affect whether values are representable as $(\mantlen,\explen)$-bit normal floats or the direction of rounding. All values in the proof above, then, can be multiplied by $2^m$, so, for all $m$ such that $-(2^{\explen-1}-2) \leq m \leq 2^{\explen-1}-2-j-\mantlen$, $\sid\bsfp_{L,U} \geq 2^{\mantlen+j+m}$.

We observe that $U = 2^{\mantlen+m}$ and $\sid \bs_{L,U} = \max\{ |L|, U \} = U$. Therefore, $$\frac{\sid \bsfp_{L,U}}{\sid \bs_{L,U}} \geq 2^j.$$

\end{proof}

\fi

\subsection{Repeated Rounding Attack on Floats II}
\label{sec:accum-round-floats-ii}

We now show that we can have an even larger blow-up in sensitivity relative to $\sid\bs_{L,U} = U$ as a function of the dataset length $n$ when we consider datasets consisting of both positive and negative values.

\begin{theorem}
\label{thrm:quadratic-round-attack}
Let $\bsfp_{L,U}:\Vect(T_{[L,U]})\to T$ be iterative bounded sum on the type $T$ of $(\mantlen, \explen)$-bit floats. Additionally, let $a,j\in\mathbb{Z}$ satisfy $2\leq j < k$, $-(2^{\explen-1}-2)\leq a$, $-(2^{\explen-1}-2) + 1 + \mantlen\leq j+a\leq 2^{\explen-1}-1$ (these conditions ensure that values are representable as $(\mantlen, \explen)$-bit floats), $n = 2^j$, $m = n/2$. Then, for
$$U = 2^a, L  = -\left(\frac{U\cdot m}{2^\mantlen}\right)\cdot \left(\frac{1}{2} - \frac{1}{2^\mantlen}\right),$$
there are datasets $u,v\in\Vect(T_{[L,U]})$ with $\did(u,v) = 1$ such that
$$|\bsfp_{L,U}(u) - \bsfp_{L,U}(v)| \geq \frac{n^2}{2^{\mantlen + 3}} \cdot U + U.$$
In particular,
$$\frac{\sid \bsfp_{L,U}}{\sid \bs_{L,U}} = \frac{n^2}{2^{\mantlen+3}} + 1.$$
\end{theorem}

\begin{exampleattack}\label{ea:rr2}[64-bit Floats]
Here, $\mantlen = 52$ and $\explen = 11$. We set $j=30$ and $a = 0$. We note that these settings satisfy the conditions of \cref{thrm:quadratic-round-attack}, which lead to $L=2^{-23}\cdot -(\frac{1}{2} + 2^{-52})$, $U=1$, and $\len(u) = 2^{30}$ and $\len(v) = 2^{30}-1$. This dataset construction results in a difference in sums of $2^{5} + U = 33$. This is a factor of $33$ larger than the sensitivity $\sid \bs_{L,U} = \max\{|L|, U \} = 1$.

This means, then, that a DP mechanism that claims to offer $\varepsilon$-DP but calibrates its random distribution to the idealized sensitivity would instead offer no better than $33 \varepsilon$-DP when dealing with the insert-delete distance notion of neighboring datasets.
\end{exampleattack}

\ifnum\CCSFORMAT=0
The proof of Theorem~\ref{thrm:quadratic-round-attack} can be found in Appendix~\ref{appendix:sec5}.
\fi

\ifnum\CCSFORMAT=1

\begin{proof}[Proof of \cref{thrm:quadratic-round-attack}]

Let $$x = \left(\frac{U\cdot m}{2^\mantlen}\right)\cdot \left(\frac{1}{2} + \frac{1}{2^\mantlen}\right),$$ and consider the datasets
$$u = [U_1,\ldots, U_{m-1}, U_{m}, x_1, L_2, x_3, L_4,\ldots, x_{m-1}, L_m]$$
and
$$v = [U_1,\ldots, U_{m-1}, x_1, L_2, x_3, L_4,\ldots, x_{m-1}, L_m].$$
where, for all $i$, $L_i = L$ and $U_i = U$. 
Note that $u$ and $v$ are equivalent with the exception that $v$ contains $m-1$ copies of $U$ rather than $m$ copies of $U$, so $\did(u,v) = 1$.

We first show that $\bsfp_{L,U}[U_1,\ldots,U_{m-1}] = (m-1)\cdot U$ and $\bsfp_{L,U}[U_1,\ldots,U_{m}] = m\cdot U$. We take the approach of showing that all intermediate sums computed in the calculation of $\bsfp_{L,U}[U_1,\ldots, U_m]$ can be represented exactly as $(\mantlen,\explen)$-bit floats, which is done by first showing that all of these intermediate sums are multiples of their $\ulp$ and then applying \cref{lemma:ulp-float-test}.

We observe that $\{U,\ldots, m\cdot U \}$ is the set of intermediate sums that result from calculating $\bs_{L,U}[U_1,\ldots, U_m]$. Let $\UU$ denote this set $\{U,\ldots, m\cdot U\}$. We note that $m\cdot U = 2^{a+j-1}$, so $\ulp(m\cdot U) = 2^{a+j-1-\mantlen}$. For all $x\in \UU$, then, $$\ulp(x)\leq \ulp(m\cdot U) = 2^{a+j-1-\mantlen}.$$
Note that $j<k$, so $2^{a+j-1-\mantlen}$ necessarily divides $U=2^a$. All such $x$ are multiples of $U=2^a$, so they are necessarily multiples of $\ulp(x)$. By \cref{lemma:ulp-float-test}, all $x\in \UU$ can be represented exactly as $(\mantlen,\explen)$-bit floats, so by \cref{lemma:floats-exact-add}, $\bsfp_{L,U}[U_1,\ldots,U_{m-1}] = (m-1)\cdot U$ and $\bsfp_{L,U}[U_1,\ldots,U_{m}] = m\cdot U$.

We now evaluate $\bsfp_{L,U,n}(u)$. We use \cref{lemma:ulp-float-test} to show that $\bs_{L,U,n}(u)=m\cdot U + x$ is not exactly representable as a $(\mantlen,\explen)$-bit float, and we show how to calculate the result $\bsfp_{L,U,n}(u)$.

As shown above, the sum of the first $m$ terms is $$\bsfp_{L,U}[U_1,\ldots, U_m] = m\cdot U.$$ Thus, $\bsfp_{L,U}[U_1,\ldots, U_m,x] = m\cdot U \oplus x$. We note that
$$m\cdot U + x \in [m\cdot U, m\cdot U + \ulp(m\cdot U)].$$ Because $m\cdot U + x$ is closer to $m\cdot U + \ulp(m\cdot U)$, $m\cdot U \oplus x = m\cdot U + \ulp(m\cdot U)$. We now consider the addition of $L$, $m\cdot U + \ulp(m\cdot U) \oplus L$. We see that $$m\cdot U + \ulp(m\cdot U) + L \in [m\cdot U, m\cdot U + \ulp(m\cdot U)],$$ and $m\cdot U + \ulp(m\cdot U) + L$ is closer to $m\cdot U + \ulp(m\cdot U)$, so $m\cdot U + \ulp(m\cdot U) \oplus L = m\cdot U + \ulp(m\cdot U) \oplus L$. Similar logic applies for all $x,L$, in which adding $L$ has the effect of adding $0$ and adding $x$ has the effect of adding $\ulp(m\cdot U)$ for a total sum of $\bsfp_{L,U}(u) = m\cdot U + \frac{m^2 \cdot U}{2^{\mantlen+1}}$.

We now evaluate $\bsfp_{L,U,n}(v)$. As shown above, the sum of the first $m-1$ terms is $\bsfp_{L,U}[U_1,\ldots, U_{m-1}] = (m-1)\cdot U$. Thus, $\bsfp_{L,U}[U_1,\ldots, U_{m-1},x] = (m-1)\cdot U \oplus x$. We note that
$$(m-1)\cdot U + x \in [(m-1)\cdot U, (m-1)\cdot U + \ulp((m-1)\cdot U)].$$ Because $(m-1)\cdot U + x$ is closer to $(m-1)\cdot U + \ulp((m-1)\cdot U)$, $(m-1)\cdot U \oplus x = (m-1)\cdot U + \ulp((m-1)\cdot U)$. We now consider the addition of $L$, $(m-1)\cdot U + \ulp((m-1)\cdot U) \oplus L$.
We see that
$$\begin{aligned}
\MoveEqLeft(m-1)\cdot U + \ulp((m-1)\cdot U) + L\\
&\in [(m-1)\cdot U, (m-1)\cdot U + \ulp((m-1)\cdot U)],
\end{aligned}$$
and $(m-1)\cdot U + \ulp((m-1)\cdot U) + L$ is closer to $(m-1)\cdot U$, so $(m-1)\cdot U + \ulp((m-1)\cdot U) \oplus L = (m-1)\cdot U$. Similar logic applies for all $x,L$, in which adding $x$ and $L$ in succession has the effect of adding $0$ for a total sum of $\bsfp_{L,U}(v) = (m-1)\cdot U$.

Therefore, $$|\bsfp_{L,U}(u) - \bsfp_{L,U}(v)| = \frac{m^2\cdot U}{2^{\mantlen + 1}} + U = \frac{n^2\cdot U}{2^{\mantlen + 3}} + U,$$ so, because $\did(u,v) = 1$, $\did\bsfp_{L,U} = \frac{m^2\cdot U}{2^{\mantlen + 1}} + U$. We observe that $\sid\bs_{L,U} = \max\{|L|,U\} = U$, so $$\frac{\sid \bsfp_{L,U}}{\sid \bs_{L,U}} = \frac{n^2}{2^{\mantlen+3}} + 1.$$

\end{proof}

\fi

\subsection{Concrete Implementations of the Attacks on DP Libraries}
\label{sec:concreteattack}

In this section, we show how malicious data analysts can exploit the fact that, as shown in \cref{sec:ideal-sens-libraries}, all libraries with a (supposedly) DP bounded sum function scale its noise according to the idealized sensitivity $\sd \bs_{L,U}$ rather than to the implemented sensitivity $\sd \bsfp_{L,U}$.

As shown in the previous parts of \cref{sec:bounded-sum-computers}, there are many cases for which $\sd \bs_{L,U} < \sd \bsfp_{L,U}$. By \cref{thrm:scale-laplace}, the fact that, for a user-specified $\varepsilon$, the scale parameter $\lambda$ is set to $\lambda = \sd \bs_{L,U}/\varepsilon < \sd \bsfp_{L,U}/\varepsilon$ means that $\varepsilon$-DP is not fulfilled by functions which claim to offer $\varepsilon$-DP.

The fact that $\varepsilon$-DP is not fulfilled indicates that there could be opportunities for malicious data analysts to get more precise information about datasets for a given value of $\varepsilon$ than should be allowed. In this section, we demonstrate how a malicious data analyst can take advantage of these errors and use the outputs of allegedly DP functions to perform membership-inference attacks and reconstruction attacks. Our attacks use the Python wrappers for OpenDP / SmartNoise and for IBM's \texttt{diffprivlib}; and we use OpenMined's PyDP Python wrapper for our interactions with Google's DP libraries.

\begin{definition}[$\mathcal{M}^*_{\varepsilon}$]
For all of the examples below, let $\mathcal{M}^*_{\varepsilon}:\Vect(\mathcal{D})\to\mathbb{R}$ represent the library's attempt to implement the function being used in the attack (either a mean function or a summation function) such that it fulfills $\varepsilon$-DP. In all of the libraries we explore, this is done using a discrete version of the Laplace mechanism, often implemented using the Snapping Mechanism described in \cite{Mironov12}.
\end{definition}

In our attacks, we perform the following process: First, we create (sometimes slightly modified) versions of the datasets in the example attacks from the previous parts of \cref{sec:bounded-sum-computers}, and we execute supposedly $\varepsilon$-DP functions on these datasets. We then apply post-processing to the result which, by \cref{prop:post-proc}, will not make the result any less differentially private. We show, though, that the results we achieve after post-processing would be very unlikely to occur for the specified value of $\varepsilon$. Finally, we show how the fact that $\varepsilon$-DP is not fulfilled can be used to recover information about the datasets that would have been protected by an $\varepsilon$-DP function.

\subsubsection{Implementing \cref{ea:modular-sum}}\label{sec:ibmattack}
Only IBM's diffprivlib was susceptible to \cref{ea:modular-sum}. This is because, although many libraries -- perhaps intentionally, perhaps naively -- compute the bounded sum function with modular arithmetic (i.e., overflow is allowed) and then also add noise such that overflow is allowed, a solution which we prove to be correct in \cref{sec:modularnoise}. IBM's diffprivlib, on the other hand, appears to compute the bounded sum with modular arithmetic and then add noise using saturation addition.

\begin{attacklibrary}[IBM's diffprivlib; Signed 32-bit Integers]
\label{al:ibm-overflow}

We create datasets as described in \cref{sec:mod-attack}, and we set $U=2^{24}$, $L=0$, and $n = 2^7$. Let $\mathcal{M}^*_{\varepsilon}$ be the bounded sum function as implemented by IBM's diffprivlib for $U$ and $L$.

We run $\mathcal{M}^*_{\varepsilon}(u)$ and $\mathcal{M}^*_{\varepsilon}(v)$, where $\mathcal{M}^*_{\varepsilon}$ is the mechanism used by the library for computing $\bsfp_{L,U}$ with DP and where $\varepsilon = 0.5$. We define the post-processing function $f:\mathbb{R}\to\{0,1\}$ as follows: $f(x) = 1$ if and only if $x > 0$; this value was chosen since it lies between the maximum and minimum values.

Running each of $f(\mathcal{M}^*_{\varepsilon}(u))$ and $f(\mathcal{M}^*_{\varepsilon}(v))$ $10{,}000$ times gives the following results.

\begin{figure}[H]
\begin{center}
    \begin{tabular}{c|c c }
     & $f(\mathcal{M}^*_{\varepsilon}(u))$ & $f(\mathcal{M}^*_{\varepsilon}(v))$ \\ \hline
     0 & 0 & 10{,}000 \\  
     1 & 10{,}000 &	0
    \end{tabular}
    \caption{``$0.5$-DP'' results from IBM's diffprivlib (with post-processing).}
    \label{fig:ibm-overflow}
\end{center}
\end{figure}

By the fact that $u\nid v$, the definition of $0.5$-DP, and the fact that DP is robust to post-processing per \cref{prop:post-proc}, we would expect to have
$$\frac{\Pr[f(\mathcal{M}^*_{\varepsilon}(v)) = 0]}{\Pr[f(\mathcal{M}^*_{\varepsilon}(u)) = 0]} \leq e^{0.5} < 1.65$$
and
$$\frac{\Pr[f(\mathcal{M}^*_{\varepsilon}(u)) = 1]}{\Pr[f(\mathcal{M}^*_{\varepsilon}(v)) = 1]} \leq e^{0.5} < 1.65.
$$

Let $p_u = \Pr[f(\mathcal{M}^*_{\varepsilon}(u)) = 1]$ and $p_v = \Pr[f(\mathcal{M}^*_{\varepsilon}(v)) = 1]$, and suppose that $\frac{p_u}{p_v}\leq e^\varepsilon$; that is, suppose the queries did satisfy $\varepsilon$-DP. Then,
$$\begin{aligned}
    (p_u)^{10{,}000} \cdot (1-p_v)^{10{,}000} &\leq (e^\varepsilon p_v \cdot (1-p_v))^{10{,}000}\\
    &\leq \left( \frac{e^\varepsilon}{4} \right)^{10{,}000}.
\end{aligned}$$

The results in \cref{fig:ibm-overflow}, then, would occur with probability $\leq 2^{-12{,}000}$ if $\mathcal{M}^*_\varepsilon$ were truly an $\varepsilon$-DP function for $\varepsilon = 0.5$. These results, then, support the calculation in \cref{ea:modular-sum} that $\mathcal{M}^*_{\varepsilon}$ does not offer $\varepsilon$-DP for $\varepsilon = 0.5$.

\end{attacklibrary}

\subsubsection{Implementing \cref{ea:32-nonassoc,ea:64-nonassoc}}
\label{sec:implement-nonassoc-attacks}

We show how a malicious data analyst can use the results shown in \cref{sec:attack-floats-assoc} to perform membership-inference attacks on datasets. We work with the ``unknown $n$'' summation functions provided in \cref{sec:ideal-sens-libraries} (where ``unknown $n$'' means that the function is designed to offer DP under the $\dsym$ notion of neighboring datasets).

\begin{motivateattack}

We create the datasets as described in \cref{ea:32-nonassoc} if we are working with 32-bit floats or as described in \ref{ea:64-nonassoc} if we are working with 64-bit floats, with the exception that we set $v = [v_1,\ldots,v_{\len(v)-1}]$ (i.e., the last element is removed).

We now describe a setting in which the example attacks can be used to perform membership-inference attacks.

Suppose we know that we either have dataset $u'$ which contains individual $i$ and has the property $h_{u'}=h_u$, or dataset $v'$ which does not contain individual $i$ and has the property $h_{v'}=h_v$. (Note that the presence or absence of individual $i$ is why $\dsym(u',v') = 1$.)

Additionally, suppose the library of functions we are working with provides us with the bounded sum function used in Example Attacks~\ref{ea:32-nonassoc} and \ref{ea:64-nonassoc}, and suppose that the DP library also provides us with the ability to reorder datasets based on some condition. Since this implementation of bounded sum indicates that the creator of the library thinks floating point addition is associative, the creator should have no issues with offering this functionality.

Then, we can conditionally reorder the dataset based on the absence or presence of $i$. If $i$ is in the dataset (meaning we have $u'$), some permuted dataset $\pi(u') = u$ is returned where $\pi\in S_{\len(u)}$ and where $\pi(u')$ has the property that all rows with value $L$ appear before any row with value $U$. Likewise, if $i$ is not in the dataset (meaning we have $v'$), some permuted dataset $\sigma(v') = v$ is returned where $\sigma\in S_{\len(v)}$ and where $\sigma(v')$ has the property that all rows with value $U$ appear before any row with value $L$.

Finally, we perform an ``$\varepsilon$-DP'' bounded sum for some $\varepsilon$ on the reordered dataset, which is either $\pi(u') = u$ or $\sigma(v') = v$. Using the ideas and evidence described in the attacks below, we will be able to identify whether we are working with $u'$ or $v'$ with higher probability than $\varepsilon$ should allow -- and, correspondingly, we will be able to determine whether $i$ is in the dataset.

(Although this attack still relies on having specific values in the dataset, a library that provides access to a ``join'' function like the one offered in \cite{wzldsg19} and a ``resize'' function like the one offered in \cite{opendp} could enable a malicious data analyst to manipulate a dataset such that this attack could work in less specific settings.)

\end{motivateattack}

\begin{attacklibrary}[SmartNoise; 32-bit Floats]
We create the datasets as described in \cref{ea:32-nonassoc}, with the exception that we set $v = [v_1,\ldots,v_{\len(v)-1}]$ (i.e., the last element is removed). The result, then, is that $\dsym(u,v) = 1$, so $u\nsym v$. From \cref{ea:32-nonassoc}, we expect to have $|\bsfp_{L,U}(u) - \bsfp_{L,U}(v)| = b\cdot L = 2^{24}$, despite the fact that $\ssym \bs_{L,U} = \max\{|L|,U\} = 2$.

We run $\mathcal{M}^*_{\varepsilon}(u)$ and $\mathcal{M}^*_{\varepsilon}(v)$, where $\mathcal{M}^*_{\varepsilon}$ is the mechanism used by the library for computing $\bsfp_{L,U}$ with DP and where $\varepsilon = 0.5$. We define the post-processing function $f:\mathbb{R}\to\{0,1\}$ as follows: $f(x) = 1$ if and only if $x\geq 25{,}165{,}824$ (this value was chosen since it lies halfway between $\bsfp_{L,U}(u)$ and $\bsfp_{L,U}(v)$).

Running each of $f(\mathcal{M}^*_{\varepsilon}(u))$ and $f(\mathcal{M}^*_{\varepsilon}(v))$ 100 times gives the following results. (We only run these computations 100 times each due to the time needed to compute these summations of large datasets.)

\begin{figure}[H]
\begin{center}
    \begin{tabular}{c|c c }
     & $f(\mathcal{M}^*_{\varepsilon}(u))$ & $f(\mathcal{M}^*_{\varepsilon}(v))$ \\ \hline
     0 & 0 & 100 \\  
     1 & 100 &	0
    \end{tabular}
\end{center}
    \caption{``$0.5$-DP'' results from SmartNoise (with post-processing).}
    \label{fig:opendp-attack-nonassoc-32}
\end{figure}
By the fact that $u\nsym v$, the definition of $0.5$-DP, and the fact that DP is robust to post-processing per \cref{prop:post-proc}, we would expect to have
$$\frac{\Pr[f(\mathcal{M}^*_{\varepsilon}(v)) = 0]}{\Pr[f(\mathcal{M}^*_{\varepsilon}(u)) = 0]} \leq e^{0.5} \leq 1.65$$
and
$$
\frac{\Pr[f(\mathcal{M}^*_{\varepsilon}(u)) = 1]}{\Pr[f(\mathcal{M}^*_{\varepsilon}(v)) = 1]} \leq e^{0.5} \leq 1.65.
$$

Let $p_u = \Pr[f(\mathcal{M}^*_{\varepsilon}(u)) = 1]$ and $p_v = \Pr[f(\mathcal{M}^*_{\varepsilon}(v)) = 1]$, and suppose that $\frac{p_u}{p_v}\leq e^\varepsilon$ -- that is, suppose the queries did satisfy $\varepsilon$-DP. Then,
$$\begin{aligned}
    (p_u)^{100} \cdot (1-p_v)^{100} &\leq (e^\varepsilon p_v \cdot (1-p_v))^{100}\\
    &\leq \left( \frac{e^\varepsilon}{4} \right)^{100}.
\end{aligned}$$

The results in \cref{fig:opendp-attack-nonassoc-32}, then, would occur with probability $\leq 2^{-127}$ if $\mathcal{M}^*_\varepsilon$ were truly an $\varepsilon$-DP function for $\varepsilon = 0.5$.

Moreover, for these 100 rounds, $\min(\mathcal{M}^*_{\varepsilon}(u))> 33{,}554{,}400$ and $\max(\mathcal{M}^*_{\varepsilon}(v)) < 16{,}777{,}230$.
These results support the calculation in \cref{ea:32-nonassoc} that we do not have $\varepsilon$-DP for $\varepsilon = 0.5$ but instead have at least $2^{23}\varepsilon$-DP.

\end{attacklibrary}

\begin{attacklibrary}[Google DP (OpenMined wrapper); 64-bit Floats]
We create the datasets as described in \cref{ea:64-nonassoc}, with the exception that we set $v = [v_1,\ldots,v_{\len(v)-1}]$ (i.e., the last element is removed). The result, then, is that $\dsym(u,v) = 1$, so $u\nsym v$. From \cref{ea:64-nonassoc}, we expect to have $|\bsfp_{L,U}(u) - \bsfp_{L,U}(v)| = 3$, despite the fact that $\ssym \bs_{L,U} = \max\{|L|,U\} = 1$.

We run $\mathcal{M}^*_{\varepsilon}(u)$ and $\mathcal{M}^*_{\varepsilon}(v)$, where $\mathcal{M}^*_{\varepsilon}$ is the mechanism used by the library for computing $\bsfp_{L,U}$ with DP and where $\varepsilon = 1.0$. We define the post-processing function $f:\mathbb{R}\to\{0,1\}$ as follows: $f(x) = 1$ if and only if $x\geq 134{,}217{,}729.5$ (this value was chosen since it lies exactly between $\bsfp_{L,U}(u)$ and $\bsfp_{L,U}(v)$).

Running each of $f(\mathcal{M}^*_{\varepsilon}(u))$ and $f(\mathcal{M}^*_{\varepsilon}(v))$ 10 times gives the following results. (We only run these computations 10 times each due to the time needed to compute these summations of large datasets.)

\begin{figure}[H]
\begin{center}
    \begin{tabular}{c|c c }
     & $f(\mathcal{M}^*_{\varepsilon}(u))$ & $f(\mathcal{M}^*_{\varepsilon}(v))$ \\ \hline
     0 & 1 & 10 \\  
     1 & 9 &	0
    \end{tabular}
    \caption{``$1$-DP'' results from Google's DP library (with post-processing).}
    \label{fig:google-attack-nonassoc-64}
\end{center}
\end{figure}
By the fact that $u\nsym v$, the definition of $1$-DP, and the fact that DP is robust to post-processing per \cref{prop:post-proc}, we would expect to have
$$\frac{\Pr[f(\mathcal{M}^*_{\varepsilon}(v)) = 0]}{\Pr[f(\mathcal{M}^*_{\varepsilon}(u)) = 0]} \leq e
\text{ and }
\frac{\Pr[f(\mathcal{M}^*_{\varepsilon}(u)) = 1]}{\Pr[f(\mathcal{M}^*_{\varepsilon}(v)) = 1]} \leq e.
$$

Let $p_u = \Pr[f(\mathcal{M}^*_{\varepsilon}(u)) = 1]$ and $p_v = \Pr[f(\mathcal{M}^*_{\varepsilon}(v)) = 1]$, and suppose that $\frac{p_u}{p_v}\leq e^\varepsilon$ -- that is, suppose the queries did satisfy $\varepsilon$-DP. Maximizing
$$(10(p_u)^9(1-p_u) + (p_u)^{10})\cdot (1-p_v)^{10}$$ over $p_u,p_v\in[0,1]$ such that $p_u\leq e^\varepsilon \cdot p_v$ and $(1-p_v)\leq e^\varepsilon \cdot p_u$ for $\varepsilon = 1$ yields a maximum of $< 0.015$.

The results in \cref{fig:opendp-attack-nonassoc-32}, then, would occur with probability $< 0.015$ if $\mathcal{M}^*_\varepsilon$ were truly an $\varepsilon$-DP function for $\varepsilon = 1$.

These results support the calculation in \cref{ea:64-nonassoc} that we do not have $\varepsilon$-DP for $\varepsilon = 1$ but instead have at least $3\varepsilon$-DP.

\end{attacklibrary}

\subsubsection{Implementing \cref{ea:ham-round-once}}
\label{sec:implement-round-attack}

We next show how a malicious data analyst can use the results shown in \cref{sec:attack-floats-assoc} to recover sensitive bits in a dataset. We work with the ``known $n$'' summation functions provided in \cref{sec:ideal-sens-libraries} (where ``known $n$'' means that the function is designed to offer DP under the $\dco$ notion of neighboring datasets). The Google DP library does not have a ``known $n$'' bounded sum, but it does have a ``known $n$'' bounded mean, so we use this.

\begin{motivateattack} Suppose that our library has a clamp function as described in \cref{defn:bounded-sum} (all of the libraries that we work with offer this function) and that we can cast between data types. 

In this attack, we essentially create the datasets described in \cref{ea:ham-round-once}, but we do not need to start with these datasets directly. Instead, we can start with datasets of binary values, cast these binary values to floats, and apply clamping function to these datasets to effectively get the datasets described in \cref{ea:ham-round-once}.

Suppose we have a dataset with a field containing binary values (for example, a dataset with a ``had\_covid'' field, where the specified column has a value of 1 if and only if the individual has had \textsc{covid-19}), and suppose we know $2^j$ individuals in the dataset with ``had\_covid'' value $0$. We begin by selecting these $2^j$ individuals. Then, suppose we want to learn the ``had\_covid'' value of the individual located at index $i$ in the dataset. We select this individual as well, and, in our selected dataset of $2^j+1$ items, we cast the values in the ``had\_covid'' column to 64-bit floats. Then, we carefully set $U$ and $L$ as in \cref{ea:ham-round-once} and clamp our values. This has the effect of setting $0$ to $L$ and $1$ to $U$.

Finally, we perform an ``$\varepsilon$-DP'' bounded sum (or bounded mean) on this dataset for some $\varepsilon$. Using the ideas and evidence in the attacks below, we will be able to correctly recover the specified individual's ``had\_covid'' value with much larger probability than the specified value of $\varepsilon$ should allow. We perform one example using the ``known $n$'' bounded sum and one example using the ``known $n$'' bounded mean.

\end{motivateattack}

\begin{attacklibrary}[IBM \texttt{diffprivlib}: Bounded Sum, 64-bit Floats]
\label{al:ibm-sized-64}
We create the datasets as described in \cref{ea:ham-round-once}, except we set $j=5$ and set $L,U$ correspondingly. Additionally, in the construction of our datasets, we set all $L_i$ in the datasets to 0 instead of to $L$; and for all $U$ in the datasets, we instead put 1. The result is two datasets $u,v$ of binary values, with $\len(u) = \len(v) = 33$ and $\dham(u,v) = 1$.

We next clamp these datasets to values between $L$ and $U$. Since $0\leq L < U \leq 1$, in clamped datasets $u$ and $v$ we have $0\mapsto L$ and $1\mapsto U$.

Finally, we run $\mathcal{M}^*_{\varepsilon}(u)$ and $\mathcal{M}^*_{\varepsilon}(v)$, where $\mathcal{M}^*_{\varepsilon}$ is the mechanism used by the library for computing $\bsfp_{L,U}$ with $\varepsilon$-DP and where $\varepsilon = 0.5$. We define the post-processing function $f:\mathbb{R}\to\{0,1\}$ as follows: $f(x) = 1$ if and only if $x > \bround(33\cdot L)$.

Running each of $f(\mathcal{M}^*_{\varepsilon}(u))$ and $f(\mathcal{M}^*_{\varepsilon}(v))$ $10{,}000$ times gives the following results.

\begin{figure}[H]
\begin{center}
    \begin{tabular}{c|c c }
     & $f(\mathcal{M}^*_{\varepsilon}(u))$ & $f(\mathcal{M}^*_{\varepsilon}(v))$ \\ \hline
     0 & 1 & 9998 \\  
     1 & 9999 &	2
    \end{tabular}
    \caption{``$0.5$-DP'' results from IBM's \texttt{diffprivlib} (with post-processing).}
    \label{fig:ibm-ham-round}
\end{center}
\end{figure}
By the fact that $u\nham v$, the definition of $0.5$-DP, and the fact that DP is robust to post-processing per \cref{prop:post-proc}, we would expect to have
$$\frac{\Pr[f(\mathcal{M}^*_{\varepsilon}(v)) = 0]}{\Pr[f(\mathcal{M}^*_{\varepsilon}(u)) = 0]} \leq e^{0.5} < 1.65$$
and
$$\frac{\Pr[f(\mathcal{M}^*_{\varepsilon}(u)) = 1]}{\Pr[f(\mathcal{M}^*_{\varepsilon}(v)) = 1]} \leq e^{0.5} < 1.65.
$$

Let $p_u = \Pr[f(\mathcal{M}^*_{\varepsilon}(u)) = 1]$ and $p_v = \Pr[f(\mathcal{M}^*_{\varepsilon}(v)) = 1]$, and suppose that $\frac{p_u}{p_v}\leq e^\varepsilon$ -- that is, suppose the queries did satisfy $\varepsilon$-DP. Maximizing
$$\begin{aligned}
\MoveEqLeft \left(p_u^{10{,}000} + 10{,}000p_u^{9999}(1-p_u)\right)\\
&\cdot ((1-p_v)^{10{,}000}+10{,}000(1-p_v)^{9999}p_v\\
&+ \binom{10{,}000}{2}(1-p_v)^{9998}p_v^2)
\end{aligned}
$$
over $p_u,p_v\in[0,1]$ such that $p_u\leq e^\varepsilon \cdot p_v$ and $(1-p_v)\leq e^\varepsilon \cdot p_u$ for $\varepsilon = 0.5$ yields a maximum of $<2^{-100}$. The results in \cref{fig:ibm-ham-round}, then, would occur with probability $<2^{-100}$ if $\mathcal{M}^*_\varepsilon$ were truly an $\varepsilon$-DP function for $\varepsilon = 0.5$.

These results, then, support the calculation in \cref{ea:ham-round-once} that $\mathcal{M}^*_{\varepsilon}$ does not offer $\varepsilon$-DP for $\varepsilon = 0.5$.

\end{attacklibrary}

\begin{attacklibrary}[Google DP (OpenMined wrappings): Bounded Mean, 64-bit Floats]
We perform the same attack as described in \cref{al:ibm-sized-64}, except we set $j=7$ and set $L,U$ correspondingly. The result is two datasets $u,v$ of binary values, with $\len(u) = \len(v) = 33$ and $\dham(u,v) = 1$.

Finally, we run $\mathcal{M}^*_{\varepsilon}(u)$ and $\mathcal{M}^*_{\varepsilon}(v)$, where $\mathcal{M}^*_{\varepsilon}$ is the mechanism used by the library for computing the bounded mean (with values bounded by $L$ and $U$) with $\varepsilon$-DP and where $\varepsilon = 0.5$. We define the post-processing function $f:\mathbb{R}\to\{0,1\}$ as follows: $f(x) = 1$ if and only if $x > L$.

Running each of $f(\mathcal{M}^*_{\varepsilon}(u))$ and $f(\mathcal{M}^*_{\varepsilon}(v))$ $10{,}000$ times gives the following results.

\begin{figure}[H]
\begin{center}
    \begin{tabular}{c|c c }
     & $f(\mathcal{M}^*_{\varepsilon}(u))$ & $f(\mathcal{M}^*_{\varepsilon}(v))$ \\ \hline
     0 & 0 & 10{,}000 \\  
     1 & 10{,}000 &	0
    \end{tabular}
    \caption{``$0.5$-DP'' results from Google's DP library (with post-processing).}
    \label{fig:google-dp-ham}
\end{center}
\end{figure}
By the fact that $u\nham v$, the definition of $0.5$-DP, and the fact that DP is robust to post-processing per \cref{prop:post-proc}, we would expect to have
$$\frac{\Pr[f(\mathcal{M}^*_{\varepsilon}(v)) = 0]}{\Pr[f(\mathcal{M}^*_{\varepsilon}(u)) = 0]} \leq e^{0.5} < 1.65$$
and
$$\frac{\Pr[f(\mathcal{M}^*_{\varepsilon}(u)) = 1]}{\Pr[f(\mathcal{M}^*_{\varepsilon}(v)) = 1]} \leq e^{0.5} < 1.65.
$$

Let $p_u = \Pr[f(\mathcal{M}^*_{\varepsilon}(u)) = 1]$ and $p_v = \Pr[f(\mathcal{M}^*_{\varepsilon}(v)) = 1]$, and suppose that $\frac{p_u}{p_v}\leq e^\varepsilon$ -- that is, suppose the queries did satisfy $\varepsilon$-DP. Then,
$$\begin{aligned}
    (p_u)^{10{,}000} \cdot (1-p_v)^{10{,}000} &\leq (e^\varepsilon p_v \cdot (1-p_v))^{10{,}000}\\
    &\leq \left( \frac{e^\varepsilon}{4} \right)^{10{,}000}.
\end{aligned}$$

The results in \cref{fig:google-dp-ham}, then, would occur with probability $\leq 2^{-12{,}000}$ if $\mathcal{M}^*_\varepsilon$ were truly an $\varepsilon$-DP function for $\varepsilon = 0.5$. These results, then, support the calculation in \cref{ea:ham-round-once} that $\mathcal{M}^*_{\varepsilon}$ does not offer $\varepsilon$-DP for $\varepsilon = 0.5$.

\end{attacklibrary}

\section{Proposed Solutions to Incorrect Sensitivities}
\label{sec:sols}

In this section, we propose several solutions in which we modify the implementation of bounded sum in order to recover a sensitivity that equals or is close to its idealized sensitivity. Each of the five methods is fitting in different types of situations and has different advantages and disadvantages, as we will explain.

We can broadly classify our solutions into two different types of fixes: the ones that modify the ``standard'' iterative sum function and that ones that check for the parameters before the call of the sum function. 

We remark that all theorems in this section hold for normal floats (Definition~\ref{defn:fp}). We believe that the results also hold for subnormal floats (Definition~\ref{note:subnormal-floats}), and we will formalize this in a future version of this paper.

\subsection{Bounding \texorpdfstring{$n$}{n}
}\label{sec:checkmul}

This strategy can be used as a stand-alone solution for the integers, and ideas from this strategy are important when implementing a solution over floats.

\subsubsection{Solution for Integers}\label{sec:61ints}

A natural first solution for integer data types is to directly avoid any overflow. We can achieve this by checking for overflow before calling the bounded sum function. This method applies to integer data types. Its main benefit is that it allows for the use of the idealized sensitivity (as proven in Theorem~\ref{thrm:reals-sens}). However, this method only works when the length of the input dataset is known.

\iffalse
\begin{method}[Check Multiplication]
\label{method:check-mul}
Let $u$ be a dataset of integers of type $T_{[L, U]}$, $n$ the length of the dataset, $U$ an upper bound on the largest number in the dataset, and $L$ a lower bound on the smallest number in the dataset such that $U \cdot n \leq \max(T), L \cdot n \geq \min(T)$. Then, the \textit{check multiplication} function returns $\bs^{*I}_{L, U, n}(u)$.

We define \texttt{check\_mul} as follows:

\begin{lstlisting}[language = Python, escapechar=|]
def check_mul(L,U,n,u,T)
    the_sum = 0
    for elt in u:
        the_sum = the_sum + elt
    return the_sum 
\end{lstlisting}
\end{method}

\fi

The direct check for overflow method, which we call \textit{check multiplication} for conciseness, only works when $\len(u)$ is known to the data analyst (or when a DP count is performed to estimate $\len(u)$, or when an upper bound on $\len(u)$ is known), and it is overprotective in the sense that the potential overflow may only occur in a worst-case scenario. This method does not immediately work for floating-point arithmetic because blow-up can occur even when $U \cdot n < \max(T)$ and $L \cdot n > \min(T)$ due to rounding effect. The examples in Sections~\ref{sec:round-attack-floats}, \ref{sec:accum-round-floats}, and \ref{sec:accum-round-floats-ii} illustrate this effect.

\begin{theorem}[Sensitivity of Bounded Sum With Check Multiplication]
\label{thrm:checked-arith-known}
    Let $\bs^*_{L, U, n}: T_{[L, U]}^n \rightarrow T$ be the iterative bounded sum function on the type $T$ of $k$-bit integers such that $U \cdot n \leq \max(T), L \cdot n \geq \min(T)$. Then, $$\sco \bs^*_{L, U, n} = \sham \bs^*_{L, U, n} = U-L.$$
\end{theorem}

\begin{proof}

Let $u, v \in T^n_{[L, U]}$ be two datasets such that $u \nco v$. If $n \cdot U \leq \max(T)$ and $n \cdot L \geq \min(T)$ then we claim that $\bsfp(u) = \bs(u)$ and $\bsfp(v) = \bs(v)$, per \cref{defn:sat-add}. Thus,
\[
    |\bsfp(u) - \bsfp(v)| = |\bsfp(u) - \bsfp(v)| \leq U-L,
\]
so $$\sco \bs^*_{L, U, n} = \sham \bs^*_{L, U, n} = U-L.$$

\end{proof}

\subsubsection{Use Case for Floats}\label{sec:61floats}

When working with floats, it is also necessary to ensure that overflow cannot occur -- that is, we must check that a computation cannot result in an answer of $\pm\texttt{inf}$. Per \cref{defn:inf}, this is necessary for all the solutions we present for floating-point values. As a result of this behavior, a dataset that results in a sum of $\pm\texttt{inf}$ will always return an answer of $\pm\texttt{inf}$, even after noise addition; this deterministic behavior will cause privacy violations.

We now describe a generic check for overflow. Let $\bs_{L,U,n}$ be the idealized bounded sum, $\bsfp_{L,U,n}$ be an implementation of the bounded sum function, and $\mathit{acc}\in\mathbb{R}$ (where \textit{acc} stands for accuracy) be some bound such that $|\bs_{L,U,n}(v) - \bsfp_{L,U,n}(v)|\leq \mathit{acc}$ for all $v\in T^n_{[L,U]}$. Additionally, let $\mathit{Round}_{\infty}:\mathbb{R}\to\mathbb{R}$ return the smallest float greater than or equal to the input; and let $\mathit{Round}_{-\infty}:\mathbb{R}\to\mathbb{R}$ return the largest float less than or equal to the input. We remark that in the IEEE 754 standard, $\mathit{Round}_{\infty}$ is referred to as ``round toward $+\infty$'', and $\mathit{Round}_{-\infty}$ is referred to as ``round toward $-\infty$'' \cite{ieee08}.

\begin{enumerate}
    \item Compute $L' = \mathit{Round}_{-\infty}(L\cdot n)$ and $U' = \mathit{Round}_{\infty}(U\cdot n)$ using floating-point arithmetic. (This can be computed by performing floating-point multiplication with banker's rounding and then going to the next float toward $\pm\texttt{inf}$ as appropriate.)
    \item Ensure that $\mathit{Round}_{-\infty}(L' - \mathit{acc})\neq -\texttt{inf}$ and that $\mathit{Round}_{\infty}(U' + \mathit{acc})\neq \texttt{inf}$. (This can be computed by performing floating-point addition with banker's rounding and then going to the next float toward $\pm\texttt{inf}$ as appropriate.)
    
    \item If these checks pass, overflow cannot occur, and the computation can proceed; otherwise, reject computations using these parameter values $n,L,U$.
\end{enumerate}

\begin{remark}
    
    While this check for overflow works for any value of $acc$ such that $|\bs_{L,U,n}(v) - \bsfp_{L,U,n}(v)|\leq \mathit{acc}$ for all $v\in T^n_{[L,U]}$, the accuracy bounds presented in \cref{thrm:acc-upper-bounds} can be used to compute $\mathit{acc}$.
    
\end{remark}

\subsection{Modular Noise Addition 
}\label{sec:modularnoise}

As a next step, we would like a solution that works for integer types when the length of the input dataset is not known, and which is not as conservative as the check multiplication method, given that check multiplication might not be applicable even if no overflow would actually occur in the computation of the bounded sum.

For this reason, we introduce \emph{modular bounded sum}, which computes a bounded sum using modular arithmetic. We remark that modular bounded sum applies to both \emph{signed} and \emph{unsigned} integers, since we can treat the signed integers as a re-naming of the unsigned integers. This idea is described more formally in Section~\ref{sec:ints-are-mod-n}.

This method allows for overflow, which has the effect that two adjacent datasets can produce answers that, in terms of the absolute value of the difference between outputs, are as far from each other as possible; e.g., 0 and $2^k-1$ for $k$-bit unsigned integers. However, when doing modular arithmetic, we should consider these two values to have a distance of only 1, captured by the following definition:

\begin{definition}[Modular Distance]\label{defn:moddist}
    Let $x, y \in \mathbb{Z}/m\mathbb{Z}$ for some $m \in \mathbb{Z}$. Then, the \emph{modular distance} between $x$ and $y$ is defined as
    $$\dmod^m(x, y) = \min\{x-y, y-x\}$$
    where we use $\{0,\ldots, m-1\}$ to represent the elements of $\mathbb{Z}/m\mathbb{Z}$ and, for the purposes of evaluating ``$\min$'', we specify that the elements are ordered as $0 < 1 < \cdots < m-1$.
    
\end{definition}

Because we are working with modular distance, we measure sensitivity using \emph{modular sensitivity}, as defined below.

\begin{definition}[Modular Sensitivity]
\label{defn:mod-sens}
We define the \emph{modular sensitivity} of a function $f:\Vect(\mathcal{D}) \rightarrow \mathbb{Z}/m\mathbb{Z}$, with respect to a metric $d$ on $\Vect(\mathcal{D})$, as

$$
    \Delta^m_d f = \sup_{
    \substack{
        u,v\in \Vect(\mathcal{D})  \\
        u\simeq_\mathit{d} v
        }
    }\dmod^m(f(u),f(v)).
$$
\end{definition}

We recall from the \cref{al:ibm-overflow} that modular bounded sum does not fulfill DP when used with non-modular noise.

If noise were added to these bounded sum results in a non-modular fashion, $u$ and $u'$ would be very unlikely to result in similar answers, so a data analyst could readily tell the difference between a sum resulting from $u$ and a sum resulting from $u'$, even though $u\simeq_\mathit{d} u'$. For this reason, we define \emph{modular noise addition} and prove that modular noise addition can fulfill DP.

\iffalse

\begin{definition}[Modular Noise Addition]
    Given any additive mechanism $\mathcal{M}$ that samples from a distribution with integer support, \emph{modular noise addition} consists of adding the sampled noise modulo $m$ to a value.
\end{definition}

\fi

\begin{theorem}[Modular Noise Addition Offers $\varepsilon$-DP]\label{thm:modularnoise}

Let $Z$ be an integer-valued random variable and $d$ a metric on $\Vect(D)$. Suppose that for every function $f: \Vect(D) \rightarrow \mathbb{Z}$ of sensitivity at most $c$ with respect to $d$, the mechanism $\mathcal{M}(n) = f(n) + Z$ is $\varepsilon$-DP. Then for every function $f: \Vect(D) \rightarrow \mathbb{Z}/n\mathbb{Z}$ of modular sensitivity at most $c$ with respect to $d$, the mechanism $\mathcal{M}_{\textrm{mod}}(n) = f(n) + Z \mod m$ is $\varepsilon$-DP.

\end{theorem}

\begin{proof}

Let $\Delta_d f \leq c$. This means that, for $u\simeq_d u'$, we have $|f(u) - f(u')| \leq c$. Therefore, for all $m\in\mathbb{Z}$, there exists $k\in\mathbb{Z}$ such that $f(u) - f(u') \leq b+km$, where $b\in [-c,c]\cap \mathbb{Z}$. This is equivalent to $|f(u) - (f(u') + km)|\leq c$.

Because $|f(u) - (f(u') + km) |\leq c$, by the theorem statement $f(u)+Z$ and $f(u')+km + Z$ are $\varepsilon$-close. By post-processing, then, $f(u)+Z\bmod m$ and $f(u')+km+Z\equiv f(u') + Z\bmod m$ are $\varepsilon$-close. Therefore, adding $Z$ modulo $m$ provides $\varepsilon$-DP for $\mathbb{Z}/m\mathbb{Z}$-valued functions with modular sensitivity $\leq c$ with respect to $d$ on input datasets.
\end{proof}

We now proceed to the proof that modular bounded sum results in a useful sensitivity. 

\iffalse

Before providing this proof, we note some properties and prove some lemmas about the behavior of modular addition on the signed and unsigned integers that will be helpful in the proof.

\begin{property}[Associativity of Addition and Subtraction in $\mathbb{Z}/m\mathbb{Z}$]
\label{prop:associative-mod-n}
Addition and subtraction are associative in $\mathbb{Z}/m\mathbb{Z}$ for all $m\in \mathbb{Z}$.
\end{property}

\begin{property}[No Rounding Occurs]
\label{prop:addition-no-rounding}
For all $a,b\in\mathbb{Z}/m\mathbb{Z}$, we have $\dmod^m(a+b,a)\leq b'$, where $b'\equiv b\pmod{m}$ and $b'\in[0,m-1]\cap\mathbb{Z}$. (Although this may feel obvious, an analogous property is not true for floats.)
\end{property}

\begin{lemma}
\label{lemma:signed-unsigned-mod-associative}
Modular addition and subtraction are associative on the signed and unsigned integers.
\end{lemma}

\begin{proof}
By Property \ref{prop:associative-mod-n}, addition and subtraction are associative in $\mathbb{Z}/m\mathbb{Z}$, so because the signed and unsigned integers are equal to $\mathbb{Z}/m\mathbb{Z}$ for some $m\in\mathbb{Z}$ (as described in Section~\ref{sec:ints-are-mod-n}), addition and subtraction are associative for the signed and unsigned integers.
\end{proof}

\fi

\begin{theorem}[Modular Sensitivity of Modular Bounded Sum]
\label{thrm:modular-arith-unknown}
\label{thrm:modular-arith-known}
    Let $T$ be a type of (signed or unsigned) $k$-bit integers and let $L, U \in T$. Let $\bs^*_{L, U}: \Vect(T_{[L, U]}) \rightarrow T$ be the iterative bounded sum function on $k$-bit integers of type $T$ implemented with modular addition, and similarly for $\bs^*_{L, U, n}: T_{[L, U]}^n \rightarrow T$ . Then, where we let $m=2^k$,
    \begin{enumerate}
        \item (Unknown $n$.) 
        
        \[
            \ssym^m \bsfp_{L, U} = \sid^m \bsfp_{L, U} =
        \]
        \[
            = \min\Big\{\left\lfloor \frac{m}{2} \right\rfloor, \max\{|L|,U\}\Big\} \leq \max\{|L|,U\}.
        \]
        \item (Known $n$.) $$\sco^m \bsfp_{L, U, n} = \sham^m \bsfp_{L, U, n} = \min\Big\{\left\lfloor \frac{m}{2} \right\rfloor, U-L\Big\} \leq U-L.$$
    \end{enumerate}
    
\end{theorem}

\begin{note}[Implement with caution, unknown $n$]
When working with $k$-bit signed integers, we must ensure that computing the sensitivity $|L|$ does not cause overflow. This is most readily done by disallowing $L=-2^k$ or by storing the sensitivity as an unsigned $k$-bit integer. When working with the 32-bit integers, for example, in the case where $L=-2^{31}$ we would have $|L|$ overflow to $-2^{31}$, which would lead to an underestimate when computing $\max\{|L|,U\}$. An easy method is to disallow $L=-2^{31}$; or, more generally, to disallow $L=\min(T)$ when working with the signed integers of type $T$.
\end{note}

\begin{note}[Implement with caution, known $n$]
When working with $k$-bit signed integers, the computation of the sensitivity $(U-L)$ should be stored as the corresponding unsigned value (i.e., the value in $[0,2^k-1]\cap \mathbb{Z}$, rather than a value in $[-2^{k-1},2^{k-1}-1]\cap \mathbb{Z}$) to avoid issues of overflow. For example, when working with the 32-bit integers, setting $U=0$ and $L=-2^{31}$ would result in a sensitivity of $2^{31}$, which would overflow to $-2^{31}$ in the signed integers and result in a negative sensitivity.
\end{note}

\begin{proof}[Proof of \cref{thrm:modular-arith-known}]

We prove each part below. Let $m=2^k$. Because addition over $\mathbb{Z}/m\mathbb{Z}$ is associative, we can essentially use the proof of Theorem~\ref{thrm:reals-sens}, with slight modifications.

\medskip

(1) Let $u,u'\in\Vect(T_{[L,U]})$ be two datasets such that $u \nsym u'$. From the formal definition of a histogram in Definition~\ref{defn:histogram}, we observe that, for all $v\in\Vect(T_{[L,U]})$, 
$$\begin{aligned}
    \bsfp_{L,U}(v) &= \left( \sum_{i=1}^{\len(v)} v_i \right) \bmod{m}\\
    &= \left( \sum_{z\in T }h_v(z)\cdot z \right) \bmod{m}
\end{aligned}
$$
Because $u\nsym u'$, we know that there is at most one value $z^*$ such that $|h_u(z^*) - h_{u'}(z^*)| = 1$, and that, for all $z\neq z^*$, we have $|h_u(z) - h_{u'}(z)| = 0$. 

We now calculate $\dmod^m(\bsfp_{L,U}(u), \bsfp_{L,U}(u'))$. Recall that $\dmod^m(x, y) = \min\{x-y, y-x\}$. We only evaluate the left term within the ``$\min$'' expression and then compute the other side of the expression by a symmetric argument. Assume, without loss of generality, that $h_u(z^*)\geq h_{u'}(z^*)$. We can then write the following expressions. 
\begin{align*} 
\MoveEqLeft\bsfp_{L,U}(u) - \bsfp_{L,U}(u') \\
&= \left(\sum_{z\in T}h_u(z)\cdot z - \sum_{z\in T}h_{u'}(z)\cdot z \right)\bmod{m} \\
&= \left(\sum_{z\in T\setminus z^*}h_u(z)\cdot z - \sum_{z\in T\setminus z^*}h_{u'}(z)\cdot z \right)
+ \\
&+\left(h_u(z^*)\cdot z^* - h_{u'}(z^*)\cdot z^*\right) \bmod{m} \\
&= z^* \cdot (h_u(z^*) - h_{u'}(z^*)) \bmod{m}\\
&\in \{0,\ldots,z^*\}.
\end{align*}
By a symmetric argument, $\bsfp_{L,U}(u') - \bsfp_{L,U}(u)\in \{0,\ldots,m-z^*\}$. Modular sensitivity is defined as the maximum possible modular distance between the sums of neighboring datasets. We observe that setting $$z^*=\min\{\lfloor m/2 \rfloor, \max\{|L|,U\}\}$$ will maximize the expression $\min\{z^*, -z^*\}$. By the definition of sensitivity, then,
$$\ssym^m \bsfp_{L,U} \leq \min\{\lfloor m/2 \rfloor, \max\{|L|,U\}\} \leq \max\{|L|,U\}.$$
By Theorem~\ref{thm:dcotoham}, we also have
$$\sid^m \bsfp_{L,U} \leq \min\{\lfloor m/2 \rfloor, \max\{|L|,U\}\} \leq \max\{|L|,U\}.$$

For the lower bound, consider the two datasets $u=[0]$ and $u'=[0,\min\{\lfloor m/2 \rfloor, \max\{|L|,U\}\}]$. Then, $u\nid u'$. Moreover, $$\dmod^m(\bsfp_{L,U}(u), \bsfp_{L,U}(u')) = \min\{\lfloor m/2 \rfloor, \max\{|L|,U\}\}.$$ This means, then, that $\sid^m \bsfp_{L,U} \geq \min\{\lfloor m/2 \rfloor, \max\{|L|,U\}\}$. By the contrapositive of Theorem~\ref{thm:dcotoham}, then, $\ssym^m \bsfp_{L,U} \geq \max\{|L|,U\}$.

Combining these upper and lower bounds on the idealized sensitivity tells us, then, that $$\ssym^m \bsfp_{L,U} = \sid^m \bsfp_{L,U} = \min\{\lfloor m/2 \rfloor, \max\{|L|,U\}\}.$$

\medskip

(2) Let $u,u'\in\mathbb{R}^n$ be two datasets such that $u\simeq_\mathit{CO} u'$. By Lemma~\ref{lemma:metric-relate}, there is a permutation $\pi\in S_n$ such that $\pi(u)\simeq_\mathit{Ham} u'$. This means there is at most one index $i^*$ such that $\pi(u)_{i^*}\neq u'_{i^*}$. We now calculate $\dmod^m(\bsfp_{L,U}(u), \bsfp_{L,U}(u'))$. Recall that $\dmod^m(x, y) = \min\{x-y, y-x\}$. We only evaluate the left term within the ``$\min$'' expression and then compute the other side of the expression by a symmetric argument. We can then write the following expressions. 

\begin{align*} 
\MoveEqLeft\bsfp_{L,U,n}(u) - \bsfp_{L,U,n}(u') = \bsfp_{L,U,n}(\pi(u)) - \bsfp_{L,U,n}(u') \\
&= \left( \sum_{i=1}^n \pi(u)_i - \sum_{i=1}^n u'_i \right) \bmod{m} \\
&= \left( \pi(u)_{i^*} - u'_{i^*} \right) \bmod{m}\\
&\in \{0,\ldots, U-L\}, 
\end{align*}
where we represent $U-L\in\{0,\ldots, m-1\}$. The final line follows from the fact that all values are clamped to the interval $[L,U]$, so the largest difference in sums arises when, without loss of generality, $\pi(u)_{i^*} = U$ and $u'_{i^*} = L$. By similar logic, $\bsfp_{L,U,n}(u') - \bsfp_{L,U,n}(u)\in \{0,\ldots,m-(L-U)\}$. Thus
\[
    \dmod^m(\bsfp_{L, U}(u'), \bsfp_{L, U}(u)) \leq \min\{U-L, m-(U-L)\} =
\]
\[
    = \min\Big\{\left\lfloor \frac{m}{2} \right\rfloor, U-L\Big\} \leq U-L.
\]

For the lower bound, consider the datasets $u=[z]$ and $u'= [z']$, where $z,z'\in[L,U]\cap \mathbb{Z}$ such that $z-z'=\min\{\lfloor \frac{m}{2} \rfloor, U-L\}$. We note that $u \simeq_{\mathit{Ham}} u'$ and $\dmod^m(\bs_{L,U,n}(u), \bs_{L,U,n}(u')) = \min\{\lfloor \frac{m}{2} \rfloor, U-L\}$. This means, then, that $\sid^m \bsfp_{L,U} \geq \min\{\lfloor \frac{m}{2} \rfloor, U-L\}$. By the contrapositive of Theorem~\ref{thm:dcotoham}, it follows that $\sco^m \bsfp_{L,U,n} \geq \min\{\lfloor \frac{m}{2} \rfloor, U-L\}$. Combining these upper and lower bounds on the idealized sensitivity tells us, then, that $$\sco^m \bsfp_{L,U,n} = \sham^m \bsfp_{L,U,n} = \min\Big\{\left\lfloor \frac{m}{2} \right\rfloor, U-L\Big\}.$$

This concludes the proof.

\end{proof}

Note that \cref{thrm:modular-arith-unknown} relies on modular arithmetic, so it can only be applied to integer types. The primary disadvantage of this summation method is that computing a DP bounded sum on a dataset with a very large true sum may output a very small sum (or vice versa) due to the modular reduction. An advantage of this method is that it does \textit{not} require loosening the sensitivity relation, and thus Theorem~\ref{thrm:modular-arith-unknown} achieves the same sensitivities as the corresponding idealized sensitivities in Section \ref{sec:theoreticalsensitivity}. 

Another benefit of the modular addition method is that it preserves associativity. This means that applying bounded sum with modular addition to different orderings of an unordered dataset will always yield the same output. As we show in Sections~\ref{sec:non-assoc-ints} and \ref{sec:attack-floats-assoc}, some implementations of bounded sum do not necessarily preserve associativity, in which case applying a function to different orderings of an unordered dataset can yield different outputs.

A final benefit of modular sum is that many libraries perform modular arithmetic on the integers by default. For example, overflow in both Python and C++ occurs in the modular fashion described in this section. Therefore, a naive implementation of bounded sum on the integers may actually be implementing the modular bounded sum and modular noise addition method described here. Interestingly, this means that libraries that ignore issues of overflow often fulfill the privacy guarantees; on the other hand, libraries that try to prevent overflow by using a strategy like saturation addition may encounter the vulnerabilities described in Section \ref{sec:non-assoc-ints}. 

Moreover, as demonstrated in \cref{al:ibm-overflow}, it is essential that the added noise is also of integer type and that noise addition is done with modular arithmetic. Lastly, we remark that \cref{thm:modularnoise} (modular noise addition offers $\varepsilon$-DP) only applies to additive mechanisms, and hence this modular approach does not necessarily work for other uses of sensitivity, such as in the exponential mechanism.

\subsection{Split Summation for Saturation Arithmetic}\label{sec:splitsum}

Before we turn to methods that work for floats, we present two methods that work for integers with saturation arithmetic even when the length of the input dataset is not known, as an alternative to the modular arithmetic method presented in Section \ref{sec:modularnoise}. In this section we present Method \ref{method:split-sum-ints}, which we call \emph{split summation for saturating integers}, and in the next section we present Method \ref{method:sat-add-RP-ints}, called \emph{randomized permutation} (RP). We will show that the implemented sensitivity for both methods matches the idealized sensitivity for bounded sum stated in Theorem~\ref{thrm:reals-sens}.

As we showed in \cref{sec:non-assoc-ints}, integers with saturating arithmetic do \emph{not} satisfy associativity. However, as noted in Property~\ref{prop:sat-add-result}, we observe that non-associativity only arises from the mixing of positive and negative integers. This motivates the introduction of the \emph{split summation method}, which consists of adding the negative terms and non-negative terms separately. We remark that this is the first of our methods which changes the summation function: instead of adding the elements in the order in which they appear in the dataset, we add them in a different pre-defined order.

\begin{method}[Split Summation]
\label{method:split-sum-ints}
    Let $U$ denote the upper bound, $L$ the lower bound, and $u$ the input dataset with numeric elements of type $T$. We define the \textit{split summation} function as follows: 
\begin{lstlisting}[language = Python, escapechar=|,frame=single]
def split_summation(u)
    P = 0
    N = 0
    for elt in u:
        if elt < 0:
            N = N + elt
        else:
            P = P + elt
    return P + N
\end{lstlisting}

\end{method}

\begin{theorem}[Sensitivity of bounded sum with split summation]\label{thrm:splitsummation}
    Let $\bs^*_{L, U}: \Vect(T_{[L, U]}) \rightarrow T$ be the bounded sum function with split summation (Method~\ref{method:split-sum-ints}) on the $k$-bit integers of type $T_{[L, U]}$ with saturation arithmetic, and similarly for $\bs^*_{L, U, n}: T_{[L, U]}^n \rightarrow T$. Then,
    \begin{enumerate}
        \item (Unknown $n$.) $\ssym \bsfp_{L, U} = \sid \bsfp_{L, U}  = \max \{ |L|, U \}$.
        \item (Known $n$.) $\sco \bsfp_{L, U, n} = \sham \bsfp_{L, U,n} = U-L$.
    \end{enumerate}
\end{theorem}

\begin{proof}

We prove each part below. Let $m = 2^k$.

\begin{enumerate}

\item Let $u,u'\in\Vect(T)$ such that $u\nsym u'$. Additionally, for computing $\bsfp_{L,U}(u)$, let $P,N$ be the final values of $\texttt{P},\texttt{N}$ in Method~\ref{method:split-sum-ints}. Likewise, for computing $\bsfp_{L,U}(u')$, let $P',N'$ be the final values of $\texttt{P},\texttt{N}$ in Method~\ref{method:split-sum-ints}.

Note that saturation addition is associative. Because $u\simeq_\mathit{Sym} u'$, we know that there is at most one value $z^*$ such that $|h_u(z^*) - h_{u'}(z^*)| = 1$, and that, for all $z\neq z^*$, we have $|h_u(z) - h_{u'}(z)| = 0$.

Consider the case where $z^* \geq 0$. By the logic in the proof of Theorem~\ref{thrm:reals-sens}, then,
$|N-N'| = 0$ and $|P-P'|\leq U$.

Now, consider the case where $z^* < 0$. By the logic in the proof of Theorem~\ref{thrm:reals-sens}, then,
$|N-N'| \leq |L|$ and $|P-P'| = 0$.

By the definition of sensitivity, then, $\ssym \bsfp_{L,U} \leq \max\{|L|,U\}$. By Theorem~\ref{thm:dcotoham}, we also have $\sid \bs_{L,U} \leq \max\{|L|,U\}$.

By the logic used in the proof of Theorem~\ref{thrm:reals-sens}, we see that this is also a lower bound on the sensitivity, which implies that $\ssym \bsfp_{L,U} = \sid \bsfp_{L,U} = \max\{|L|, U\}$.

\item Let $u,u'\in T^n$ such that $u\nco u'$. Additionally, for computing $\bsfp_{L,U,n}(u)$, let $P,N$ be the final values of $\texttt{P},\texttt{N}$ in Method~\ref{method:split-sum-ints}. Likewise, for computing $\bsfp_{L,U,n}(u')$, let $P',N'$ be the final values of $\texttt{P},\texttt{N}$ in Method~\ref{method:split-sum-ints}.

Recall that saturation addition is associative. By Lemma~\ref{lemma:metric-relate}, there is a permutation $\pi\in S_n$ such that $\pi(u)\simeq_\mathit{Ham} u'$. This means there is at most one index $i^*$ such that $\pi(u)_{i^*}\neq u'_{i^*}$.

Let $\pi(u)^P$ represent the vector of values in $\pi(u)\geq 0$; likewise for $\pi(u)^N, u'^P, u'^N$. We consider the following cases and apply the logic of Theorem~\ref{thrm:reals-sens} due to associativity.

\smallskip
\noindent \textit{(Case i.)} $\pi(u)_{i^*} \geq 0$ and $u'_{i^*}\geq 0$. Here, $\pi(u)^P\nham u'^P$, so $|\bsfp_{L,U,n}(\pi(u)^P) - \bsfp_{L,U,n}(u'^P)| \leq U-\max\{0,L\}$. Also, $\pi(u)^N = u'^N$, so $|\bsfp_{L,U,n}(u) - \bsfp_{L,U,n}(u')| \leq U-\max\{0,L\}$.

\smallskip
\noindent \textit{(Case ii.)} $\pi(u)_{i^*} \geq 0$ and $u'_{i^*} < 0$. Here, $\pi(u)^P\nsym u'^P$, so $|\bsfp_{L,U,n}(\pi(u)^P) - \bsfp_{L,U}(u'^P)| \leq U$. Likewise, $\pi(u)^N \nsym u'^N$, so $|\bsfp_{L,U,n}(\pi(u)^N) - \bsfp_{L,U}(u'^N)| \leq L$. Hence, it follows that $|\bsfp_{L,U,n}(u) - \bsfp_{L,U,n}(u')| \leq U - L$.

\smallskip
\noindent \textit{(Case iii.)} $\pi(u)_{i^*} < 0$ and $u'_{i^*} \geq 0$. Here, $\pi(u)^P\nsym u'^P$, so $|\bsfp_{L,U,n}(\pi(u)^P) - \bsfp_{L,U}(u'^P)| \leq U$. Likewise, $\pi(u)^N \nsym u'^N$, so $|\bsfp_{L,U,n}(\pi(u)^N) - \bsfp_{L,U}(u'^N)| \leq L$. ence, it follows that $|\bsfp_{L,U,n}(u) - \bsfp_{L,U,n}(u')| \leq U - L$.

\smallskip
\noindent \textit{(Case iv.)} $\pi(u)_{i^*} < 0$ and $u'_{i^*}< 0$. Here, $\pi(u)^N \nham u'^N$, so $|\bsfp_{L,U,n}(\pi(u)^N) - \bsfp_{L,U,n}(u'^N)| \leq \min\{0,U\} - L$. Also, $\pi(u)^P = u'^P$, so $|\bsfp_{L,U,n}(u) - \bsfp_{L,U,n}(u')| \leq \min\{0,U\} - L$.

By the definition of sensitivity, then, $\sco\bsfp_{L,U,n} \leq U-L$. By Theorem~\ref{thm:dcotoham}, we have $\sham \bsfp_{L,U,n} \leq U-L$.

By the logic used in the proof of Theorem~\ref{thrm:reals-sens}, we see that this is also a lower bound on the sensitivity, which implies that $\sco \bsfp_{L,U,n} = \sham \bsfp_{L,U,n} = U-L$.
    
\end{enumerate}

\end{proof}

Note that Method~\ref{method:split-sum-ints} would also work with modular addition instead of saturation arithmetic. This is because, when working with integers, issues with integer arithmetic only arise due to non-associativity. However, modular addition is associative, so the same procedure would hold.

\subsection{Randomized Permutation for Saturating Integers 
}\label{sec:rp-sat-int}

We now present the second alternative method for integers with saturating arithmetic, called \emph{randomized permutation (RP) for saturating integers}. Like the split summation method, RP matches the sensitivity of bounded sum, but when combined with the split summation method in Method~\ref{method:split-sum-ints}, will lead to one of our methods for float types (\cref{sec:splitsumrp}). The RP method randomly permutes the dataset before carrying out naive iterative summation. A potential accuracy advantage of RP over split summation is that since the positive and negative values are spread out, it is less likely that saturation will occur and cause a deviation from the true answer. On the other hand, the randomization has the downside that different runs may produce different results, but that is less of a concern given that we will be adding noise for privacy at the end anyway.

\begin{definition}[Randomized Permutation (RP) Function]
\label{defn:rr-transform}

We define $\Pi$ as a randomized function that, on input some dataset $u\in\mathcal{D}^n$ returns $\pi(u)$, where $\pi\in S_n$ is chosen uniformly at random from $S_n$, and $\pi(n)$ is as in Definition~\ref{def:perm}.
\end{definition}

Until now, the notion of sensitivity that we have employed (Definition~\ref{defn:abs-sens}) was for deterministic functions. To define sensitivity for randomized functions, we follow \cite{sv21} and require \emph{couplings} between output distributions of neighboring datasets, and we then work with a definition of sensitivity based on these couplings. The idea of using couplings for this purpose was introduced in \cite{sv21}.

\begin{definition}[Coupling]\label{defn:coupling}
    Let $r_1, r_2$ be two random variables defined over the probability spaces $\mathcal{R}_1$ and $\mathcal{R}_2$, respectively. A \emph{coupling} of $r_1$ and $r_2$ is a joint variable $(\tilde{r_1}, \tilde{r_2})$ taking values in the product space $\mathcal{R}_1 \times \mathcal{R}_2$ such that $\tilde{r_1}$ has the same marginal distribution as $r_1$ and $\tilde{r_2}$ has the same marginal distribution as $r_2$.
\end{definition}

\begin{definition}[Sensitivity of Randomized Functions]

A randomized function $f: \Vect(\mathcal{D}) \rightarrow \mathbb{R}$ has \textit{sensitivity} $\Delta f$ if for all neighboring datasets $u \nd u'$ there exists a coupling $(\tilde{s}_u, \tilde{s}_{u'})$ of the random variables $s_u = f(u)$ and $s_{u'} = f(u')$ such that with probability 1, $d (\tilde{s}_u, \tilde{s}_{u'}) \leq \Delta f$.

\end{definition}

Having defined the necessary terms, we now state our new summation method based on randomized permutations.

\begin{method}[Randomized Permutation]
\label{method:sat-add-RP-ints}

Let $u\in\Vect(T)$ for some type $T$. We define the \emph{random permutation} of $u$ as follows:

\begin{lstlisting}[language = Python, escapechar=|, frame=single]
def randomly_permute(u):
    return |$\Pi$|(u) // per Definition |\ref{defn:rr-transform} |
\end{lstlisting}

\end{method}

\begin{theorem}[Sensitivity of Bounded Sum with RP-based Summation]
\label{thrm:rpsummation}
    Let $T$ be the type of $k$-bit signed integer with saturation arithmetic. Let $\bs^*_{L, U}: \Vect(T_{[L, U]}) \rightarrow T$ be the bounded sum function on a dataset to which RP has been applied (Method~\ref{method:sat-add-RP-ints}), and similarly for $\bs^*_{L, U, n}: T_{[L, U]}^n \rightarrow T$. Then, 
    \begin{enumerate}
        \item (Unknown $n$.) $\Delta_\mathit{Sym} \bs^*_{L, U} = \max \{|L|, U\}$.
        \item (Known $n$.) $\Delta_\mathit{CO} \bs^*_{L, U, n} = (U-L)$.
    \end{enumerate}
\end{theorem}

We first show that the idealized sensitivity is preserved if an \emph{ordered} distance metric is used. Then, we show that the randomized permutation allows datasets at some unordered distance $c$ to be treated as if they are at ordered distance $c$.
\subsubsection{Sensitivities with Respect to Ordered Distance for Saturating Integers}
\label{sec:rp-sat-ints}

We first show that the implemented sensitivity with respect to ordered distance metrics is equal to the idealized sensitivity.

\begin{theorem}[Ordered Distance Sensitivities]
\label{thrm:ord-dist-saturation-sens}

Let $T$ be the type of $k$-bit integers with saturation arithmetic, let $\bs^{*I}_{L,U}: \Vect(T) \rightarrow T$ be the iterative bounded sum function from Definition~\ref{def:bsiterative}, and similarly for $\bsfp_{L, U, n}: T^n \rightarrow T$. Then,
\begin{enumerate}
    \item (Unknown $n$.) $\Delta_\mathit{ID} \bs^*_{L, U} = \max \{ |L|, U\}$.
    \item (Known $n$.) $\Delta_\mathit{Ham} \bs^*_{L, U, n} = (U-L)$.
\end{enumerate}
\end{theorem}

Before we begin with the proof of Theorem \ref{thrm:ord-dist-saturation-sens}, we provide two lemmas.

\begin{lemma}
\label{lemma:sat-inequality}

Let $T$ be the type of $k$-bit integers. Then, for all $y,z,z'\in T$, $$|(y\boxplus z) - (y\boxplus z')| \leq |z-z'|.$$
\end{lemma}

\begin{proof}

Consider the cases where saturation does not occur (i.e., where $\min(T)\leq y+z\leq \max(T)$, and $\min(T)\leq y+z'\leq \max(T)$). Then $| (y\boxplus z) - (y\boxplus z') | = |(y+z) - (y+z')| = |z-z'|$.

Now, consider the case where exactly one of the sums saturates, say $y\boxplus z$. Then, $|y\boxplus z| < |y+z|$, and thus
\[
    |(y\boxplus z) - (y\boxplus z')| \leq |(y+z) - (y\boxplus z')| \leq |(y+z) - (y+z')| = |z-z'|.
\]
Finally, if both $y \boxplus z$ and $y \boxplus z'$ saturate, then they must both equal $\max(T)$ (if $y >0$) or both equal $\min(T)$ (if $y<0$), so $|(y \boxplus z) - (y \boxplus z')| = 0 \leq |z-z'|$.

\end{proof}

\begin{corollary}
\label{cor:sat-inequality-cor}
    Let $T$ be the type of $k$-bit integers. Then, for all $y,z,z'\in T$, $$|(z\boxplus y) - (z'\boxplus y)| \leq |z-z'|.$$
\end{corollary}

\begin{proof}
Saturation addition is commutative. Therefore, by Lemma \ref{lemma:sat-inequality}, we see that $|(z\boxplus y) - (z'\boxplus y)| \leq |z-z'|$.
\end{proof}

\begin{proof}[Proof of Theorem \ref{thrm:ord-dist-saturation-sens}]

We begin with the proof of (2). 

(2) Let $u\simeq_\mathit{Ham} u'$ be two datasets of length $n$ with integer values of type $T$. Because $u\simeq_\mathit{Ham} u'$, there is at most one value $i\in [n]$ such that $u_{i}\neq u'_{i}$.

We now consider what happens at index $i$. For all $j<i$, we have $u_j = u'_j$, so we see that $$\bsfp_{L,U,{i-1}}[u_1,\ldots, u_{i-1}] = \bsfp_{L,U,{i-1}}[u'_1,\ldots, u'_{i-1}].$$ By Lemma \ref{lemma:sat-inequality}, we see that 
\begin{equation}
\label{ex:sat-ham-ineq}
|\bsfp_{L,U,i}[u_1,\ldots, u_{i}] -  \bsfp_{L,U,i}[u'_1,\ldots, u'_i] | \leq |u_i - u'_i| \leq U-L.
\end{equation}

Note that, for $j>i$, we have $u_j = u'_j$. By \cref{cor:sat-inequality-cor} and by induction, the sums will get no further apart from each other than they were in the expression above. Therefore, $|\bsfp_{L,U,n}(s) - \bsfp_{L,U,n}(s')| \leq (U-L)$. By the logic used in the proof of Theorem~\ref{thrm:reals-sens}, we see that this is also a lower bound on the sensitivity, so $$\sham \bsfp_{L,U,n} = (U-L).$$

(1) Let $u\simeq_\mathit{ID} u'$ be two datasets with integer values of type $T$. There is at most one index $i^*$ such that, without loss of generality, $u_i = u'_i$ for all $i<i^*$, and $u_j = u'_{j+1}$ for all $j>i^*$. 

Consider the case where there is no such $i^*$. Then $u=u'$, so $\bsfp_{L,U}(u) = \bsfp_{L,U}(u')$ and $|\bsfp_{L,U}(u) = \bsfp_{L,U}(u')| = 0$.

Now, consider the case where there is such an $i^*$.
We note that, for all $y\in T$, we have $y\boxplus 0 = y$. Therefore,inserting a value $u_{i^*} = 0$ will not affect the final sums. It will also provide the effect that, for all $i\neq i^*$, $u_i = u'_i$.

We can now borrow the logic from \cref{ex:sat-ham-ineq} to see that
\begin{equation}
|\bsfp_{L,U}[u_1,\ldots, u_{i}] -  \bsfp_{L,U}[u'_1,\ldots, u'_i] | \leq |u_{i^*} - u'_i| \leq \max\{|L|,U\}.
\end{equation}
By \cref{cor:sat-inequality-cor} and by induction, the sums will get no further apart from each other than they were in the expression above. Therefore, $|\bsfp_{L,U}(u) - \bsfp_{L,U}(u')| \leq \max\{|L|,U\}$. By the logic used in the proof of Theorem~\ref{thrm:reals-sens}, we see that this is also a lower bound on the sensitivity, so $$\sid \bsfp_{L,U} = \max\{|L|,U\}.$$

\end{proof}

\subsubsection{Randomized Permutation is a Transformation between Distance Metrics}
\label{sec:rp-transform-dist}

Having shown that useful sensitivities can be established for ordered distance metrics, we show that randomized permutation provides a transformation from unordered distance metrics to ordered distance metrics. Using randomized permutation in combination with the theorems in Section \ref{sec:rp-sat-ints} will then allow us to establish useful sensitivities for unordered distance metrics.

\begin{theorem}[Coupling, Unknown $n$]
\label{thrm:couple-unknown-n}
    Let $u\in\mathcal{D}^m$ and $u'\in\mathcal{D}^n$. Additionally, let $\Pi: \mathcal{D}^m \rightarrow \mathcal{D}^m, \Pi': \mathcal{D}^n \rightarrow \mathcal{D}^n$ be randomized permutations of datasets $u$ and $u'$, respectively. For all datasets $u, u'$ such that $u\simeq_\mathit{Sym} u'$, there exists a coupling $(\tilde{s}_u, \tilde{s}_{u'})$ of the datasets $s_u = \Pi(x)$ and $s_{u'} = \Pi'(u')$ such that $\tilde{s}_u \simeq_{ID} \tilde{s}_{u'}$. 
\end{theorem}

\begin{proof}

Without loss of generality, let $\len(u') = \len(u) + 1$. (In the $u\simeq_\mathit{Sym}u'$ case, we could have $\len(u) = \len(u')$, but then we would have a trivial coupling where each ordering $\Pi(u)$ is matched with the same ordering of elements $\Pi'(u')$, so we only consider the $|\len(u) = \len(u') + 1$ case.) We can then think of $u'$ as being $u$, but with one ``insertion''.

We now describe a way to enumerate $u$ and $u'$ so that equivalent terms are matched. Let $\mathcal{J}[i]$ be the $i^\mathit{th}$ term in a sequence $\mathcal{J}$. Because $\dsym(u,u') = 1$, we know that $u$ and $u'$ are identical from an unordered perspective, except $u'$ has one additional row.  This means, then, that there is some sequence $\mathcal{J}$ of terms $\{1,\ldots, m+1 \}$ such that, for all $i\in\{1,\ldots,m\}$, we have $u_i = u'_{\mathcal{J}[i]}$. As stated, there will be an unmatched $u'$ term since $n| = m + 1$; we call this term $u'_{\mathcal{J}[m+1]}$.

We now use this enumeration scheme to show that, for every permutation $\Pi(u)$, there is a permutation $\Pi'(u')$ such that $\did(\Pi(u),\Pi'(u')) = 1$. We first construct an intermediate ordering $u''$. Let $\Pi(u)_k$ denote the $k^\mathit{th}$ value in the vector $\Pi(u)$. Suppose that $\Pi(u)_k = u_i$. We pair this with $u_{\mathit{J}[i]}$. We proceed in this way to construct a dataset $u''$ where $\did(\Pi(u), u'') = 0$.

Note that $u''$ does not include the leftover term $u'_{\mathcal{J}[m+1]}$. This means, then, that there are $(m+1)$ permutations $\Pi'(u')$ such that $\did(\Pi(u),\Pi'(u')) = 1$. This construction applies to every $\Pi(u)$. Therefore, every $\Pi(u)$ has $m+1$ permutations $\Pi'(u')$ such that $\did(\Pi(u),\Pi'(u')) = 1$. 

Recall that, without loss of generality, we are assuming that $m\geq n$. This means that $u$ has $m!$ permutations, while $u'$ has $n!$ permutations. Note that our method pairs each permutation of $u$ with $(|u|+1)$ permutations of $u'$, and the mapping from permutations of $u$ to permutations of $u'$ is injective, with its inverse surjective. Because RP samples over permutations uniformly at random, we have shown that there is a coupling of randomness $(\tilde{s}_u, \tilde{s}_{u'})$ between $s_u = \Pi(u)$ and $s_{u'} = \Pi(u')$ such that $\tilde{s}_u \simeq_{ID} \tilde{s}_{u'}$.

\end{proof}

\begin{corollary}[Coupling, Known $n$]\label{cor:coupling-known-n}
    Let $\Pi: \mathcal{D}^n \rightarrow \mathcal{D}^n$ be a randomized permutation of datasets of length $n$. For all datasets $u, u'\in\mathcal{D}^n$ such that $u\simeq_\mathit{CO} u'$, there exists a coupling $(\tilde{s}_u, \tilde{s}_{u'})$ of the datasets $s_u = \Pi(u)$ and $s_{u'} = \Pi(u')$ such that $\tilde{s}_x \simeq_\mathit{Ham} \tilde{s}_{u'}$.
\end{corollary}

\begin{proof}

Suppose $u\simeq_\mathit{CO} u'$. By Theorem \ref{lemma:dcotodsym}, this means that $\dsym(u,u') = 2$. Using the path property of $\dsym$ (Lemma~\ref{lemma:metrics-path-prop}) and by Theorem \ref{thrm:couple-unknown-n}, then, there is a coupling of randomness such that $\did(\Pi(u), \Pi(u')) \leq 2$ \cite{sv21}. By Theorem \ref{lemma:dhamtodid}, then, $\dham(\Pi(u), \Pi(u')) \leq 1$, so there is a coupling of randomness $(\tilde{s}_x, \tilde{s}_{u'})$ between $s_u = \Pi(u)$ and $s_{u'} = \Pi(u')$ such that $\tilde{s}_u \simeq_\mathit{Ham} \tilde{s}_{u'}$.

\end{proof}

We remark that the proofs of Theorem~\ref{thrm:couple-unknown-n} and Corollary~\ref{cor:coupling-known-n} are \emph{constructive}, and so not only do we show that the coupling \emph{exists}, but we also show how to find it efficiently.

\subsubsection{Sensitivities for Unordered Distance}

Now that we have proven the coupling between unordered distance metrics and ordered distance metrics, we can finally prove the sensitivity of bounded sum with RP-based summation; i.e., Theorem~\ref{thrm:rpsummation}.

\begin{proof}[Proof of Theorem~\ref{thrm:rpsummation}]

We now prove each part of Theorem~\ref{thrm:rpsummation}, which shows that the implemented sensitivities for Method \ref{method:sat-add-RP-ints} match the idealized sensitivities of bounded sum, with respect to unordered distances between datasets.

\begin{enumerate}
    \item Method \ref{method:sat-add-RP-ints} first applies RP to the input dataset, and Theorem \ref{thrm:couple-unknown-n} tells us that RP transforms datasets that are neighboring with respect to $\dsym$ to a coupling of datasets that are neighboring with respect to $\did$. We can now work with datasets that are neighboring with respect to $\did$. Theorem \ref{thrm:ord-dist-saturation-sens} tells us that $\Delta_\mathit{ID}\bsfp_{L,U}\leq \max\{|L|,U\}$, so we must also have $\Delta_\mathit{Sym}\bsfp_{L,U}\leq \max\{|L|,U \}$. By the logic used in the proof of Theorem~\ref{thrm:reals-sens}, we see that this is also a lower bound on the sensitivity, so $$\Delta_\mathit{Sym}\bsfp_{L,U} = \max\{|L|,U \},$$ completing the proof.

    \item Method \ref{method:sat-add-RP-ints} first applies RP to the input dataset, and Corollary \ref{cor:coupling-known-n} tells us that RP transforms datasets that are neighboring with respect to $\dco$ to a coupling of datasets that are neighboring with respect to $\dham$. We can now work with datasets that are neighboring with respect to $\dham$. Theorem \ref{thrm:ord-dist-saturation-sens} tells us that $\Delta_\mathit{Ham}\bsfp_{L,U,n}\leq (U-L)$, so we must also have $\Delta_\mathit{CO}\bsfp_{L,U,n}\leq (U-L)$. By the logic used in the proof of Theorem~\ref{thrm:reals-sens}, we see that this is also a lower bound on the sensitivity, so $$\Delta_\mathit{CO}\bsfp_{L,U,n} = U-L,$$ completing the proof.

\end{enumerate}

\end{proof}

\subsection{Sensitivity from Accuracy}\label{sec:accuracy}

In this section, we show a relationship between an upper bound on a function's implemented sensitivity, and the function's idealized sensitivity and accuracy. The sensitivities that we derive apply directly to an unaltered implementation of iterative summation. The accuracy of a summation function is dependent on the number of terms being summed (meaning the sensitivities we derive are dependent on the number of terms being summed), so we also introduce a method called \textit{truncated summation} that allows these sensitivities to apply to the unbounded DP setting in which the number of terms being summed is not necessarily known to the data analyst.

Because these upper bounds on implemented sensitivity depend on a function's accuracy, we also present methods from the numerical analysis literature for improving the accuracy of a summation function, and we compute the resulting sensitivities for these methods.

\begin{lemma}[Relating Sensitivity and Accuracy]\label{lemma:sa}
$$\begin{aligned}
\MoveEqLeft \Delta \bsfp \leq \\
& \max_{u\simeq u'}\left(|\bs(u) - \bs(u')| + |\bsfp(u) - \bs(u)| + |\bsfp(u') - \bs(u')|\right).
\end{aligned}
$$
\end{lemma}

\begin{proof}
    This follows immediately from two applications of the triangle inequality.
\end{proof}

The term $|\bs(u) - \bs(u')|$ is bounded by the idealized sensitivity. The terms $|\bsfp(u) - \bs(u)|$ and $|\bsfp(u') - \bs(u')|$ can be bounded by the accuracy of the function $\bsfp$ as compared to its idealized counterpart $\bs$. We follow Higham, Wilkinson's, and Kahan's work on the accuracy of the sum function \cite{h93, w60, kahan1965pracniques} for derivations of the accuracies of various summation strategies, including iterative summation. 

\begin{theorem}[Accuracy of Summation Algorithms \cite{h93, w60, kahan1965pracniques}]
\label{thrm:acc-upper-bounds}

Let $\bsfp_{L, U,n}: T^n_{[L, U]} \rightarrow T$ be an implementation of the bounded sum function on the normal $(k, \ell)$-bit floats of type $T$, and let $\bs_{L,U,n}: \mathbb{R}^n_{[L, U]} \rightarrow \mathbb{R}$ compute the idealized bounded sum.
\begin{enumerate}
    \item If $\bsfp$ corresponds to iterative summation (Definition~\ref{def:bsiterative}), then for all $u\in T^n_{[L,U]}$
    \[
        \left|\bsfp_{L, U, n}(u) - \bs_{L,U,n}(u)\right| \leq \frac{n^2}{2^{k+1}}\cdot \max\{|L|,U\}.
    \]
    \item If $\bsfp$ corresponds to pairwise summation (Definition~\ref{def:bspairwise}) and $n<2^k$, then for all $u\in T^n_{[L,U]}$
    \[
        \left|\bsfp_{L, U, n}(u) - \bs_{L,U,n}(u)\right| \leq O\left(\dfrac{n \log n}{2^k}\right)\cdot \max\{|L|,U\}.
    \]
    \item If $\bsfp$ corresponds to Kahan summation (Definition~\ref{def:bskahan}) and $n<2^k$, then for all $u\in T^n_{[L,U]}$
    \[
        \left|\bsfp_{L, U, n}(u) - \bs_{L,U,n}(u)\right| \leq O\Big(\dfrac{n}{2^k}\Big)\cdot \max\{|L|,U\}.
    \]
\end{enumerate}

\end{theorem}

\begin{proof}
We show each bound for the corresponding sum function separately. For all of the proofs below, let $U\geq |L|$; a symmetric result follows for $U<|L|$, which allows us to get the bounds stated in the theorem.

\begin{enumerate}
    \item 
    We first show that 
    \[
        |\bsfp(u) - \bs(u)| \leq C_{n, k} \cdot \sum_{i=1}^n |u_i| \leq C_{n, k} \cdot n U,
    \]
    where
    \[
        C_{n, k} = \Big( 1 + \dfrac{1}{2^{\mantlen+1}} \Big)^{n-1} - 1 \approx \dfrac{n-1}{2^{\mantlen+1}}
    \]
    for $n<2^k$. We proceed by induction on $n$. Hence assume that the proposition holds for $n-1$, and consider the addition step $\bsfp(u) = \bsfp(u_1, \ldots, u_{n-1}) \oplus u_n$. Then,
    \[
        |\bsfp(u) - \bs(u)| \leq |\bsfp(u) - (\bsfp(u_1, \ldots, u_{n-1}) + u_n)| +
    \]
    \[
        |\bsfp(u_1, \ldots, u_{n-1} + u_n) - \bs(u_1, \ldots, u_{n-1}) + u_n|.
    \]
    Let $s^* = \bsfp(u_1, \ldots, u_{n-1}), s = \bs(u_1, \ldots, u_{n-1})$. Then,
    $$\begin{aligned}
    \MoveEqLeft
        |\bsfp(u) - \bs(u)|\\
        &= |\br(s^*+u_n) - (s^*+u_n)| + |s^*-s|\\
        &\leq \dfrac{|s^*+u_n|}{2^{\mantlen+1}} + |s^*-s| \\
        &\leq \dfrac{|s^*|+|u_n|}{2^{\mantlen+1}} + \Big( 1 + \dfrac{1}{2^{\mantlen + 1}} \Big) \cdot |s^*-s| \\
        &\leq \dfrac{\sum_{i=1}^n |u_i|}{2^{\mantlen+1}} + \Big( 1 + \dfrac{1}{2^{\mantlen+1}} \Big) \cdot |s^*-s| \\
        &\leq \Big( \dfrac{1}{2^{\mantlen+1}} + \Big(1 + \dfrac{1}{2^{\mantlen+1}} \Big) C_{n-1} \Big) \cdot \sum_{i=1}^n |u_i|.
    \end{aligned}
    $$
    Since 
    \[
        C_n = \dfrac{1}{2^{\mantlen+1}} + \Big( 1 + \dfrac{1}{2^{\mantlen+1}} \Big) C_{n-1},
    \]
    the induction step concludes. We remark that a similar accuracy bound is derived in \cite{h93} and \cite{kahan1965pracniques}.
    
    \item We employ the accuracy bound shown in \cite{h93}. Let $t = 1/2^{\mantlen+1}$ (sometimes referred to as ``unit roundoff'' or ``machine epsilon). Then, \cite[Eq. 3.6]{h93} shows that
    \[
        |\bsfp(u) - \bs(u)| \leq \dfrac{\log_2 n t}{1-\log_2 n t} \sum_{i=1}^n |v_i|.
    \]
    In our setting, we have an upper bound on $\sum_{i=1}^n |v_i|$, and hence we can write
    \[
        |\bsfp(u) - \bs(u)| \leq \dfrac{\log_2 n t}{1-\log_2 n t} \cdot nU.
    \]

    \item We employ the accuracy bound shown in \cite{h93}. Let $t = 1/2^{\mantlen + 1}$. Then, \cite[Eq. 3.11]{h93} shows that
    \[
        |\bsfp(u) - \bs(u)| \leq (2t +O(nt^2)) \sum_{i=1}^n |v_i|
    \]
    In our setting, we have an upper bound on $\sum_{i=1}^n |v_i|$ (namely, $nU$), and hence we can write
    \[
        |\bsfp(u) - \bs(u)| \leq (2t +O(nt^2)) \cdot nU.
    \]

\end{enumerate}
\end{proof}

\begin{theorem}[Sensitivity upper bounds]
\label{thrm:sens-upper-bounds}
    Let $\bsfp_{L, U,n}: T^n_{[L, U]} \rightarrow T$ be a bounded sum function on the normal $(k, \ell)$-bit floats of type $T$. Additionally, let $-U\leq L < U$. Then,
    \begin{enumerate}
        \item If $\bsfp$ corresponds to iterative summation (Definition~\ref{def:bsiterative}), then
        \[
            \sco \bsfp_{L, U, n} \leq (U-L) + \Big( \dfrac{n^2}{2^k} \Big)\cdot \max\{|L|,U\}.
        \]
        \item If $\bsfp$ corresponds to pairwise summation (Definition~\ref{def:bspairwise}) and $n<2^k$, then
        \[
            \sco \bsfp_{L, U, n} \leq (U-L) + O\left(\dfrac{n \log n}{2^k}\right) \cdot \max\{|L|,U\}.
        \]
        \item If $\bsfp$ corresponds to Kahan summation (Definition~\ref{def:bskahan}) and $n<2^k$, then
        \[
            \sco \bsfp_{L, U, n} \leq (U-L) + O\Big(\dfrac{n}{2^k}\Big) \cdot \max\{|L|,U\}.
        \]
    \end{enumerate}

\end{theorem}

\begin{proof}
    The proof of each part follows from \cref{lemma:sa} and the corresponding part of \cref{thrm:acc-upper-bounds}.
\end{proof}

In the cases where $n$ is unknown, we can still make use of the upper bounds on the sensitivity that we showed in \cref{thrm:sens-upper-bounds} by introducing a truncation method. To that end, we introduce the definition of the \textit{truncated bounded sum} method.

\begin{method}[Truncated Bounded Sum]\label{method:trunc}

Let $u$ be a dataset of integers of type $T_{[L, U]}$, $U$ an upper bound on the largest number in the dataset, $L$ a lower bound on the smallest number in the dataset, $n_{\max}$ a point at which to truncate the summation, and $\bsfp_{L,U}$ any summation method. Then, the \textit{truncated bounded sum} $\bs^{**}_{L,U,n_{\max}}:\Vect(T_{[L,U]})\to T$ returns
\begin{itemize}
    \item $\bsfp_{L, U}(u)$ if $\len(u)\leq n_{\max}$.
    \item $\bsfp_{L, U}[u_1,\ldots, u_{n_{\max}}]$ if $\len(u) > n_{\max}$.
\end{itemize}
\end{method}

\begin{theorem}[Idealized Sensitivity of Truncated Bounded Sum]
\label{thrm:ideal-trunc-sens}

When the summation method $\bs^{**}_{L,U,n_{\max}}:\Vect(T_{[L,U]})\to T$ corresponds to the idealized bounded sum $\bs_{L,U}$ (Definition~\ref{defn:bounded-sum}), the truncated bounded sum $\bs^{**}_{L,U,n_{\max}}:\Vect(T_{[L,U]})\to T$ as defined in Method~\ref{method:trunc} has sensitivity $\sid \bs^{**}_{L,U,n_{\max}} = \max\{|L|,U,U-L\}$.

\end{theorem}
\begin{proof}
    Consider all pairs of neighboring datasets $u\nid u'$. There are two possible cases: (1) Truncation does not occur for either dataset. (2) Truncation occurs for at least one dataset. 
    \begin{enumerate}
        \item In case 1, we consider the standard bounded sum over a subset of all pairs of datasets $u\nid u'$, so the idealized sensitivity can be bounded by $\max\{|L|,U\}$ due to \cref{thrm:reals-sens}.
        \item In case 2, we consider the standard bounded sum over all pairs of datasets $u\nham u'$, where $u,u'\in T^{n_{\max}}_{[L,U]}$, so the idealized sensitivity can be bounded by $U-L$ due to \cref{thrm:reals-sens}.
    \end{enumerate}
    Therefore, $\sid \bs^{**}_{L,U,n_{\max}}\leq \max\{|L|,U,U-L\}$. We next prove that $\sid \bs^{**}_{L,U,n_{\max}}\geq \max\{|L|,U,U-L\}$. Consider the datasets $u = [U_1,\ldots,U_{n_{\max}},L]$ and $u'= [U_1,\ldots,U_{n_{\max}-1},L]$. We see that $u\nid u'$. We also see that $\bs^{**}_{L,U,n_{\max}} = U-L$, and $\bs^{**}_{L,U,n_{\max}+1} = \max\{|L|,U\}$.
    
    Therefore, $\sid \bs^{**}_{L,U,n_{\max}} = \max\{|L|,U,U-L\}$.
\end{proof}

\begin{corollary}
\label{cor:trunc-sens}
Let $\bs^{**}_{L, U, n_{\max}}: \Vect(T_{[L, U]}) \rightarrow T$ be the bounded sum function implemented with Method~\ref{method:trunc} on normal $(k, \ell)$ floats. Then,
\begin{enumerate}
    \item If Method~\ref{method:trunc} is implemented with iterative summation, then
        \[
        \begin{aligned}
            \sid \bs^{**}_{L, U, n_{\max}}
            &\leq \Big( 1 + \dfrac{n_{\max}^2}{2^k} \Big)\cdot \max\{|L|, U, U-L\} \\
            &\leq \Big( 2 + \dfrac{n_{\max}^2}{2^k} \Big)\cdot \max\{|L|, U\}.
        \end{aligned}
        \]
        
    \item If Method~\ref{method:trunc} is implemented with pairwise summation, then
        \[
        \begin{aligned}
            \sid \bs^{**}_{L, U, n_{\max}}
            &\leq \Big( 1 + O\left(\dfrac{n_{\max} \log n}{2^k}\right) \Big)\cdot \max\{|L|, U, U-L\} \\
            &\leq \Big( 2 + O\left(\dfrac{n_{\max} \log n}{2^k}\right) \Big)\cdot \max\{|L|, U\}.
        \end{aligned}
        \]
        
    \item If Method~\ref{method:trunc} is implemented with Kahan summation, then
        \[
        \begin{aligned}
            \sid \bsfp_{L, U, n_{\max}}
            &\leq \Big( 1 + O\Big(\dfrac{n_{\max}}{2^k}\Big) \Big)\cdot \max\{|L|, U, U-L\} \\
            &\leq \Big( 2 + O\Big(\dfrac{n_{\max}}{2^k}\Big) \Big)\cdot \max\{|L|, U\}.
        \end{aligned}
        \]
        
\end{enumerate}
\end{corollary}

\begin{remark}

We note that, where $U,L\geq 0$ or $U,L\leq 0$, we have $U-L\leq \max\{|L|,U\}$, which gives us the following improvements on the sensitivity bounds provided in \cref{cor:trunc-sens}.

\begin{enumerate}
    \item If Method~\ref{method:trunc} is implemented with iterative summation, then
        \[
            \sid \bs^{**}_{L, U, n_{\max}}
            \leq \Big( 1 + \dfrac{n_{\max}^2}{2^k} \Big)\cdot \max\{|L|, U\}.
        \]
        
    \item If Method~\ref{method:trunc} is implemented with pairwise summation, then
        \[
            \sid \bs^{**}_{L, U, n_{\max}}
            \leq \Big( 1 + O\left(\dfrac{n_{\max} \log n}{2^k}\right) \Big)\cdot \max\{|L|, U\}.
        \]

    \item If Method~\ref{method:trunc} is implemented with Kahan summation, then
        \[
            \sid \bsfp_{L, U, n_{\max}}
            \leq \Big( 1 + O\Big(\dfrac{n_{\max}}{2^k}\Big) \Big)\cdot \max\{|L|, U\}.
        \]

\end{enumerate}

\end{remark}

\color{black}

Note that, although \cref{cor:trunc-sens} is stated in terms of ordered sensitivities, pre-processing datasets with RP (\cref{method:sat-add-RP-ints}) will yield the same sensitivity bounds with respect to unordered distance metrics. Additionally, we can also employ a truncation strategy in the known $n$ case if we want to establish an upper bound on the sensitivity of the implemented bounded sum function.

Data analysts should carefully choose $n_{\max}$. Because the implemented sensitivity grows with the length of the dataset, a larger $n_{\max}$ means that the result will be noisier. However, if $n_{\max}\ll n$, causing many values of the dataset to be truncated, then the statistic may not be meaningful to the data analyst. Often, setting $n_{\max}$ to an unrealistically large value can ensure that the sum is computed over the entire dataset while  offering a minimal blow-up in the implemented sensitivity; for example, when working with 64-bit floats, the accuracy term will be less than 1 times the idealized sensitivity for $n_{\max} < 67~\text{million}$ (meaning that the majority of the implemented sensitivity will be attributable to the idealized sensitivity). In some cases, a noisy count can be helpful for choosing $n_{\max}$.

\subsubsection{Shifting Bounds}
\label{sec:shift-bounds}

We note that, when working in the bounded DP (known $n$) setting with $|U|,|L| \gg U-L$, the sensitivity provided in \cref{thrm:sens-upper-bounds} may be much larger than the idealized sensitivity $\sham \bs_{L,U,n} = \sco \bs_{L,U,n} = U-L$. (To reduce notation, the following discussion considers the case where $U,L>0$; a symmetric method can be used when $U,L<0$.) We can achieve a sensitivity much closer to the idealized sensitivity by subtracting $L$ from every element of the dataset so that all elements lie in the interval $[L'=0,U'=U-L]$. We can then apply our solutions from \cref{sec:accuracy}, using 0 and $U'$ in the sensitivity bound, and adding $L\cdot n$ at the end as post-processing.  That is, we convert a method $\bsfp$ from \cref{sec:accuracy} into a method $\bs^{**}$ with sensitivity depending on $U-L$ as follows:
$$\bs^{**}_{L',U',n}(u) = \bs^*_{L,U,n}(u_1-L,u_2-L,\ldots,u_n-L)+L\cdot n,$$
where a noisy version of the $\bsfp_{L,U,n}$ term is computed first and $L\cdot n$ is added as a post-processing step.

Per \cref{cor:trunc-sens}, then, this results in a sensitivity of
\[
        \begin{aligned}
            \sid \bs^{**}_{L', U', n_{\max}}
            &\leq \Big( 1 + \dfrac{n_{\max}^2}{2^k} \Big)\cdot \max\{|L'|, U', U'-L'\} \\
            &\leq \Big( 1 + \dfrac{n_{\max}^2}{2^k} \Big)\cdot (U-L).        \end{aligned}
        \]

\subsection{Split Summation with Random Permutation}\label{sec:splitsumrp}

In this section, we present our method for a bounded sum function on floats that recovers sensitivity within a constant factor of the idealized sensitivity. While the attack presented in \cref{sec:accum-round-floats-ii} shows that the application of RP alone does not solve issues related to floating point arithmetic, we will show that we can achieve a good sensitivity for bounded sum by combining RP with our previous method of split summation and by changing the rounding mode from banker's rounding to round toward 0 (see Definition~\ref{defn:roundto0}).

Round toward 0 allows us to achieve tighter bounds due to the property that, in many cases, the rounding never leads to an increase in the magnitude of the sum. We use this to argue that the sums of adjacent datasets rarely diverge. On the other hand -- as we saw in the counterexamples in Sections \ref{sec:round-attack-floats} and \ref{sec:accum-round-floats} -- banker's rounding can make the magnitude of the sum increase much more quickly or much less quickly than expected, causing similar datasets to sometimes produce surprisingly dissimilar sums.

\begin{method}[Split Summation on Floats]
\label{method:fsplit-sum}

Let $U$ denote the upper bound, $L$ the lower bound, and $u$ the input dataset with elements of type $T$ with round toward zero. We define \texttt{fsplit\_sum} as follows:

\begin{lstlisting}[language = Python, escapechar=|,frame=single]
def fsplit_sum(u)
    P = 0
    N = 0
    for elt in s:
        if elt < 0:
            N = N + elt
        else:
            P = P + elt
    return P |$\oplus$| N 
\end{lstlisting}

\end{method}

There are two important features to note in the algorithm above.
\begin{enumerate}
    \item The iterative sums $N$ and $P$ are computed using round toward zero as defined in Definition~\ref{defn:roundto0}.
    \item The returned sum $P\oplus N$ is computed using standard banker's rounding as defined in Definition~\ref{defn:bankers-rounding}. (This is the only use of banker's rounding in the algorithm.)
\end{enumerate}

\begin{remark}[Avoiding Overflow]
Let \texttt{max} be the largest float $<\texttt{inf}$, and let $\texttt{min}$ be the smallest float $>-\texttt{inf}$. To prevent overflow from occurring (which is necessary, as described in \cref{sec:61floats}), we do the following. If $P=\texttt{inf}$, we set $P = \texttt{max}$, and if $N = -\texttt{inf}$, we set $N = \texttt{min}$. Issues related to non-associativity will not arise because we are using split summation and issues of non-associativity do not affect split summation (see \cref{sec:splitsum}).

\end{remark}

We now state a theorem about the sensitivity of the function described in \cref{method:fsplit-sum}. 

\begin{theorem}[$\ssym$ and $\sco$ of \cref{method:fsplit-sum}]
\label{thrm:rp-split-summation}
    Let $\bs^*_{L, U}: \Vect(T) \rightarrow T$ be the bounded sum function with split summation and round toward zero (\cref{method:fsplit-sum}) on the normal $(\mantlen,\explen)$-bit floats of type $T$, and similarly for $\bs^*_{L, U, n}: T^n \rightarrow T$. Define $f_{L,U} = \bsfp_{L,U} \circ \rp$ and $f_{L,U,n} = \bsfp_{L,U,n}\circ \rp$. Then,
    \begin{enumerate}
        \item (Unknown $n$.) Let $b_U = \lfloor \log_2 U\rfloor, b_L = \lfloor \log_2|L|\rfloor$.
        Then, 
        $$
        \begin{aligned}
        \ssym f_{L, U} &\leq \max\{2^{b_U} + U, 2^{b_L} + |L|\} +\max\{2^{b_U},2^{b_L} \}\\
        &\leq 3\cdot\max\{U,|L|\}.
        \end{aligned}
        $$
        
        \item (Known $n$.) Let $c_U = \lfloor \log_2 U\rfloor$, $b_U = \min\{c_U, \lfloor \log_2(nU)\rfloor - k\}$,
        $c_L = \lfloor \log_2 |L|\rfloor$, $b_L = \min\{c_L, \lfloor \log_2(|nL|\rfloor - k\}$.
        
        Then,
        $$
        \begin{aligned}
        \MoveEqLeft
        \sco f_{L, U, n}\\
        &\leq 2^{b_U} + U + 2^{b_L} + |L|\\
        &\quad +\max\{\min\{2^{b_U},2^{c_U+1}\},\min\{2^{b_L}, 2^{c_L+1}\}\}\\
        &\leq 5\cdot\max\{U,|L|\}.
        \end{aligned}
        $$
    \end{enumerate}
\end{theorem}

Recall, from Theorem~\ref{thrm:couple-unknown-n} for unknown dataset length and Corollary~\ref{cor:coupling-known-n} for known dataset length, that RP allows us to reason about the ordered distance between datasets, even when we are only provided with a guarantee about the unordered distance between the datasets. For this reason, we state a theorem about the sensitivities of $\bsfp_{L,U}$ and $\bsfp_{L,U,n}$ with respect to ordered distance metrics, since sensitivities with respect to ordered distance metrics will be able to be converted to statements about sensitivities with respect to unordered distances.

\begin{theorem}[Ordered Sensitivity of Split Summation with RP for Floats]
\label{thrm:ord-toward-zero-known}
\label{thrm:toward-zero-known}
    Let $\bs^*_{L, U}: \Vect(T) \rightarrow T$ be the bounded sum function with split summation and round toward zero (\cref{method:fsplit-sum}) on the normal $(\mantlen,\explen)$-bit floats of type $T$, and similarly for $\bs^*_{L, U, n}: T^n \rightarrow T$. Then,
\begin{enumerate}
        \item (Unknown $n$.) Let $b_U = \lfloor \log_2 U\rfloor, b_L = \lfloor \log_2|L|\rfloor$.
        Then, 
        $$
        \begin{aligned}
        \sid \bsfp_{L, U} &\leq \max\{2^{b_U} + U, 2^{b_L} + |L|\} +\max\{2^{b_U},2^{b_L} \} \\
        &\leq 3\cdot\max\{U,|L|\}.
        \end{aligned}
        $$
        \item (Known $n$.) Let $c_U = \lfloor \log_2 U\rfloor$, $b_U = \min\{c_U, \lfloor \log_2(nU)\rfloor - k\}$,
        $c_L = \lfloor \log_2 |L|\rfloor$, $b_L = \min\{c_L, \lfloor \log_2(|nL|\rfloor - k\}$.
        Then,
        $$
        \begin{aligned}
        \MoveEqLeft
        \sham \bsfp_{L, U, n}\\
        &\leq 2^{b_U} + U + 2^{b_L} + |L|+\max\{\min\{2^{b_U},2^{c_U+1}\},\\
        &\quad \min\{2^{b_L}, 2^{c_L+1}\}\}\\
        &\leq 5\cdot\max\{U,|L|\}.
        \end{aligned}
        $$
\end{enumerate}
\end{theorem}

Before proceeding with the proof of Theorem~\ref{thrm:toward-zero-known}, we state and prove some lemmas about the behavior of floating point arithmetic. The motivation for these lemmas is to determine the situations in which rounding can cause the iterative sums of adjacent datasets to diverge. Specifically, we prove that this divergence-causing rounding only occurs when the intermediate sum moves from a floating point interval of the form $[2^k,2^{k+1})$ for $k\in\mathbb{Z}$ to a floating point interval of the form $[2^\ell,2^{\ell+1})$ where $k\neq\ell\in\mathbb{Z}$.

    \begin{lemma}[Distance Between Sums -- Part 1]
    \label{lemma:dist-rarely-grows}
    
    Let $s,s',v$ be $(\mantlen, \explen)$-bit floats such that $\textit{sign}(s) = \textit{sign}(s') = \textit{sign}(v)$, $|s'| \geq |s|$, and $\ulp(s+v) = \ulp(s)$. Then it follows that $|\rtz{s'+v} - \rtz{s+v}|\leq |s-s'|$.
    \end{lemma}
    
    \begin{proof}
    We present the proof for the case that $s,s',v$ are all positive; the case in which they are all negative is similar. We first consider $\rtz{s+v}$.
    $$
    \begin{aligned}
    \rtz{s+v} &= \Big\lfloor \frac{s+v}{\ulp(s+v)}\Big\rfloor \cdot \ulp(s+v)\\
    &= \Big\lfloor \frac{s+v}{\ulp(s)} \Big\rfloor \cdot \ulp(s)\\
    &= s+ \Big\lfloor\frac{v}{\ulp(s)}\Big\rfloor\cdot \ulp(s),
    \end{aligned}
    $$
    with the final equality following from the fact that $s$ is an integer multiple of $\ulp(s)$.
    
    We now consider $\rtz{s'+ v}$.
    $$
    \begin{aligned}
    \rtz{s'+v} &= \Big\lfloor \frac{s'+v}{\ulp(s'+v)}\Big\rfloor \cdot \ulp(s+v)\\
    &\leq s' + \Big\lfloor \frac{v}{\ulp(s'+v)}\Big\rfloor\cdot \ulp(s'+v)\\
    &\leq s' + \Big\lfloor \frac{v}{\ulp(s)} \Big\rfloor \cdot \ulp(s).
    \end{aligned}
    $$

    Therefore, we see that $|\rtz{s'+v} - \rtz{s+v}|\leq |s-s'|$.
    \end{proof}
    
    Observe that Lemma~\ref{lemma:dist-rarely-grows} is not affected by $\ulp(\rtz{s'+v})$. This is because the spacing between NRFs increases in each interval, so adding $v$ to $s'$ will not have a larger effect RTZ on the sum than adding $v$ to $s$, where $|s'| > |s|$. Note that this does not hold if banker's rounding were used instead of round toward zero.
    
    \begin{lemma}[Distance Between Sums -- Part 2]
    \label{lemma:round-at-boundary-crossing}
    
    If two $(\mantlen, \explen)$-bit floats $s, s'\geq 0$ are such that $s < s'$ (with $s'-s = d$), and we add a term $v > 0$ to both terms such that $\ulp(s + v) = 2\cdot \ulp(s)$, it follows that $\rtz{s'+v} - \rtz{s+v} \leq d + \ulp(s).$

    \end{lemma}
    
    \begin{proof}
    
    Let $s\in [2^{\mantlen + m}, 2^{\mantlen + m})$ for some $m\in\mathbb{Z}$, so $\ulp(s) = \frac{2^{\mantlen + m}}{2^{\mantlen + m}} = 2^m$. Also, suppose that $v\in [p\cdot 2^m, p\cdot 2^{m+1})$ for some $p\in\mathbb{Z}$, and that $2^{\mantlen + m + 1} - s = q\cdot 2^{m}$ for some $q\in\mathbb{Z}$. Additionally, suppose that $s' \geq 2^{\mantlen + m + 1}$.
    
    We now note that $\rtz{s+v} = s + q\cdot 2^{m} + \lfloor \frac{p-q}{2}\rfloor \cdot 2^{m+1}$. We proceed with a proof by cases.

    \begin{enumerate}
        \item Suppose that $p$ is even, but that $(p-q)$ is odd. This means, then, that computing $\rtz{s'+v}$ will yield $\rtz{s'+v}\leq s' + \frac{p}{2}\cdot 2^{m+1} = s' + p\cdot 2^m$. On the other hand, computing $\rtz{s+v}$ will result in a value of $\rtz{s+v} = s + q\cdot 2^m + \frac{p-q-1}{2}\cdot 2^{m+1} = s + q\cdot 2^m + (p-q-1)\cdot 2^m = s + (p-1)\cdot 2^m$.
        
        In this case, then, $\rtz{s' + v} - \rtz{s+v} \leq d + 2^k = d + \ulp(s)$.
    
        \item Now, suppose that $p$ is even and $(p-q)$ is even. This means, then, that computing $\rtz{s'+v}$ will yield $\rtz{s'+v}\leq s' + \frac{p}{2}\cdot 2^{m+1} = s' + p\cdot 2^m$. On the other hand, computing $\rtz{s+v}$ will result in a value of $\rtz{s+v} = s + q\cdot 2^m + \frac{p-q}{2}\cdot 2^{m+1} = s + p\cdot 2^m$.
        
        In this case, then, $\rtz{s' + v} - \rtz{s+v} \leq d$.

        \item Finally, suppose that $p$ is odd. This means, then, that computing $\rtz{s'+v}$ will yield $\rtz{s'+v}\leq \frac{p-1}{2}\cdot 2^{m+1} = s' + (p-1)\cdot 2^m$. On the other hand, computing $\rtz{s+v}$ will result in a value of $\rtz{s+v} = s + q\cdot 2^m + \lfloor\frac{p-q}{2}\rfloor\cdot 2^{m+1} \geq s + q\cdot 2^m + (p-q-1)\cdot 2^m = s + (p-1)\cdot 2^m.$
        
        In this case, then, $\rtz{s' + v} - \rtz{s+v} \leq d$.
    
    \end{enumerate}
    
    Note that, in the cases above, we only considered $s' \geq 2^{\mantlen+m+1}$. We now consider the case where $s,s'\in [2^{\mantlen+m}, 2^{\mantlen+m+1})$. Let $q'$ be the corresponding ``$q$ value'' for $s'$. The worst case arises when $(p-q)$ is odd and $(p-q')$ is even. We then see the behavior provided in case (1), where $p$ is even but $(p-q)$ is odd, so the bounds on the difference as provided in the theorem statement still hold.
    
    Therefore, $\rtz{s' + v} - \rtz{s+v} \leq d + 2^m = d + \ulp(s)$.
    
    \end{proof}
    
    \begin{corollary}
    \label{cor:round-at-boundary-crossing}
    
        If two $(\mantlen, \explen)$-bit floats $s, s'\geq 0$ are such that $s < s'$ (with $s'-s = d$), and we add a term $v > 0$ to both terms such that $\ulp(\rtz{s + v}) = 2^m\cdot \ulp(s)$ for $m\in\mathbb{Z}$, it follows that $\rtz{s'+v} - \rtz{s+v} \leq d + (2^m\cdot\ulp(s)-\ulp(s)).$
    \end{corollary}

    Before proceeding to the next lemma, we introduce a new term.
    \begin{definition}[Boundary]
    \label{defn:boundary}
    
    We refer to a value of the form $2^m$ for $m\in\mathbb{Z}$ as a \textit{boundary}. This is because, for values $s\in[2^m,2^{m+1})$, the distance between NRFs is twice that of the distance between NRFs for values $s\in[2^{m-1},2^m)$. Whenever we have a floating point value $s\in[2^{m},2^{m+1})$, and we add a value $v$ where $\rtz{s+v} \in(2^{m+j},2^{m+j+1})$ for $j\in\mathbb{Z}^+$, we say that the sum \textit{crossed a boundary}.
    
    \end{definition}

    \begin{lemma}[Maximum Boundary Crossing Given $U$]
    \label{lemma:max-boundary-crossing}
    Let $U$ be the upper bound for the bounded sum function over the $(\mantlen,\explen)$-bit floats, and let $b \in \mathbb{Z}$ be such that $U \in [2^b, 2^{b+1})$. Then, in the worst case, the final crossed boundary corresponds to the boundary $2^{\mantlen+b}$, between the intervals $[2^{\mantlen+b-1},2^{\mantlen+b})$ and $[2^{\mantlen+b},2^{\mantlen+b+1})$.
    \end{lemma}

    \begin{proof}
    Under round toward zero, the maximum sum that a dataset can reach with $U \in [2^b, 2^{b+1})$ is $2^{b+1}\cdot 2^{\mantlen}$. Therefore, the final boundary that we cross is no greater than the boundary between the intervals $[2^{\mantlen+b-1},2^{\mantlen+b})$ and $[2^{\mantlen+b},2^{\mantlen+b+1})$. Note that, because the maximum sum is $2^{\mantlen+b+1}$, we have not yet \emph{crossed} into the interval $(2^{\mantlen+b+1},2^{\mantlen+b+2})$.

    \end{proof}
    
    \begin{lemma}[Maximum Boundary Crossing Given $U,n$]\label{lemma:max-boundary-known-n}
        Let $U$ be the upper bound for the bounded sum function over the $(\mantlen, \explen)$-bit floats, let $n$ be the length of the dataset, and let $m \in \mathbb{Z}$ be such that $U\cdot n \in [2^m, 2^{m+1})$. Then, in the worst case, the final crossed boundary corresponds to the boundary $2^{m}$, between the intervals $[2^{m-1},2^m)$ and $[2^{m},2^{m+1})$.
    \end{lemma}
    \begin{proof}
        Since we are in the round-toward-zero rounding mode, the maximum sum of a length $n$ dataset with upper bound $U$ is $n\cdot U$. Because we define $m \in \mathbb{Z}$ be such that $U\cdot n \in [2^m, 2^{m+1})$ and $U\cdot n \in [2^m, 2^{m+1})$, we see that $2^m$ is the largest boundary.
    \end{proof}
    
    Having developed lemmas about the behavior of floating point arithmetic under round toward 0, we are now ready to reason about the sensitivity of the bounded sum function described in Method~\ref{method:fsplit-sum}.
    
    We remark that the intervals stated in Lemma~\ref{lemma:max-boundary-crossing} are a worst-case scenario, and thus are not necessarily reached by the sum value.
    
    \begin{lemma}[Sensitivity $\Delta_\mathit{Ham}\bsfp_{0,U,n}$]
    \label{lemma:rpf-sens-0-known-n}
    Suppose we have two length $n$ datasets $u\simeq_\mathit{Ham} u'$ of $(\mantlen,\explen)$-bit floats, with values clipped within the interval $[0,U]$ with $U\in[2^c,2^{c+1})$, where the maximum possible sum (in this case, $\min\{ n\cdot U, 2^{\mantlen+c} \}$) is $\in [2^{\mantlen+b}, 2^{\mantlen+b+1})$. Then, the maximum difference in sums will be 
    $$\Delta_\mathit{Ham}\bsfp_{0,U,n}
    \leq \sum_{i=-\infty}^{\mantlen+b-1} \frac{2^i}{2^{\mantlen}} + U = 2^b + U\leq 2U,$$
    where $2^b \leq \min\{U,\frac{n\cdot U}{2^{\mantlen}} \}$.
    
    \end{lemma}

    \begin{proof}
    
    Let $u\simeq_\mathit{Ham} u'$ be two length $n$ datasets of $(\mantlen, \explen)$-bit floats. 
    In the worst case, this means that $u$ and $u'$ are identical, except they differ at some value $u_i\neq u'_i$. In the worst case, without loss of generality, $u_i = 0$ and $u'_i = U$.
    
    We note that we have  $\bsfp_{0,U}([u'_1,\ldots, u'_{i}]) - \bsfp_{0,U}([u_1,\ldots, u_{i}])\leq U$. 
    
    By Lemma \ref{lemma:dist-rarely-grows}, if for all $j\geq i$ and some $m\in\mathbb{Z}$ we have 
    $$\bsfp_{0,U}( [u_1,\ldots, u_{j-1}] ), \bsfp_{0,U}( [u_1,\ldots, u_{j}] ) \in [2^m,2^{m+1}),$$ then we will continue to have $$\bsfp_{0,U}([u'_1,\ldots, u'_{j}]) - \bsfp_{0,U}([u_1,\ldots, u_{j}])\leq U.$$
    
    However, we may have $\bsfp_{0,U}( [u_1,\ldots, u_{j-1}] )\in [2^m,2^{m+1})$ but $\bsfp_{0,U}( [u_1,\ldots, u_{j}] ) \not\in [2^m,2^{m+1})$. In this case, we apply Lemma \ref{lemma:round-at-boundary-crossing}. This tells us that the rounding that causes the difference in sums to grow only occurs when adding a value to the smaller sum (in this case, adding $u_j$) causes the sum to cross from an interval of values in $[2^m,2^{m+1})$ to an interval of values in $[2^p,2^{p+1})$ for some $p\in\mathbb{Z}>m$. By Lemma \ref{lemma:dist-rarely-grows}, the difference in sums does \emph{not} grow when the smaller sum does \emph{not} cross from one interval into another interval. Therefore, we only experience growth in the difference of sums (recall that the ``true'' difference is $\bs_{0,U}(u') - \bs_{0,U}(u) = U$, so we expect this much of a difference in sums) when these interval crossings occur.
    
    Suppose we are crossing from a value $r$ with $\ulp(r) = 2^q$ to a value $r'$ with $\ulp(r') = 2^{q+1}$. By Lemma $\ref{lemma:round-at-boundary-crossing}$, we experience growth in the difference in sums of at most $2^q$.
    
    By Lemma~\ref{lemma:max-boundary-crossing}, the maximum round-up occurs when crossing from an interval with a $\ulp$ of $2^{b-2}$ to an interval with a $\ulp$ of $2^{b-1}$. Therefore, the greatest growth in the difference in sums that we can experience is $2^{b-1}$. We note that the crossing immediately prior to this can result in a growth in the difference in sums of $2^{b-2}$. Therefore, we can bound the overall growth in difference in sums by $2^{b-1} + 2^{b-2} + 2^{b-3} + \ldots = \sum_{m=-\infty}^{b-1} 2^m = \sum_{m = -\infty}^{\mantlen+b-1} \frac{2^m}{2^\mantlen}$. Because $m$ ranges from $-\infty$ to $(\mantlen+b-1)$, the growth in the difference in sums is bounded by $2 \cdot 2^{b-1} = 2^b$, which is turn bounded by $U$ given that, by assumption, $U < 2^{b+1}$. (Note that the bound provided in Corollary~\ref{cor:round-at-boundary-crossing} would have a smaller effect than the roundings at every boundary crossing that we have here, so we only consider the case described in Lemma~\ref{lemma:round-at-boundary-crossing} when evaluating worst-case behavior.)

    Therefore, the overall difference in sums for two datasets with difference in sums $|\bs_{L,U}(u) - \bs_{L,U}(u')| \leq U$ is at most the bound provided in the lemma statement,
    $$|\bsfp_{0,U,n}(u) - \bsfp_{0,U,n}(u')|\leq \sum_{i=-\infty}^{\mantlen+b-1} \frac{2^i}{2^{\mantlen}} + U = 2^b + U\leq 2U.$$
    \end{proof}

    \begin{corollary}[Sensitivity $\Delta_\mathit{ID}\bsfp_{0,U}$]
    \label{cor:rpf-sens-0-unknown-n}
    Suppose we have two datasets $u\simeq_\mathit{ID} u'$ of $(\mantlen,\explen)$-bit floats, with values clipped to the interval $[0,U]$ and with $U\in[2^b,2^{b+1})$, meaning that the maximum possible sum is $\in [2^{\mantlen+b}, 2^{\mantlen+b+1})$. Then, the maximum difference in sums will be 
    $$\Delta_\mathit{ID}\bsfp_{0,U}
    \leq \sum_{i=-\infty}^{\mantlen+b-1} \frac{2^i}{2^{\mantlen}} + U = 2^b + U\leq 2U.$$
        
    \end{corollary}
    
    \begin{proof}
    Let $u\simeq_\mathit{ID} u'$ be two datasets of $(\mantlen,\explen)$-bit floats. Without loss of generality, let $\len(u)<\len(u')$. Let $u''$ be the dataset that results from inserting a 0 immediately before the first location $u_i$ where we have $u_i\neq u'_i$. (Because 0's have no effect on a floating point summation, this insertion will not affect the difference in sums.)
    
     We now note that $u''\simeq_\mathit{Ham} u'$, so we can apply Lemma \ref{lemma:rpf-sens-0-known-n}. Because $n$ is not known, we can pretend that we have $n=\infty$, which causes the ``$\min$'' in the lemma statement to evaluate to $2^{\mantlen+b}$. (We could not actually have $n=\infty$ since computers are finite, but this helps evaluate the worst case value for the ``$\min$'' term.) Using this lemma and by the definition of sensitivity, we have
    $$|\bsfp_{0,U,n}(u'') - \bsfp_{0,U,n}(u')|
    \leq \sum_{i=-\infty}^{\mantlen+b-1} \frac{2^i}{2^{\mantlen}} + d = 2^b + d\leq 2U.$$
    
    Because of how we constructed $u''$, we equivalently have
    $$|\bsfp_{0,U}(u) - \bsfp_{0,U}(u')|
    \leq \sum_{i=-\infty}^{\mantlen+b-1} \frac{2^i}{2^{\mantlen}} + d = 2^b + d\leq 2U.$$
    
    We know that we have $u\simeq_\mathit{ID}u'$, so by the definition of sensitivity, we have 
    $$\Delta_\mathit{ID}\bsfp_{0,U}
    \leq \sum_{i=-\infty}^{\mantlen+b-1} \frac{2^i}{2^{\mantlen}} + d = 2^b + d\leq 2U.$$

    \end{proof}

    \begin{proof}[Proof of Theorem \ref{thrm:ord-toward-zero-known}]
    
    Recall that we are working with the bounded sum function as described in Method~\ref{method:fsplit-sum}. Additionally, note that we prove (1) $\Delta_\mathit{ID}\bsfp_{L,U}$ and (2) $\Delta_\mathit{Ham}\bsfp_{L,U,n}$.
    
    \begin{enumerate}[1.]
        \item
        Let $u\simeq_\mathit{Ham} u'$ be two datasets of length $n$ with $(\mantlen,\explen)$-bit floating point values of type $T$. In the worst case, this means that $u$ and $u'$ are identical, except they differ at some value $u_i\neq u'_i$. 
        
    By Corollary~\ref{cor:rpf-sens-0-unknown-n}, where $U\in[2^b,2^{b+1})$, the calculation of $P$ has a sensitivity of at most $2^b+U$. By generalizing Corollary~\ref{cor:rpf-sens-0-unknown-n} to work with values in $[L,0]$ instead, it follows by symmetry that, where $L\in(-2^{b'+1},-2^{b'}]$, the calculation of $N$ has a sensitivity of at most $2^{b'} + |L|$.
    
    The last step of Method~\ref{method:fsplit-sum} consists of adding $P$ and $N$ using banker's rounding. We need to account for the rounding that can occur in this step. Let $M_P$ be the maximum value of $P$, and let $M_N$ be the minimum value (i.e., maximum magnitude) of $N$. We see, then, that $M_P = 2^{b}\cdot 2^{\mantlen+1}$, and that $M_N = -2^{b'}\cdot 2^{\mantlen+1}$.
    
    Suppose that $M_P\geq |M_N|$. In the worst case, for all values $v<M_P$, we have $\ulp(v) \leq 2^b$, so the worst-case round-down that we can get from combining $P$ and $N$ (using banker's rounding) is $\frac{2^b}{2}$. Likewise, the worst-case round-up that we can get from combining $P$ and $N$ (using banker's rounding) is $\frac{2^b}{2}$. Therefore, the worst-case overall rounding is $2^b$. Likewise, for the $M_P\leq |M_N|$ case, the worst-case rounding is $2^{b'}$.
    
    We recall that the sensitivity for computing $N$ is $\leq 2^{b'} + |L|$; and that the sensitivity for $P$ is $\leq 2^{b} + U$. Therefore, the overall sensitivity initially appears to be 
    $$\Delta_\mathit{ID}\bsfp_{L,U}\leq 2^{b} + U + 2^{b'} + |L|+\max\{2^b,2^{b'} \}.$$
    
    We can improve on this upper bound on the sensitivity, though, with the observation that we have $\did(s,s') = 1$. A single insertion or a single deletion will only affect one of $M_P$ or $M_N$ (whereas the sensitivity above assumes that both $M_P$ and $M_N$ would be affected). In the worst case, the single insertion or deletion would affect the value that causes maximum rounding. Therefore, we get
    $$\Delta_\mathit{ID}\bsfp_{L,U}\leq \max\{2^{b} + U, 2^{b'} + |L|\} +\max\{2^b,2^{b'} \}.$$
    
    For $U\geq |L|$, then, the overall sensitivity is $\leq 3U$, and for $U<|L|$, the overall sensitivity is $\leq 3\cdot|L|$. The overall sensitivity, then, can be upper-bounded by $3\cdot \max\{U,|L|\}$.

    \item 
    Let $u\simeq_\mathit{ID} u'$ be two datasets with $(\mantlen,\explen)$-bit floating-point values of type $T$. Without loss of generality, let $\len(u) < \len(u')$.
    
    By Lemma~\ref{lemma:rpf-sens-0-known-n}, where $U\in[2^c,2^{c+1})$, and where $\min\{n\cdot U, 2^{\mantlen+c}\}\in[2^{\mantlen+b},2^{\mantlen+b+1})$, the calculation of $P$ has a sensitivity of at most $2^b+U$. By generalizing Lemma~\ref{lemma:rpf-sens-0-known-n} to work with values in $[L,0]$ instead, it follows by symmetry that, where $L\in(-2^{c'+1},-2^{c'}]$, and where $\min\{|n\cdot L|, 2^{\mantlen+c'}\}\in(-2^{\mantlen+b'+1},-2^{\mantlen+b'}]$, the calculation of $N$ has a sensitivity of at most $2^{b'} + |L|$.
    
    The last step of Method~\ref{method:fsplit-sum} consists of adding $P$ and $N$ using banker's rounding. We need to account for the rounding that can occur in this step.
    
    If $n\cdot U\geq 2^{\mantlen+c}$, then -- by the reasoning in (1) -- the maximum rounding effect is $2^b$ when $M_P\geq |M_L|$. Likewise, if $|n\cdot L|\geq 2^{\mantlen+c'}$, then -- by the reasoning in (1) -- the maximum rounding effect is $2^{b'}$ when $M_P < |M_L|$.
    
    We now consider the case where $n\cdot U < 2^{\mantlen+c}$, and we also consider the case where $|n\cdot L|\geq 2^{\mantlen+c'}$.
    
    Let $M_P$ be the maximum value of $P$, and let $M_N$ be the minimum value (i.e., maximum magnitude) of $N$. We see, then, that $M_P \leq 2^{c+1}\cdot 2^{\mantlen}$, and that $|M_N| \leq |-2^{c'+1}\cdot 2^{\mantlen}|$.
    
    Suppose that $M_P\geq |M_N|$. In the worst case, $\ulp(M_P) = 2^{c+1}$, so the worst-case round-down that we can get from combining $P$ and $N$ (using banker's rounding) is $\frac{2^{c+1}}{2}$. Likewise, the worst-case round-up that we can get from combining $P$ and $N$ (using banker's rounding) is $\frac{2^{c+1}}{2}$. Therefore, the worst-case overall rounding is $2^{c+1}$. Likewise, for the $M_P\leq |M_N|$ case, the worst-case rounding is $2^{c'+1}$.
    
    We recall that the sensitivity for computing $N$ is $\leq 2^{b'} + |L|$; and that the sensitivity for $P$ is $\leq 2^{b} + U$. Therefore, the overall sensitivity is 
    $$
    \begin{aligned}
    \MoveEqLeft
    \Delta_\mathit{Ham}\bsfp_{L,U}\\
    &\leq 2^{b} + U + 2^{b'} + |L|+\max\{\min\{2^b,2^{c+1}\},\min\{2^{b'}, 2^{c'+1}\}\}.
    \end{aligned}
    $$
    For $U\geq |L|$, then, the overall sensitivity is $\leq 5U$, and for $U<|L|$, the overall sensitivity is $\leq 5L$. The overall sensitivity, then, can be upper-bounded by $5\cdot \max\{U,|L| \}$.

    \end{enumerate}
    
    \end{proof}

    Having proven sensitivities with respect to ordered distance metrics, we now prove sensitivities with respect to unordered distance metrics.
    
    \begin{proof}[Proof of Theorem \ref{thrm:rp-split-summation}]
    
    Recall that we are working with the bounded sum function as described in the theorem statement. Additionally, note that we prove (1) $\Delta_\mathit{Sym}\bsfp_{L,U}$ and (2) $\Delta_\mathit{CO}\bsfp_{L,U,n}$.
    
    \begin{enumerate}[1.]
        \item We first apply RP to the input dataset, and Theorem~\ref{thrm:couple-unknown-n} tells us that RP transforms datasets that are neighboring with respect to $\dsym$ to a coupling of datasets that are neighboring with respect to $\did$. We can now work with datasets that are neighboring with respect to $\did$. Theorem \ref{thrm:ord-toward-zero-known} tells us that 
        $$\Delta_\mathit{ID}\bsfp_{L,U}\leq \max\{2^{b} + U, 2^{b'} + |L|\} +\max\{2^b,2^{b'} \},$$
        so we also have
        $$\Delta_\mathit{Sym}\bsfp_{L,U}\leq \max\{2^{b} + U, 2^{b'} + |L|\} +\max\{2^b,2^{b'} \},$$

        completing the proof.
        
        \item We first apply RP to the input dataset, and Corollary~\ref{cor:coupling-known-n} tells us that RP transforms datasets that are neighboring with respect to $\dco$ to a coupling of datasets that are neighboring with respect to $\dham$. We can now work with datasets that are neighboring with respect to $\dham$. Theorem \ref{thrm:ord-toward-zero-known} tells us that 
        $$\Delta_\mathit{Ham}\bsfp_{L,U}
        \leq 2^{b} + U + 2^{b'} + |L|+\max\{\min\{2^b,2^{c+1}\},\min\{2^{b'}, 2^{c'+1}\}\},$$
        so we also have
        $$\Delta_\mathit{CO}\bsfp_{L,U}
        \leq 2^{b} + U + 2^{b'} + |L|+\max\{\min\{2^b,2^{c+1}\},\min\{2^{b'}, 2^{c'+1}\}\},$$
        completing the proof.
    \end{enumerate}
    \end{proof}

\subsection{Reducing Floats to Integers}\label{sec:floats2ints}

Another attractive solution for floating-point summation is to reduce it to integer summation, since the latter achieves the idealized
sensitivity with simple solutions. The idea behind this method is to cast floats to fixed-point numbers, which we can think of as $k$-bit integers. To achieve this casting, we introduce a discretization parameter $D$ (which is chosen by the data analyst), round each of the dataset elements according to the discretized interval, and apply one of our integer solutions. We can think of this process of ``integerizing'' as a dataset transform.

We first choose a discretization parameter $D$ (we discuss how to make this choice below) and set $M$ as a function of $L$ and $U$ (if $|U|\geq |L|$, we set $M = \max\{L,0\}$; otherwise, we set $M = \min\{U,0\}$) to shift the discretization range to better take advantage of the range of signed integers and obtain accuracy that is better than a no-shift strategy. We consider the function mapping a dataset
$u=[u_1,\ldots,u_n]$ of floats to the signed integer dataset
\begin{equation}\label{eq:float2int}
    \FloatToInt_{L,U,D}(u) = [\round((u_1-M)/D),\ldots,\round((u_n-M)/D)],
\end{equation}
where $\round(\cdot)$ denotes rounding to the nearest integer (we explore various implementations of the $\round(\cdot)$ function in Definitions \ref{defn:round-to-nearest-int} and \ref{defn:rand-round-to-nearest-int} and explore their impacts throughout the accuracy theorems in this section).

\begin{method}[Reducing floats to integers, bounded DP case]
\label{method:red-float-int}

Let $L$ denote the lower bound, $U$ the upper bound, $u = [u_1, \ldots, u_n]$ the input dataset of normal $(k, \ell)$-bit floats, and $D$ the discretization parameter (where $L,U,D$ are $(\mantlen, \explen)$-bit normal floats). We additionally require that $|U|\geq |L|$, though symmetric results hold for $|L| > |U|$.

Let $M = \max\{L,0\}$, and let $K = \lceil (U-M)/D \rceil$ and $J = \lfloor (L-M)/D \rfloor$. Then, the reducing floats to ints method is as follows: 

\begin{enumerate}
    \item Compute $$\bs^{**}_{L,U,D,n}(u) = (\bs^*_{J,K,n}\circ \FloatToInt_{L,U,D,n})(u),$$ where $\bsfp$ is any one of the integer-summation methods discussed in Section \ref{sec:61ints}, \ref{sec:splitsum}, or \ref{sec:rp-sat-int}.\footnote{This method can be adapted to work with the modular summation described in \cref{sec:modularnoise}, too, though the accuracy guarantees may not hold.}
    \item Perform noise addition to obtain the DP mechanism
    \[
        \cM_{L,U,D,n}(u)
         = \bs^{**}_{L,U,D,n}(u) + \Noise((K-J)/\varepsilon),
    \]
    where $\Noise(s)$ denotes a noise distribution with scale $s$ suitable for $k$-bit integers (e.g., the discrete Laplace Mechanism~\cite{GhoshRoSu12}).
    \item Rescale and shift the resulting number to obtain the final floating-point result
    \[
        \cM_{L,U,D,n}(u) \cdot D + M \cdot n.
    \]
\end{enumerate}

Note that this method can be adapted to the setting in which the length of the dataset is not known by using the truncation technique (and associated sensitivities) in \cref{method:trunc}.

\end{method}

\begin{remark}
    The requirement that $|U| \geq |L|$ is not restrictive. Indeed, if $|L| > |U|$, then we apply Method~\ref{method:red-float-int} with $M = \min\{U, 0\}$ instead of $M = \max\{L, 0\}$.
\end{remark}

\begin{theorem}[Sensitivity of $\bs^{**}_{L,U,D,n}$]
\label{thrm:sens-f2i}
Let $\bs^{**}_{L,U,D,n}$ be as defined in \cref{method:red-float-int}. The sensitivity of this function is $\sco \bs^{**}_{L,U,D,n} = \sco \bsfp_{J,K,n} \leq K-J$, where $K = \left\lceil (U-M)/D \right\rceil $ and $J = \left\lfloor (L-M)/D \right\rfloor $.
\end{theorem}

\begin{proof} 

Let $u, v$ be two datasets of length $n$ of normal $(k, \ell)$-bit floats. Then, we claim that
\[
    \dco(\FloatToInt_{L,U,D}(u), \FloatToInt_{L,U,D}(v)) \leq \dco(u, v).
\]
Indeed, since by definition the $\FloatToInt$ function subtracts $M$ from each element and divides the result by $D$, for any elements $u_j, v_k$ such that $u_j = v_k$, it follows that $\round((u_j-M)/D) = \round((v_j-M/D)$. For any pair $u_j, v_k$ such that $u_j \neq v_k$, either $\round((u_j-M)/D) \neq \round((v_k-M)/D)$ or $\round((u_j-M)/D) = \round((v_k-M)/D)$ due to rounding errors. Either way, the minimum number of elements in $\FloatToInt(u)$ that need to be changed to produce $\FloatToInt(v)$ can only be less than the minimum number of elements in $u$ that need to be changed to produce $v$.

Moreover, the smallest value produced by this operation is $\geq J$, and the largest value produced by this operation is $\leq K$, so $J$ and $K$ are correct lower and upper bounds, respectively, on the elements in our transformed dataset. Therefore, the sensitivity is $\sco\bs^{**}_{L,U,D,n} \leq K-J$.

\end{proof}

\begin{corollary}
\label{cor:cm-f2i-dp}
$\cM_{L,U,D,n}$ satisfies $\epsilon$-DP.
\end{corollary}
\begin{proof}
By \cref{thrm:sens-f2i}, the sensitivity of $\bs^{**}_{L,U,D,n}$ is $\sco \bs^{**}_{L,U,D,n} = (K-J)$ so, because scaling our noise distribution to $\sco \bs^{**}_{L,U,D,n} / \varepsilon$ satisfies $\varepsilon$-DP, it follows that $\cM_{L,U,D,n}$ satisfies $\varepsilon$-DP.
\end{proof}

\begin{corollary}
\cref{method:red-float-int} satisfies $\varepsilon$-DP.
\end{corollary}
\begin{proof}
By \cref{cor:cm-f2i-dp}, $\cM_{L,U,D,n}$ satisfies $\varepsilon$-DP. Moreover, that $D,M,n$ are all non-private values. Step 3 of \cref{method:red-float-int}, therefore, just post-processes a result that satisfies $\varepsilon$-DP. Therefore, because DP is robust to post-processing (\cref{prop:post-proc}), Step 3 satisfies $\varepsilon$-DP and \cref{method:red-float-int} satisfies $\varepsilon$-DP.
\end{proof}

In order to analyze the accuracy of $\bs^{**}_{L,U,D,n}$, we need to set the $round(\cdot)$ function to a concrete rounding mode. We consider the following two rounding modes.

\begin{definition}[Round to nearest integer]
\label{defn:round-to-nearest-int}
For a value $x\in\mathbb{R}$, \textit{round to nearest integer} maps $x$ to the integer $x'\in\mathbb{Z}$ such that, for all $y\in\mathbb{Z}$, we have $|x-x'| < |x-y|$.
\end{definition}

\begin{definition}[Randomized round to nearest integer]
\label{defn:rand-round-to-nearest-int}
Let $x\in\mathbb{R} \setminus \mathbb{Z}$, $y\in\mathbb{Z}$ be the largest integer such that $y<x$, and $z\in\mathbb{Z}$ the smallest integer such that $x<z$. \textit{Randomized round to nearest integer} maps $x$ to $y$ with probability $x-y$, and maps $x$ to $z$ otherwise.
\end{definition}

\begin{theorem}[Accuracy of $\bs^{**}_{L,U,D,n}$ under round to nearest]\label{thm:accuracy1}

When the $round(\cdot)$ function in $\FloatToInt$ (Equation~\ref{eq:float2int}) corresponds to round to nearest integer (\cref{defn:round-to-nearest-int}), then
$$|\bs^{**}_{L,U,D,n}(u) - (\bs_{L,U,n}(u)-M\cdot n) / D| \leq n/2.$$
\end{theorem}

\begin{proof}
Since two consecutive integers are at distance 1, each term $u_i$ in the dataset is rounded to a value within 1/2 of the true value $u_i / D$. Because $u$ is of length $n$, the total value of the sum can be rounded to a value within $n/2$ of the true value, which yields the above accuracy. 
\end{proof}

\begin{theorem}[Accuracy of $\bs^{**}_{L,U,D,n}$ under randomized rounding]\label{thm:accuracy2}

When the $round(\cdot)$ function in $\FloatToInt$ (Equation~\ref{eq:float2int}) corresponds to randomized round to nearest integer (\cref{defn:rand-round-to-nearest-int}), then with probability $1-2/e^{2c^2}$ for all $c>0$,
$$|\bs^{**}_{L,U,D,n}(u) - (\bs_{L,U,n}(u)-M\cdot n) / D| \leq c\sqrt{n} = O(\sqrt{n}).$$ 
\end{theorem}

\begin{proof}
We use Hoeffding bounds. Note that, due to the use of randomized rounding, $\bs^{**}_{L,U,D,n}(u)$ is a random variable. We can write $\bs^{**}_{L,U,D,n}(u) = \sum_{i\in[n]} (u_i-M)/D.$ We see that $u_i-M$ is exactly representable as a floating-point number. We also note that, for all $x\in\mathbb{R}$, $\Ex[\round(x/D)] = x/D$. Therefore, $\Ex[\round((u_i-M)/D)] = (u_i-M)/D$. Therefore, $\Ex[\bs^{**}_{L,U,D,n}(u)] = \bs_{L,U,n}(u)-M\cdot n$.

We now apply Hoeffding bounds to $\bs^{**}_{L,U,D,n}(u)$. For all $c\in\mathbb{R}$ where $c>0$, we have
\[
\Pr[|\bs^{**}_{L,U,D,n}(u) - (\bs_{L,U,n}(u)-M\cdot n) / D| \geq c\sqrt{n}]
\leq
2 \exp(-2c^2).
\]

\end{proof}

Next we prove that overflow in Step 2 of \cref{method:red-float-int} is very unlikely. Note that if (floating-point) overflow occurs in Step 3, this means that the user was dealing with a dataset with very large values (i.e., close to the maximum representable floating-point exponent) that are unlikely to result in an accurate answer regardless of the summation strategy.

\begin{theorem}[Overflow is unlikely]
\label{thrm:overflow-unlikely}

Let the discretization parameter $D$ in Method~\ref{method:red-float-int} be
\[
    D=(U-M)\cdot n/(2^{m-2}),
\]
the $\Noise$ term in Step 2 added from the discrete Laplace mechanism, and let $\FloatToInt$ cast to $m$-bit signed integers. Then, the probability of overflow in the computation of $\cM_{L,U,D,n}(u)$ is
\[
    \exp(-\Omega(K\cdot n/(K/\varepsilon))) = \exp(-\Omega(\eps\cdot n)).
\]

\end{theorem}

\begin{proof}

By Step 2 in Method~\ref{method:red-float-int}, we draw noise with scale $(K-J)/\epsilon \leq 2K / \epsilon$. We know that the geometric mechanism has PDF $f(x) \propto \exp(-\epsilon \cdot |x| / 2)$ \cite{BalcerVa18}.

Let $X\sim \Noise((K-J)/\epsilon)$. We are interested in the probability of overflow -- that is, we are interested in the probability that $K\cdot n + X \geq 2^{m-1}$. Using the CDF of the geometric mechanism, we see that
\[
\begin{aligned}
\Pr[|K\cdot n + X| \geq 2^{m-1}]
&= \exp( -\Omega((2^{m-1} - Kn)/((K-J)/\epsilon) )) \\
&= \exp( -\Omega((2^{m-1} - Kn)/(2K/\epsilon) )) \\
&= \exp( -\Omega(Kn/(2K/\epsilon) )) \\
&= \exp( -\Omega(n \epsilon)),
\end{aligned}
\]
with the third equality following from the fact that $2^{m-1} - Kn > Kn$.

\end{proof}

\begin{theorem}

Set the discretization parameter $D$ in Method~\ref{method:red-float-int} to be
\[
    D=(U-M)\cdot n/2^{m-2},
\]
where we operate with normal $(k, \ell)$-bit floats. Then:
\begin{itemize}
\item If the rounding mode in $\FloatToInt$ corresponds to round to nearest,
\[
\begin{aligned}
|(\bs^{**}_{L,U,n}(u)\cdot D + M\cdot n) - \bs_{L,U,n}(u)| \leq \dfrac{D}{n^k} + \dfrac{Mn}{2^{k+1}} &= \dfrac{(U-M)n}{2^{m+k-2}} + \dfrac{Mn}{2^{k+1}} \\
&= O \left( \frac{U\cdot n^2}{2^{k+1}} \right),
\end{aligned}
\]
where all the operations in the first set of parentheses are done with floating-point values using banker's rounding.

\item If the rounding mode in $\FloatToInt$ corresponds to randomized randomized round to nearest,
\[
\begin{aligned}
|(\bs^{**}_{L,U,n}(u)\cdot D) - \bs_{L,U,n}(u)| \leq  O \Big( \dfrac{D \sqrt{n} + Mn}{2^{k+1}} \Big) &= O \Big( \dfrac{(U-M)n^{3/2}}{2^{m+k-1}} + \dfrac{Mn}{2^{k+1}} \Big) \\
&= O\left(\frac{U\cdot n^{3/2}}{2^{k+1}}\right)
\end{aligned}
\]
with high probability, where all the operations in the first set of parentheses are done with floating-point values using banker's rounding.
\end{itemize}
\end{theorem}

\begin{proof}

For the first inequality, by Theorem~\ref{thm:accuracy1} it follows that
\[
    |\bs^{**}_{L,U,D,n}(u) - \bs_{L,U,n}(u) / D| \leq n/2.
\]
To obtain the term $(\bs^{**}_{L,U,D,n}(u) \cdot D + M \cdot n)$, we multiply by $D$ and obtain a floating point number. This yields the $Dn/2^k$ term in the RHS. Similarly, the term $Mn$ requires multiplying $M$ and $n$ together to obtain a floating point number, which yields the $Mn/2^{k+1}$ term in the RHS. Hence by setting $D$ as above the result follows. 

For the second inequality, by Theorem~\ref{thm:accuracy2} it follows that
\[
    |\bs^{**}_{L,U,D,n}(u) - \bs_{L,U,n}(u) / D| \leq \sqrt{n}
\]
with high probability. The same error analysis as for the first inequality yields the RHS expression, and by setting $D$ as above the result follows.

\end{proof}

\begin{remark}[Adapting \cref{method:red-float-int} to the unbounded DP setting]
    In the unbounded DP setting (i.e., when the size of the dataset is not known to the data analyst), the addition of $M\cdot n$ as post-processing in step 3 will not be possible since $n$ is potentially private. To adapt \cref{method:red-float-int} to this setting, the user should use integer summation where overflow is handled with wraparound (i.e., modular summation -- see \cref{sec:modularnoise}) and set $M = 0$ (resulting in no subtraction from elements).
    
    Additionally, the user should choose an upper bound $n_{\max}$ on their best guess at the dataset size and set the discretization parameter $D = (U-M)\cdot n_{\max} / (2^m)$. Then, the accuracy guarantees (e.g., \cref{thm:accuracy1} and \cref{thm:accuracy2}) and overflow guarantees (e.g., \cref{thrm:overflow-unlikely}) will hold as long as $n\leq n_{\max}$; these guarantees break down for $n > n_{\max}$.
    
    Alternatively, the user can use truncation (\cref{method:trunc}) and the associated sensitivity, 
    \[
        \max\{\sco\bsfp_{L,U,n_{\max}, \sid\bsfp_{L,U}} \}.
    \]Then, for $D = (U-M)\cdot n_{\max} / (2^m)$, the accuracy and overflow guarantees will hold since the truncated dataset is guaranteed to have size $\leq n_{\max}$. (Note that the user should still set $M=0$ since the length of the dataset is not known, so adding $M\cdot n_{\max}$ could provide a signficant overestimate on the true result.)
    
\end{remark}

\section*{Acknowledgements}
    We would like to thank Grace Tian for her contributions in the preliminary stages of this work. We also thank the anonymous CCS reviewers for many helpful comments that improved the presentation.
    
    S\'ilvia Casacuberta was supported by the Harvard Program for Research in Science and Engineering (PRISE) and a grant from the Sloan Foundation. Michael Shoemate was supported by a grant from the Sloan Foundation, a grant from the US Census Bureau, and gifts from Apple and Facebook.
    Salil Vadhan was supported by a grant from the Sloan Foundation, gifts from Apple and Facebook, and a Simons Investigator Award.
    Connor Wagaman was supported by the Harvard College Research Program (HCRP); much of this work was completed during his time at Harvard University.
    
    Any opinions, findings, and conclusions or recommendations expressed in this material are those of the authors and do not necessarily reflect the views of our funders.

\printbibliography[heading=bibintoc]

\ifnum\CCSFORMAT=0

\appendix

\section{Missing Proofs from Section~\ref{sec:preliminaries}}\label{appendix:sec2}

\begin{lemma}[Relating Metrics, Lemma \ref{lemma:metric-relate}]
These metrics are related as follows.
\begin{enumerate}
    \item For $u,v\in\Vect(\mathcal{D})$, $$\dsym(u,v) = \min_{\pi\in S_{\len(u)}} \did(\pi(u),v)\leq \did(u,v).$$
    
    \item For $u,v\in \mathcal{D}^n$, $$\dco(u,v) = \min_{\pi\in S_n} \dham(\pi(u),v)\leq \dham(u,v).$$
    
    \item For $u,v\in \mathcal{D}^n$, $$\dsym(u,v) = 2\cdot \dco(u,v).$$
    
    \item For $u,v\in \mathcal{D}^n$, $$\did(u,v) \leq 2\cdot \dham(u,v).$$
\end{enumerate}
\end{lemma}

\begin{proof}
We prove each part below.
\begin{enumerate}
    \item Let $u,v\in\Vect(\mathcal{D})$.
    
    We first show that there exists a permutation $\pi\in S_{\len(u)}$ such that $\did(\pi(u),v)\leq \dsym(u,v)$.
    
    We begin by constructing $v'$ such that $\did(v,v') \leq \dsym(u,v)$ and $h_{v'} = h_u$. We show that we can use $\dsym(u,v)$ insertions and deletions on $v$ to get $v'$ such that $h_{v'} = h_u$. We create $v'$ from $v$ as follows: (1) for all $z\in\mathcal{D}$ such that $h_v(z) > h_{u}(z)$, we delete $h_v(z) - h_u(z)$ copies of $z$ from $v$; (2) for all $z\in\mathcal{D}$ such that $h_v(z) < h_u(z)$, we insert $h_v(z) - h_u(z)$ copies of $z$ to $v$. By \cref{defn:id-dist}, then, 
    \begin{equation}
    \label{eqn:metric-relate-did}
    \did(v',v)\leq \dsym(u,v).
    \end{equation}
    
    The resulting vector $v'$ also has the property that, for all $z\in\mathcal{D}$, $h_{v'}(z) = h_u(z)$, so $h_u = h_{v'}$. By \cref{lemma:hists-equal}, $h_u = h_{v'}$ implies that there exists a permutation $\pi\in S_{\len(u)}$ such that $\pi(u) = v'$. Substituting this into \cref{eqn:metric-relate-did}, we get
    \begin{equation}
    \did(\pi(u),v)\leq \dsym(u,v).
    \end{equation}
    
    We next show that, for all permutations $\pi\in S_{\len(u)}$, we have $\dsym(u,v)\leq \did(\pi(u),v)$, which means that $\dsym(u,v) \leq \min_{\pi\in S_{\len(u)}}\did(\pi(u),v)$. This is true because, for all permutations $\pi\in S_{\len(u)}$, $\dsym(u,v) = \dsym(\pi(u),v)\leq \did(\pi(u),v)$, with the equality following from the fact that $h_{\pi(u)} = h_u$.

    \item Let $u,v\in\mathcal{D}^n$.
    
    We first show that there is a permutation $\pi\in S_n$ such that $\dco(u,v) = \dham(\pi(u),v)$.
    
    We begin by constructing $v'$ such that $\dham(v,v') \leq \dco(u,v)$ and $h_{v'} = h_u$. We create $v'$ from $v$ as follows: (1) for all $z\in\mathcal{D}$ such that $h_v(z) > h_{u}(z)$, replace $h_v(z) - h_u(z)$ copies of $z$ with some value $\diamond\not\in\mathcal{D}$; (2) for all $z\in\mathcal{D}$ such that $h_v(z) < h_u(z)$, replace $h_u(z)-h_v(z)$ $\diamond$ values with $z$. By \cref{defn:ham-dist}, then, 
    \begin{equation}
    \label{eqn:metric-relate-dham}
    \dham(v',v)\leq \dco(u,v).
    \end{equation}
    
    The resulting vector $v'$ also has the property that, for all $z\in\mathcal{D}$, $h_{v'}(z) = h_u(z)$, so $h_u = h_{v'}$. By \cref{lemma:hists-equal}, $h_u = h_{v'}$ implies that there exists a permutation $\pi\in S_n$ such that $\pi(u) = v'$. Substituting this into \cref{eqn:metric-relate-dham}, we get
    \begin{equation}
    \dham(\pi(u),v)\leq \dco(u,v).
    \end{equation}
    
    We next show that, for all permutations $\pi\in S_n$, we have $\dco(u,v)\leq \dham(\pi(u),v)$, which means that $\dco(u,v) \leq \min_{\pi\in S_n}\dham(\pi(u),v)$. This is true because, for every permutation $\pi\in S_n$, $$\dco(u,v) = \dco(\pi(u), v)\leq \dham(\pi(u), v),$$
    with the equality following from the fact that $h_{\pi(u)} = h_u$.
    
    \item Let $u,v\in\mathcal{D}^n$. Then,
    $$\begin{aligned}
    \dsym(u,v) &= \sum_{z\in \mathcal{D}} |h_u(z) - h_v(z)| \\
    &= \sum_{\substack{z\in \mathcal{D} \text{ s.t.}\\ h_u(z) > h_v(z)}} \left( h_u(z) - h_v(z)\right) \\
    &+ \sum_{\substack{z\in \mathcal{D} \text{ s.t.}\\ h_v(z) > h_u(z)}} \left( h_v(z) - h_u(z)\right)\\
    &= \dco(u,v) + \dco(u,v)\\
    &= 2\cdot \dco(u,v).
    \end{aligned}$$

    \item Let $u,v\in \mathcal{D}^n$, and let $\dham(u,v) = c$. Let $\mathcal{I}$ be the set of all indices $i^*$ such that $u_{i^*} \neq v_{i^*}$. We note that, to make it so that $u_{i^*} = v_{i^*}$, we can perform 1 insertion and 1 deletion to change $u_{i^*}$ to $v_{i^*}$. We note that the cardinality of $\mathcal{I}$ is $c$. Therefore, a total of (at most) $2c$ insertions and deletions need to be performed to change $u$ into $u'$ such that $u'=v$. By the definition of $\did$, then, $\did(u,v)\leq 2c = 2\cdot \dham(u,v)$.

\end{enumerate}

\end{proof}

\begin{lemma}[Lemma \ref{lemma:metrics-path-prop}]
The dataset metrics $\dsym, \dco, \dham$, and $\did$ all satisfy the path property.
\end{lemma}
\begin{proof}

By Definitions~\ref{defn:symm-dist}, \ref{defn:co-dist}, \ref{defn:id-dist}, \ref{defn:ham-dist} and the fact that the image of the histogram function is $\mathbb{N}$ it is clear that all of $\dsym, \dco, \dham$, and $\did$ satisfy Condition 1 of the Path Property.

For Condition 2, we show constructively how to build the sequence $u^0, u^1, \ldots, u^d$ for each of the metrics. Let $u, v$ be datasets such that $d(u, v) = d$ for some $d \in \mathbb{Z}$.

\begin{enumerate}
    \item For $\dsym$, let set $S$ be the symmetric difference of $u$ and $v$. Let $\mathcal{I}_u$ be the set of indices $k$ such that $u_k \in S$, and similarly for $\mathcal{J}_v$. Wlog, assume that $|\mathcal{I}_u| \leq |\mathcal{J}_v|$ (otherwise, swap $u$ and $v$ in what follows). To construct the $u^{(i)}$, we iteratively apply the following two steps. To go from $u^{(i)}$ to $u^{(i+2)}$, first pick one of the indices $k \in \mathcal{I}_u$ and delete it from $u^{(i)}$, which yields vector $u^{(i+1)}$. Next, insert into $u^{(i+1)}$ element $v_j$ such that $j \in \mathcal{J}_v$, never repeating either index $k$ or $j$ in the process. After $2|\mathcal{I}_u|$ steps, to continue constructing the $u^{(i)}$ iteratively we insert elements $v_j$ for all the remaining $j \in \mathcal{J}_v$ one at a time (without repeating any index $j$). Hence, after $|\mathcal{J}_v| - |\mathcal{I}_u|$ more steps, $u^{(i)} = v$.
    \item For $\dco$, to construct each $u^{(i)}$ we iteratively change an element of $u$ which is different from $v$ (when viewed as multisets). It follows directly from the definition of $\dco$ that this procedure requires $\dco(u, v) = d$ steps.
    \item For $\dham$, we apply the same procedure, except that $u$ and $v$ are now viewed as ordered vectors.
    \item For $\did$, we apply the same procedure as in the case of $\dco$, where we do one insert or one delete operation at a time. It then follows directly from the definition of $\did$ that this procedure requires $\did(u, v) = d$ steps.
\end{enumerate}

\end{proof}

\begin{lemma}[Convert Sensitivities, Lemma \ref{thm:dsymtoid}]

We can convert between sensitivities in the following ways.
\begin{enumerate}
    \item For every function $f: \Vect(\mathcal{D}) \rightarrow \mathbb{R}$, $\Delta_\mathit{ID} f \leq \Delta_\mathit{Sym} f$.

    \item For every function $f: \mathcal{D}^n \rightarrow \mathbb{R}$, $\Delta_\mathit{Ham} f \leq \Delta_\mathit{CO} f$.
    
    \item For every function $f: \mathcal{D}^n \rightarrow \mathbb{R}$, $\Delta_\mathit{CO} f \leq 2 \Delta_\mathit{Sym}f$.

    \item For every function $f: \mathcal{D}^n \rightarrow \mathbb{R}$, $\Delta_\mathit{Ham} f \leq 2 \Delta_\mathit{ID} f$.

\end{enumerate}

\end{lemma}

\begin{proof} We use Lemma~\ref{lemma:metric-relate} to prove each part. Let each $f$ below be defined as described in its respective part above.
\begin{enumerate}
    \item By Lemma~\ref{lemma:metric-relate}, we see that, for all $u,v\in \Vect(\mathcal{D})$, we have $\dsym(u,v) \leq \did(u,v)$. Therefore, for all $u\nid u'$, we have $u\nsym u'$. By the definition of sensitivity in Definition~\ref{defn:abs-sens}, this means that $\sid f = \max_{u\nid u'} |f(u)-f(u')| \leq \max_{u\nsym u'} |f(u)-f(u')| = \ssym f.$
    
    \item By Lemma~\ref{lemma:metric-relate}, we see that, for all $u,v\in \mathcal{D}^n$, we have $\dco(u,v) \leq \dham(u,v)$. Therefore, for all $u\nham u'$, we have $u\nco u'$. By the definition of sensitivity in Definition~\ref{defn:abs-sens}, this means that $\sham f = \max_{u\nham u'} |f(u)-f(u')| \leq \max_{u\nco u'} |f(u)-f(u')| = \sco f.$
    
    \item We have:
    $$
    \begin{aligned}
    \sco f &= \max_{u\nco u'}|f(u) - f(u')| \\
    &\leq  \max_{u\nco u'} \ssym f \cdot \dsym(u,u') \quad \text{(\cref{thrm:apply-path-prop})} \\
    &= \max_{u\nco u'}\ssym f \cdot 2 \cdot \dco(u,u') \quad \text{(\cref{lemma:metric-relate})}\\
    &= 2 \ssym f,
    \end{aligned}
    $$
    so $\sco f \leq 2 \ssym f.$
    
    \item We have:
    $$
    \begin{aligned}
    \sham f &= \max_{u\nham u'}|f(u) - f(u')| \\
    &\leq \max_{u\nham u'} \sid f \cdot \did(u,u')  \quad \text{(\cref{thrm:apply-path-prop})}\\
    &\leq \max_{u\nham u'} \sid f \cdot 2 \cdot \dham(u,u') \quad \text{(\cref{lemma:metric-relate})}\\
    &= 2 \sid f,
    \end{aligned}
    $$
    so $\sham f \leq 2 \sid f.$

\end{enumerate}

\end{proof}

\section{Missing Proofs from Section~\ref{sec:bs}}\label{appendix:sec4}

\begin{theorem}[Idealized sensitivities of $\bs_{L, U}$ and $\bs_{L, U, n}$, Theorem \ref{thrm:reals-sens}]

The sensitivities of the bounded sum function are the following.
\begin{enumerate}
    \item (Unknown $n$.) $\Delta_\mathit{Sym}\bs_{L,U} = \Delta_\mathit{ID}\bs_{L,U} = \max\{|L|, U\}.$
    \item (Known $n$.) $\Delta_\mathit{CO}\bs_{L,U,n} = \Delta_\mathit{Ham}\bs_{L,U,n} = U-L.$
\end{enumerate}

\end{theorem}

\begin{proof}
Each part is proven below.
\begin{enumerate}
\item 

Let $u,u'\in\Vect(\mathbb{R}_{[L,U]})$ be two datasets such that $u\simeq_\mathit{Sym} u'$. From the formal definition of a histogram in Definition~\ref{defn:histogram}, we observe that, for all $v\in\Vect(\mathbb{R})$, $\bs_{L,U}(v) = \sum_{i=1}^{\len(v)} v_i = \sum_{z\in\mathbb{R}}h_v(z)\cdot z$, where the last sum is well defined because $h_v(z)\neq 0$ for only finitely many values of $z$. Because $u\simeq_\mathit{Sym} u'$, we know that there is at most one value $z^*$ such that $|h_u(z^*) - h_{u'}(z^*)| = 1$, and that, for all $z\neq z^*$, we have $|h_u(z) - h_{u'}(z)| = 0$. 

We can then write the following expressions.
$$
\begin{aligned} 
\MoveEqLeft\left|\bs_{L,U}(u) - \bs_{L,U}(u') \right| \\ &= \left|\sum_{z\in\mathbb{R}}h_u(z)\cdot z - \sum_{z\in\mathbb{R}}h_{u'}(z)\cdot z \right| \\
&= \left|\sum_{z\in\mathbb{R}\setminus z^*}h_u(z)\cdot z - \sum_{z\in\mathbb{R}\setminus z^*}h_{u'}(z)\cdot z \right|\\
&\quad + \left|h_u(z^*)\cdot z^* - h_{u'}(z^*)\cdot z^*\right| \\
&= 0 + \left|z^* \cdot (h_u(z^*) - h_{u'}(z^*))\right|\\
&\leq |z^*|\\
&\leq \max\{|L|,U\},
\end{aligned}
$$
with the final inequality following from the fact that all values are clamped to the interval $[L,U]$, so the largest difference in sums arises when $z^* = \max\{|L|,U\}$.

By the definition of sensitivity, then, $\Delta_\mathit{Sym}\bs_{L,U} \leq \max\{|L|,U\}$. By Theorem~\ref{thm:dcotoham}, we also have $\Delta_{ID} \bs_{L,U} \leq \max\{|L|,U\}$.

For the lower bound, consider the datasets $u=[0], u'=[0,\max\{|L|,U\}]$. We note that $u\simeq_\mathit{ID} u'$. We also note that $|\bs_{L,U}(u) - \bs_{L,U}(u')| = \max\{|L|,U\}$. This means, then, that $\Delta_\mathit{ID}\bs_{L,U} \geq \max\{|L|,U\}$. By the contrapositive of Theorem~\ref{thm:dcotoham}, then, $\Delta_\mathit{Sym}\bs_{L,U} \geq \max\{|L|,U\}$.

Combining these upper and lower bounds on the idealized sensitivity tells us, then, that $$\Delta_\mathit{Sym}\bs_{L,U} = \Delta_\mathit{ID}\bs_{L,U} = \max\{|L|,U\}.$$

\item

Let $u,u'\in\mathbb{R}^n$ be two datasets such that $u\simeq_\mathit{CO} u'$. By Lemma~\ref{lemma:metric-relate}, there is a permutation $\pi\in S_n$ such that $\pi(u)\simeq_\mathit{Ham} u'$. This means there is at most one index $i^*$ such that $\pi(u)_{i^*}\neq u'_{i^*}$. We can then write the following expressions.
$$\begin{aligned}
    \MoveEqLeft|\bsfp_{L,U,n}(u) - \bsfp_{L,U,n}(u')|\\
    &= |\bsfp_{L,U,n}(\pi(u)) - \bsfp_{L,U,n}(u')|\\
    &=\left| \sum_{i=1}^n (\pi(u)_i) - \sum_{i=1}^n (u'_i) \right|\\
    &=\left| (\pi(u)_{i^*}) - (u'_{i^*}) \right|\\
    &\leq U-L.
\end{aligned}$$
with the final inequality following from the fact that all values are clamped to the interval $[L,U]$, so the largest difference in sums arises when, without loss of generality, $\pi(u)_{i^*} = U$ and $u'_{i^*} = L$.

By the definition of sensitivity, then, $\Delta_\mathit{CO}\bs_{L,U,n} \leq U-L$. By Theorem~\ref{thm:dcotoham}, we have $\Delta_{Ham} \bs_{L,U,n} \leq U-L$.

For the lower bound, consider the datasets $u=[L]$ and $u'= [U]$. Then, $u \simeq_{\mathit{Ham}} u'$ and $|\bs_{L,U,n}(u) - \bs_{L,U,n}(u')| = U-L$. By the contrapositive of Theorem~\ref{thm:dcotoham}, it follows that $\Delta_\mathit{CO}\bs_{L,U,n} \geq U-L$.

Combining these upper and lower bounds on the idealized sensitivity, we then conclude that $$\Delta_\mathit{Ham}\bs_{L,U,n} = \Delta_\mathit{CO}\bs_{L,U,n} = U-L.$$

\end{enumerate}

\end{proof}

\section{Missing Proofs from Section~\ref{sec:bounded-sum-computers}}\label{appendix:sec5}

\begin{theorem}[Theorem \ref{prop:rounding-id-floats}]
    Let $\bsfp_{L,U} : \Vect(T_{[L,U]}) \to T$ be iterative bounded sum on the type $T$ of $(\mantlen, \explen)$-bit floats. Let $j,m\in\mathbb{Z}$ satisfy $0<j\leq \mantlen$ and $-(2^{\explen-1}-2) \leq m \leq 2^{\explen-1}-2-j-\mantlen$.
    Then for $L = 0$, $U=2^{\mantlen + m}$, and $n = j\cdot 2^{\mantlen} + 2$, there are datasets $u,v\in\Vect(T)$ where $len(u) = n$ and $\len(v) = n-1$ such that $\did(u,v) = 1$ and
    $$|\bsfp_{L,U}(u) - \bsfp_{L,U}(v)| \geq 2^{\mantlen+j+m}.$$
    In particular,
    $$\frac{\sid \bsfp_{L,U}}{\sid \bs_{L,U}} \geq 2^j.$$
\end{theorem}

\begin{proof}

Let $f(x)$ be defined as $x\mapsto(2^{x} + 2^{x-\mantlen})\cdot 2^{m}$, and consider the datasets
$$    \begin{aligned}
u ={} & [2^{\mantlen+m},2^{\mantlen+m}, f(0)_{1}, \ldots, f(0)_{2^{\mantlen}}, \\
      & f(1)_{2^{\mantlen}+1},\ldots,f(1)_{2\cdot 2^{\mantlen}},\ldots, \\
      & f(j-1)_{(j-1)\cdot 2^{\mantlen}+1},\ldots, f(j-1)_{j\cdot 2^{\mantlen}}],
      \end{aligned}
$$
and $v$ such that $v_i = u_{i+1}$ for all $i$ (so the first term of $u$ is not in $v$). Note that $\did(u,v) = 1$. Additionally, we consider $m=0$ until the conclusion of the proof.

We first evaluate $\bsfp_{L,U}(u)$. We begin by seeing that the sum of the first two terms can be computed exactly: $2^\mantlen \oplus 2^\mantlen = 2^{\mantlen + 1}$.

We next use \cref{lemma:ulp-float-test} to show that the sum of the first three terms cannot be computed exactly: $2^{\mantlen+1} + f(0) = 2^{\mantlen+1} + 1 + 2^{-\mantlen}\in [2^{\mantlen+1}, 2^{\mantlen+2})$, so $\ulp(2^{\mantlen+1} + f(0)) = 2$. However,  $2^{\mantlen+1} + f(0)$ is not a multiple of 2, so by \cref{lemma:ulp-float-test}, $2^{\mantlen+1} + f(0)$ cannot be represented exactly as a $(\mantlen,\explen)$-bit float. This means that we must use banker's rounding (described in \cref{defn:bankers-rounding}) to determine the value of $2^{\mantlen+1} \oplus f(0)$. The value $2^{\mantlen+1} + f(0)$ is between the adjacent $(\mantlen,\explen)$-bit floats $2^{\mantlen+1} + 0$ and $2^{\mantlen+1} + 2$. It is closer to $2^{\mantlen+1} + 2$, so we have $2^{\mantlen+1} \oplus f(0) = 2^{\mantlen+1} + 2$.

Having determined the result of adding $f(0)_1$ to the intermediate sum $2^{k+1}$, we now determine the result of adding the remaining $f(0)$ values. By reasoning similar to the logic used above, we find that for every addition of $f(0)_i$, adding $f(0) = 1 + 2^{-\mantlen}$ has the effect of adding $2$. Therefore, $2^{\mantlen+1}\oplus f(0)_1\oplus \cdots \oplus f(0)_{2^{\mantlen}} = 2^{\mantlen+2}$.

We use \cref{lemma:ulp-float-test} to show that none of the remaining additions are exact, and we show the rules of banker's rounding affect the computation of $\bsfp_{L,U}(u)$. By \cref{lemma:ulp-float-test}, because $\ulp(2^{\mantlen+2}) = 4$ and $2^{\mantlen+2} + f(1)$ is not a multiple of 4, we must use banker's rounding (described in \cref{defn:bankers-rounding}) to determine the value of $2^{\mantlen+2} \oplus f(1)$. Reasoning similar to the logic above shows that adding each of the $f(1) = 2+2^{-\mantlen+1}$ terms has the effect of adding $4$, so $2^{\mantlen+2}\oplus f(1)_{\mantlen+1}\oplus \cdots \oplus f(1)_{2\cdot 2^{\mantlen}} = 2^{\mantlen+3}$. Such logic follows for all of the floating point additions, resulting in an overall sum of $\bsfp_{L,U}(u) = 2^{\mantlen+j+1}$.

We now evaluate $\bsfp_{L,U}(v)$. The general structure of the proof is similar, but the resulting sums are different due to different effects of banker's rounding.

We next use \cref{lemma:ulp-float-test} to show that the sum of the first three terms cannot be computed exactly: $2^{\mantlen} + f(0) = 2^{\mantlen} + 1 + 2^{-\mantlen}\in [2^{\mantlen}, 2^{\mantlen+1})$, so the $\ulp(2^{\mantlen} + f(0)) = 2^0$. However,  $2^{\mantlen+1} + f(0)$ is not a multiple of 1,
% (it is a multiple of $2^{-\mantlen})$
so by \cref{lemma:ulp-float-test}, $2^{\mantlen} + f(0)$ cannot be represented exactly as a $(\mantlen,\explen)$-bit float. This means that we must use banker's rounding (described in \cref{defn:bankers-rounding}) to determine the value of $2^{\mantlen} \oplus f(0)$. The value $2^{\mantlen} + f(0)$ is between the adjacent $(\mantlen,\explen)$-bit floats $2^{\mantlen} + 0$ and $2^{\mantlen} + 1$. It is closer to $2^{\mantlen} + 1$, so we have $2^{\mantlen} \oplus f(0) = 2^{\mantlen} + 1$.

Having determined the result of adding $f(0)_1$ to the intermediate sum $2^{k}$, we now determine the result of adding the remaining $f(0)$ values. By reasoning similar to the logic used above, we find that for every addition of $f(0)_i$, adding $f(0) = 1 + 2^{-\mantlen}$ has the effect of adding $1$ (note in the calculation of $\bsfp_{L,U}(u)$ above that adding $f(0) = 1 + 2^{-\mantlen}$ has the effect of adding $2$ -- this is a critical difference). Therefore, $2^{\mantlen}\oplus f(0)_1\oplus \cdots \oplus f(0)_{2^{\mantlen}} = 2^{\mantlen+1}$.

We now use \cref{lemma:ulp-float-test} to show that none of the remaining additions are exact, and we show the rules of banker's rounding affect the computation of $\bsfp_{L,U}(v)$. By \cref{lemma:ulp-float-test}, because $\ulp(2^{\mantlen+1}) = 2$ and $2^{\mantlen+1} + f(1)$ is not a multiple of 2, we must use banker's rounding (described in \cref{defn:bankers-rounding}) to determine the value of $2^{\mantlen+1} \oplus f(1)$. Reasoning similar to the logic above shows that adding each of the $f(1) = 2+2^{-\mantlen+1}$ terms has the effect of adding $2$, so $2^{\mantlen+1}\oplus f(1)_{\mantlen+1}\oplus \cdots \oplus f(1)_{2\cdot 2^{\mantlen}} = 2^{\mantlen+2}$. Such logic follows for all of the floating point additions, resulting in an overall sum of $\bsfp_{L,U}(u) = 2^{\mantlen+j}$.

We see that $|\bsfp(u)-\bsfp(v)| = 2^{\mantlen+j}$.

From the bounds $-(2^{\explen-1}-2) \leq m \leq 2^{\explen-1}-2-j-\mantlen$, we see that $m$ only affects the exponent of floating point values in the proof and does not affect whether values are representable as $(\mantlen,\explen)$-bit normal floats or the direction of rounding. All values in the proof above, then, can be multiplied by $2^m$, so, for all $m$ such that $-(2^{\explen-1}-2) \leq m \leq 2^{\explen-1}-2-j-\mantlen$, $\sid\bsfp_{L,U} \geq 2^{\mantlen+j+m}$.

We observe that $U = 2^{\mantlen+m}$ and $\sid \bs_{L,U} = \max\{ |L|, U \} = U$. Therefore, $$\frac{\sid \bsfp_{L,U}}{\sid \bs_{L,U}} \geq 2^j.$$

\end{proof}

\begin{theorem}[Theorem \ref{thrm:quadratic-round-attack}]

Let $\bsfp_{L,U}:\Vect(T_{[L,U]})\to T$ be iterative bounded sum on the type $T$ of $(\mantlen, \explen)$-bit floats. Additionally, let $a,j\in\mathbb{Z}$ satisfy $2\leq j < k$, $-(2^{\explen-1}-2)\leq a$, $-(2^{\explen-1}-2) + 1 + \mantlen\leq j+a\leq 2^{\explen-1}-1$ (these conditions ensure that values are representable as $(\mantlen, \explen)$-bit floats), $n = 2^j$, $m = n/2$. Then, for
$$U = 2^a, L  = -\left(\frac{U\cdot m}{2^\mantlen}\right)\cdot \left(\frac{1}{2} - \frac{1}{2^\mantlen}\right),$$
there are datasets $u,v\in\Vect(T_{[L,U]})$ with $\did(u,v) = 1$ such that
$$|\bsfp_{L,U}(u) - \bsfp_{L,U}(v)| \geq \frac{n^2}{2^{\mantlen + 3}} \cdot U + U.$$
In particular,
$$\frac{\sid \bsfp_{L,U}}{\sid \bs_{L,U}} = \frac{n^2}{2^{\mantlen+3}} + 1.$$
\end{theorem}

\begin{proof}

Let $$x = \left(\frac{U\cdot m}{2^\mantlen}\right)\cdot \left(\frac{1}{2} + \frac{1}{2^\mantlen}\right),$$ and consider the datasets
$$u = [U_1,\ldots, U_{m-1}, U_{m}, x_1, L_2, x_3, L_4,\ldots, x_{m-1}, L_m]$$
and
$$v = [U_1,\ldots, U_{m-1}, x_1, L_2, x_3, L_4,\ldots, x_{m-1}, L_m].$$
where, for all $i$, $L_i = L$ and $U_i = U$. 
Note that $u$ and $v$ are equivalent with the exception that $v$ contains $m-1$ copies of $U$ rather than $m$ copies of $U$, so $\did(u,v) = 1$.

We first show that $\bsfp_{L,U}[U_1,\ldots,U_{m-1}] = (m-1)\cdot U$ and $\bsfp_{L,U}[U_1,\ldots,U_{m}] = m\cdot U$. We take the approach of showing that all intermediate sums computed in the calculation of $\bsfp_{L,U}[U_1,\ldots, U_m]$ can be represented exactly as $(\mantlen,\explen)$-bit floats, which is done by first showing that all of these intermediate sums are multiples of their $\ulp$ and then applying \cref{lemma:ulp-float-test}.

We observe that $\{U,\ldots, m\cdot U \}$ is the set of intermediate sums that result from calculating $\bs_{L,U}[U_1,\ldots, U_m]$. Let $\UU$ denote this set $\{U,\ldots, m\cdot U\}$. We note that $m\cdot U = 2^{a+j-1}$, so $\ulp(m\cdot U) = 2^{a+j-1-\mantlen}$. For all $x\in \UU$, then, $$\ulp(x)\leq \ulp(m\cdot U) = 2^{a+j-1-\mantlen}.$$
Note that $j<k$, so $2^{a+j-1-\mantlen}$ necessarily divides $U=2^a$. All such $x$ are multiples of $U=2^a$, so they are necessarily multiples of $\ulp(x)$. By \cref{lemma:ulp-float-test}, all $x\in \UU$ can be represented exactly as $(\mantlen,\explen)$-bit floats, so by \cref{lemma:floats-exact-add}, $\bsfp_{L,U}[U_1,\ldots,U_{m-1}] = (m-1)\cdot U$ and $\bsfp_{L,U}[U_1,\ldots,U_{m}] = m\cdot U$.

We now evaluate $\bsfp_{L,U,n}(u)$. We use \cref{lemma:ulp-float-test} to show that $\bs_{L,U,n}(u)=m\cdot U + x$ is not exactly representable as a $(\mantlen,\explen)$-bit float, and we show how to calculate the result $\bsfp_{L,U,n}(u)$.

As shown above, the sum of the first $m$ terms is $$\bsfp_{L,U}[U_1,\ldots, U_m] = m\cdot U.$$ Thus, $\bsfp_{L,U}[U_1,\ldots, U_m,x] = m\cdot U \oplus x$. We note that
$$m\cdot U + x \in [m\cdot U, m\cdot U + \ulp(m\cdot U)].$$ Because $m\cdot U + x$ is closer to $m\cdot U + \ulp(m\cdot U)$, $m\cdot U \oplus x = m\cdot U + \ulp(m\cdot U)$. We now consider the addition of $L$, $m\cdot U + \ulp(m\cdot U) \oplus L$. We see that $$m\cdot U + \ulp(m\cdot U) + L \in [m\cdot U, m\cdot U + \ulp(m\cdot U)],$$ and $m\cdot U + \ulp(m\cdot U) + L$ is closer to $m\cdot U + \ulp(m\cdot U)$, so $m\cdot U + \ulp(m\cdot U) \oplus L = m\cdot U + \ulp(m\cdot U) \oplus L$. Similar logic applies for all $x,L$, in which adding $L$ has the effect of adding $0$ and adding $x$ has the effect of adding $\ulp(m\cdot U)$ for a total sum of $\bsfp_{L,U}(u) = m\cdot U + \frac{m^2 \cdot U}{2^{\mantlen+1}}$.

We now evaluate $\bsfp_{L,U,n}(v)$. As shown above, the sum of the first $m-1$ terms corresponds to $\bsfp_{L,U}[U_1,\ldots, U_{m-1}] = (m-1)\cdot U$. Thus, $\bsfp_{L,U}[U_1,\ldots, U_{m-1},x] = (m-1)\cdot U \oplus x$. We note that
$$(m-1)\cdot U + x \in [(m-1)\cdot U, (m-1)\cdot U + \ulp((m-1)\cdot U)].$$ Because $(m-1)\cdot U + x$ is closer to $(m-1)\cdot U + \ulp((m-1)\cdot U)$, $(m-1)\cdot U \oplus x = (m-1)\cdot U + \ulp((m-1)\cdot U)$. We now consider the addition of $L$, $(m-1)\cdot U + \ulp((m-1)\cdot U) \oplus L$.
We see that
$$\begin{aligned}
\MoveEqLeft(m-1)\cdot U + \ulp((m-1)\cdot U) + L\\
&\in [(m-1)\cdot U, (m-1)\cdot U + \ulp((m-1)\cdot U)],
\end{aligned}$$
and $(m-1)\cdot U + \ulp((m-1)\cdot U) + L$ is closer to $(m-1)\cdot U$, so $(m-1)\cdot U + \ulp((m-1)\cdot U) \oplus L = (m-1)\cdot U$. Similar logic applies for all $x,L$, in which adding $x$ and $L$ in succession has the effect of adding $0$ for a total sum of $\bsfp_{L,U}(v) = (m-1)\cdot U$.

Therefore, $$|\bsfp_{L,U}(u) - \bsfp_{L,U}(v)| = \frac{m^2\cdot U}{2^{\mantlen + 1}} + U = \frac{n^2\cdot U}{2^{\mantlen + 3}} + U,$$ so, because $\did(u,v) = 1$, $\did\bsfp_{L,U} = \frac{m^2\cdot U}{2^{\mantlen + 1}} + U$. We observe that $\sid\bs_{L,U} = \max\{|L|,U\} = U$, so $$\frac{\sid \bsfp_{L,U}}{\sid \bs_{L,U}} = \frac{n^2}{2^{\mantlen+3}} + 1.$$

\end{proof}

\fi

\end{document}